\documentclass[conference]{IEEEtran}

\usepackage{amsmath} 
\usepackage{amsthm} 
\usepackage{amssymb} 
\usepackage{amsfonts} 
\usepackage{array}
\usepackage{tcolorbox}
\usepackage{comment}
\usepackage{lineno}
\usepackage{cite}
\usepackage [english]{babel}
\usepackage [autostyle, english = american]{csquotes}
\MakeOuterQuote{"}

\usepackage{epstopdf}
\usepackage{listings}

\usepackage{fancyhdr}
\usepackage{graphicx}

\usepackage{url}
\usepackage[colorinlistoftodos]{todonotes}
\usepackage{subcaption}
\usepackage{soul}

\usepackage{mathtools}
\usepackage{tabularx,environ}
\usepackage{xspace}
\usepackage[linesnumbered,ruled]{algorithm2e}
\usepackage{multirow}

\usepackage{newfloat}
\usepackage{caption}

\DeclareFloatingEnvironment[fileext=frm,placement={!ht},name=Functionality]{idealFunctionalityFloat}
\newenvironment{functionalityEnv}[1]
{
	\floatstyle{boxed}
	\restylefloat{idealFunctionalityFloat}
	
	\begin{idealFunctionalityFloat}[tb]\footnotesize
	}{
 \end{idealFunctionalityFloat}

}

\DeclareFloatingEnvironment[fileext=frm,placement={!ht},name=Protocol]{protocolFloat}
\newenvironment{protocolEnv}[1]
{
	\floatstyle{boxed}
	\restylefloat{protocolFloat}

	\begin{protocolFloat}[tb]\footnotesize
	}{
	\end{protocolFloat}
	
}

\DeclareFloatingEnvironment[fileext=frm,placement={!ht},name=Simulator]{simulatorFloat}
\newenvironment{simulatorEnv}[1]
{
	\floatstyle{boxed}
	\restylefloat{simulatorFloat}
	\begin{simulatorFloat}[tb]\footnotesize
	}{
	\end{simulatorFloat}
}

\usepackage[colorlinks=true]{hyperref}

\usepackage{float}
\newfloat{functionalityFig}{hbtp}{lop}[section]

\floatstyle{boxed}
\restylefloat{functionalityFig}

\def\headline#1{\hbox to \hsize{\hrulefill\quad\lower.3em\hbox{\textbf{#1}}\quad\hrulefill}}

\newcounter{todocounter}

\SetAlFnt{\scriptsize} 

\SetKwData{Up}{up}
\SetKwFunction{Union}{Union}
\SetKwFunction{FindCompress}{FindCompress}
\SetKwInOut{Input}{input}
\SetKwInOut{Output}{output}
\SetKwInOut{Parameters}{Parameters}
\SetKwInOut{Init}{Init}

\SetCommentSty{textup}

\usepackage{etoolbox}

\newcounter{algoline}
\newcommand\Numberline{\refstepcounter{algoline}\nlset{\thealgoline}}
\AtBeginEnvironment{algorithm}{\setcounter{algoline}{0}}

\let\oldnl\nl
\newcommand{\nonl}{\renewcommand{\nl}{\let\nl\oldnl}}

\SetCommentSty{mycommfont}

\newcommand{\predicateTextSize}{\scriptsize}

\newcommand{\verifySigName}{\ensuremath{\textsf{vSig}}}
\newcommand{\verifyPreImageName}{\ensuremath{\textsf{vPreImg}}}
\newcommand{\verifyTimeoutName}{\ensuremath{\textsf{vTime}}}

\newcommand{\verifySig}[2]{\ensuremath{\verifySigName\left(#1;#2\right)}}
\newcommand{\verifyPreImage}[2]{\ensuremath{\verifyPreImageName\left(#1;#2\right)}}
\newcommand{\verifyTimeout}[1]{\ensuremath{\verifyTimeoutName\left(#1\right)}}

\SetKwData{return}{return}
\SetKwData{true}{true}
\SetKwData{false}{false}

\newtheorem{theorem}{Theorem}
\newtheorem{lemma}{Lemma}
\newtheorem*{lemmaNo}{Lemma}

\newtheorem{corollary}{Corollary}

\newcommand{\negspace}{\vspace{-0.5\baselineskip}}

\newcommand{\timeout}{\ensuremath{T}}

\newcommand{\blockCountSpecificVal}{\ensuremath{k}}
\newcommand{\blockCountAnotherSpecificVal}{\ensuremath{j}}

\newcommand{\blockIndex}[1]{\ensuremath{b_{#1}}}

\newcommand{\htlc}{\textsl{HTLC}}
\newcommand{\madhtlc}{\textsl{MAD-HTLC}}
\newcommand{\htlcSpec}{\textsl{HTLC-Spec}}

\newcommand{\collateralContract}{\textsl{MH-Col}}
\newcommand{\depositContract}{\textsl{MH-Dep}}

\newcommand{\predicateName}{predicate{}}

\newcommand{\predicateNamePlural}{predicates{}}
\newcommand{\predicateNamePluralCap}{Predicates{}}

\newcommand{\contractName}{contract{}}

\newcommand{\contractNamePlural}{contracts{}}

\newcommand{\numberOfMiners}{\ensuremath{n}\xspace}

\newcommand{\minersOnBoardProbability}[1]{\ensuremath{\prob{u}^{#1}}\xspace}

\newcommand{\prob}[1]{\ensuremath{\lambda_{#1}}\xspace}
\newcommand{\probMin}{\ensuremath{\lambda_{\text{min}}}\xspace}

\newcommand{\txFee}{\ensuremath{f}}
\newcommand{\contractTokenLetter}{\ensuremath{v}}

\newcommand{\collateralTransactionName}{\ensuremath{\text{col}}}
\newcommand{\depositTransactionName}{\ensuremath{\text{dep}}}
\newcommand{\bothContractsTransactionName}{\ensuremath{\text{dep+col}}}

\newcommand{\depositTokens}{\ensuremath{\contractTokenLetter^\text{\depositTransactionName}}\xspace}
\newcommand{\collateralTokens}{\ensuremath{\contractTokenLetter^\text{\collateralTransactionName}}\xspace}

\newcommand{\trueConst}{\ensuremath{\textsf{red}}\xspace}
\newcommand{\falseConst}{\ensuremath{\textsf{irred}}\xspace}

\newcommand{\sig}[1]{\textsl{sig}}

\newcommand{\hashInputVar}[1]{\ensuremath{\textsl{pre}_{#1}}\xspace}
\newcommand{\hashOutputVar}[1]{\ensuremath{\textsl{dig}_{#1}}\xspace}

\newcommand{\generalSecret}{\ensuremath{\hashInputVar{}}\xspace}

\newcommand{\aliceSecret}{\ensuremath{\hashInputVar{a}}\xspace}
\newcommand{\bobSecret}{\ensuremath{\hashInputVar{b}}\xspace}
\newcommand{\contractSecretAlice}{\ensuremath{\hashOutputVar{a}}\xspace}
\newcommand{\contractSecretBob}{\ensuremath{\hashOutputVar{b}}\xspace}

\newcommand{\aliceSig}{\ensuremath{\sig{}_{a}}\xspace}
\newcommand{\bobSig}{\ensuremath{\sig{}_{b}}\xspace}

\newcommand{\initTxName}{\ensuremath{\textit{tx}_\text{init}}}

\newcommand{\simpleTxName}{\ensuremath{\textit{tx}}}

\newcommand{\transactionOne}{\ensuremath{\simpleTxName}}
\newcommand{\transactionTwo}{\ensuremath{\textit{tx'}}}

\newcommand{\aliceTransaction}{\ensuremath{\simpleTxName_a}}
\newcommand{\bobTransaction}{\ensuremath{\simpleTxName_b}}
\newcommand{\minerTransaction}{\ensuremath{\simpleTxName_m}}

\newcommand{\secParameter}{\ensuremath{\mu}}

\newcommand{\rangeSecParameter}{\ensuremath{\left\{0,1\right\}^\secParameter}}
\newcommand{\rangeGeneral}{\ensuremath{\left\{0,1\right\}^*}}

\newcommand{\idealFunctionalityHash}{\ensuremath{\mathcal{H}}}
\newcommand{\idealFunctionalityMempool}{\ensuremath{\mathcal{G}_\text{mbp}}}

\newcommand{\getsRandom}{\ensuremath{\overset{R}{\gets}}}

\newcommand{\idealFunctionalityMadhtlc}{\ensuremath{\mathcal{F}_\text{rmh}}}

\newcommand{\protocolForUC}{\ensuremath{\Pi_\text{rmh}}}
\newcommand{\protocolForBlockchain}{\ensuremath{\Pi_\text{mad-htlc}}}

\newcommand{\simulator}{\ensuremath{\textsf{Sim}}}
\newcommand{\adversary}{\ensuremath{\textsf{Adv}}}

\newcommand{\funcEnv}{\ensuremath{\mathcal{Z}}}
\newcommand{\userA}{\ensuremath{\mathcal{A}}\xspace}
\newcommand{\userB}{\ensuremath{\mathcal{B}}\xspace}
\newcommand{\userO}{\ensuremath{\mathcal{M}}\xspace}
\newcommand{\anyParty}{\ensuremath{\mathcal{P}}}

\newcommand{\invokeFunctionality}[2]{\ensuremath{{#1}\left(#2\right)}}
\newcommand{\id}{\ensuremath{\textit{sid}}}
\newcommand{\stateSetup}{\ensuremath{\textsf{setup}}}
\newcommand{\stateInit}{\ensuremath{\textsf{initiation}}}
\newcommand{\stateRedeeming}{\ensuremath{\textsf{redeeming}}}

\newcommand{\hashPrivatehSet}{\ensuremath{\textit{H}_\text{qr}}}
\newcommand{\hashVarQuery}{\ensuremath{\textit{q}}}
\newcommand{\hashVarResponse}{\ensuremath{\textit{r}}}

\newcommand{\biName}{\text{mbp}}

\newcommand{\biInternalPreimages}[1]{\ensuremath{\textit{pub}_{#1}^{\biName}}}
\newcommand{\biInternalImages}[1]{\ensuremath{\textit{dig}_{#1}^{\biName}}}

\newcommand{\biResult}{\ensuremath{\textit{res}^{\biName}}}
\newcommand{\biInit}{\ensuremath{\textit{init}^{\biName}}}
\newcommand{\biPublished}{\ensuremath{\textit{publish}^{\biName}}}
\newcommand{\biTx}{\ensuremath{\textit{tx}^{\biName}}}

\newcommand{\publishingFunctionName}{\ensuremath{\textit{rPredicate}}}
\newcommand{\publishingFunctionHashMatches}[1]{\ensuremath{\textit{h}_{#1}}}
\newcommand{\redeemPathIndex}{\ensuremath{\textit{path}}}

\newcommand{\protocolName}{\ensuremath{\text{prot}}}
\newcommand{\protocolState}[1]{\ensuremath{\textit{s}_{#1}^{\protocolName}}}
\newcommand{\protocolPreimagePassedValue}[1]{\ensuremath{\textit{p}_{#1}^{\protocolName}}}

\newcommand{\protocolReceivedPublishedParams}{\ensuremath{\textit{received}^{\protocolName}}}

\newcommand{\mhFName}{\text{mh}}

\newcommand{\mhFInternalPreimages}[1]{\ensuremath{\textit{pub}_{#1}^{\mhFName}}}
\newcommand{\mhFResult}{\ensuremath{\textit{res}^{\mhFName}}}
\newcommand{\mhFSentPreimageA}{\ensuremath{\textit{shared}^{\mhFName}}}
\newcommand{\mhFPublished}{\ensuremath{\textit{published}^{\mhFName}}}
\newcommand{\mhFInit}{\ensuremath{\textit{init}^{\mhFName}}}
\newcommand{\mhFSetupA}{\ensuremath{\textit{setup}_{a}^{\mhFName}}}
\newcommand{\mhFSetupB}{\ensuremath{\textit{setup}_{b}^{\mhFName}}}

\newcommand{\propertyInternalPreimages}[1]{\ensuremath{\textsl{pub}_{#1}}}
\newcommand{\propertySentPreimageA}{\ensuremath{\textsl{shared}}}

\newcommand{\simName}{\text{sim}}
\newcommand{\simSentPreimageA}{\ensuremath{\textit{shared}^{\simName}}}

\newcommand{\simInit}{\ensuremath{\textit{init}^{\simName}}}
\newcommand{\simPublished}{\ensuremath{\textit{published}^{\simName}}}
\newcommand{\simSetupA}{\ensuremath{\textit{setup}_{a}^{\simName}}}
\newcommand{\simSetupB}{\ensuremath{\textit{setup}_{b}^{\simName}}}

\newcommand{\simPreimagePassedValue}[1]{\ensuremath{\textit{p}_{#1}^{\simName}}}

\newcommand{\simAliceSecret}{\ensuremath{\textsl{pre}_{a}^{\simName}}}
\newcommand{\simBobSecret}{\ensuremath{\textsl{pre}_{b}^{\simName}}}
\newcommand{\simContractSecretAlice}{\ensuremath{\textsl{dig}_{a}^{\simName}}}
\newcommand{\simContractSecretBob}{\ensuremath{\textsl{dig}_{b}^{\simName}}}
\newcommand{\simTx}{\ensuremath{\textsl{tx}^{\simName}}}

\newcommand{\mhSetupA}{\ensuremath{\textsf{setup-A}}}
\newcommand{\mhSetupB}{\ensuremath{\textsf{setup-B}}}

\newcommand{\mhInit}{\ensuremath{\textsf{init}}}

\newcommand{\mhTransaction}{\ensuremath{\textsf{redeem}}}
\newcommand{\mhSharePreimage}{\ensuremath{\textsf{share}}}
\newcommand{\mhPublish}{\ensuremath{\textsf{publish}}}
\newcommand{\mhUpdate}{\ensuremath{\textsf{update}}}

\newcommand{\mhRedeemPathOne}{\ensuremath{\textsl{dep-}\mathcal{A}}\xspace}
\newcommand{\mhRedeemPathTwo}{\ensuremath{\textsl{dep-}\mathcal{B}}\xspace}
\newcommand{\mhRedeemPathThree}{\ensuremath{\textsl{dep-}\mathcal{M}}\xspace}
\newcommand{\mhRedeemPathFour}{\ensuremath{\textsl{col-}\mathcal{B}}\xspace}
\newcommand{\mhRedeemPathFive}{\ensuremath{\textsl{col-}\mathcal{M}}\xspace}

\newcommand{\htlcRedeemPathOne}{\ensuremath{\textsl{htlc-}\mathcal{A}}}
\newcommand{\htlcRedeemPathTwo}{\ensuremath{\textsl{htlc-}\mathcal{B}}}

\newcommand{\htlcTransactionName}{\ensuremath{\text{h}}}
\newcommand{\combineTransactionWithName}[2]{\ensuremath{#1^{#2}}}
\newcommand{\aliceTransactionMadAliceBob}{\ensuremath{\combineTransactionWithName{\aliceTransaction}{\depositTransactionName}}}
\newcommand{\bobTransactionMadBoth}{\ensuremath{\combineTransactionWithName{\bobTransaction}{\bothContractsTransactionName}}}
\newcommand{\bobTransactionMadDeposit}{\ensuremath{\combineTransactionWithName{\bobTransaction}{\collateralTransactionName}}}
\newcommand{\bobTransactionMadAliceBob}{\ensuremath{\combineTransactionWithName{\bobTransaction}{\depositTransactionName}}}
\newcommand{\aliceTransactionHTLC}{\ensuremath{\combineTransactionWithName{\aliceTransaction}{\htlcTransactionName}}}
\newcommand{\bobTransactionHTLC}{\ensuremath{\combineTransactionWithName{\bobTransaction}{\htlcTransactionName}}}
\newcommand{\minerTransactionMadBoth}{\ensuremath{\combineTransactionWithName{\minerTransaction}{\bothContractsTransactionName}}}
\newcommand{\minerTransactionMadDeposit}{\ensuremath{\combineTransactionWithName{\minerTransaction}{\collateralTransactionName}}}
\newcommand{\minerTransactionMadAliceBob}{\ensuremath{\combineTransactionWithName{\minerTransaction}{\depositTransactionName}}}

\newcommand{\aliceTxFee}{\ensuremath{\txFee{}_a}}
\newcommand{\bobTxFee}{\ensuremath{\txFee{}_b}}

\newcommand{\aliceFeeMadAliceBob}{\ensuremath{\combineTransactionWithName{\aliceTxFee}{\depositTransactionName}}}
\newcommand{\bobFeeMadBoth}{\ensuremath{\combineTransactionWithName{\bobTxFee}{\bothContractsTransactionName}}}
\newcommand{\bobFeeMadDeposit}{\ensuremath{\combineTransactionWithName{\bobTxFee}{\collateralTransactionName}}}
\newcommand{\bobFeeMadAliceBob}{\ensuremath{\combineTransactionWithName{\bobTxFee}{\depositTransactionName}}}

\newcommand{\aliceFeeHTLC}{\ensuremath{\combineTransactionWithName{\aliceTxFee}{\htlcTransactionName}}}
\newcommand{\bobFeeHTLC}{\ensuremath{\combineTransactionWithName{\bobTxFee}{\htlcTransactionName}}}

\newcommand{\publicKey}{\ensuremath{\textsl{pk}}}
\newcommand{\alicePublicKey}{\ensuremath{\publicKey_a}}
\newcommand{\bobPublicKey}{\ensuremath{\publicKey_b}}

\newcommand{\secretKey}{\ensuremath{\textit{sk}}}
\newcommand{\aliceSecretKey}{\ensuremath{\secretKey_a}}
\newcommand{\bobSecretKey}{\ensuremath{\secretKey_b}}

\newcommand{\hashFunction}{\ensuremath{H}}
\newcommand{\hashFunctionWithInput}[1]{\ensuremath{\hashFunction\left(#1\right)}}

\newcommand{\strategyOfEntity}{\ensuremath{\sigma}}
\newcommand{\strategyProfile}{\ensuremath{\bar{\strategyOfEntity}}}
\newcommand{\strategyAllMiners}{\ensuremath{\bar{\strategyOfEntity{}}}}

\newcommand{\utilityOfEntity}[3]{\ensuremath{u_{#1}\left(#2,#3\right)\xspace}}

\newcommand{\gameDefNameGeneric}{\ensuremath{\Gamma{}}}
\newcommand{\gameDefNameHTLC}{\ensuremath{{\gameDefNameGeneric}^{\textit{H}}}}
\newcommand{\gameDefNameDHTLC}{\ensuremath{\gameDefNameGeneric}^{\textit{MH}}}

\newcommand{\contractState}{\ensuremath{s}}
\newcommand{\subgameIndistinguishable}{\ensuremath{\cdot}}

\newcommand{\gameDefHTLC}[2]{\ensuremath{\gameDefNameHTLC{}\left({#1},{#2}\right)}}
\newcommand{\gameDefDHTLC}[2]{\ensuremath{\gameDefNameDHTLC{}\left({#1},{#2}\right)}}

\newcommand{\depositPredicateMathSpace}{\ensuremath{\; }}

\newcommand{\collateralPredicateMathSpace}{\ensuremath{\;  \; \; \;}}
\newcommand{\htlcPredicateMathSpace}{\ensuremath{\;  \; \; \;}}

\newcommand{\OPHASH}{\ensuremath{\texttt{OP\_HASH160}}}
\newcommand{\OPEQUAL}{\ensuremath{\texttt{OP\_EQUAL}}}
\newcommand{\OPIF}{\ensuremath{\texttt{OP\_IF}}}
\newcommand{\OPELSE}{\ensuremath{\texttt{OP\_ELSE}}}
\newcommand{\OPCHECKSEQUENCEVERIFY}{\ensuremath{\texttt{OP\_CHECKSEQUENCEVERIFY}}}
\newcommand{\OPDROP}{\ensuremath{\texttt{OP\_DROP}}}
\newcommand{\OPENDIF}{\ensuremath{\texttt{OP\_ENDIF}}}
\newcommand{\OPCHECKSIG}{\ensuremath{\texttt{OP\_CHECKSIG}}}
\newcommand{\OPSWAP}{\ensuremath{\texttt{OP\_SWAP}}}
\newcommand{\OPVERIFY}{\ensuremath{\texttt{OP\_VERIFY}}}
\newcommand{\OPONE}{\ensuremath{\texttt{OP\_1}}}

\newcommand{\OPZERO}{\ensuremath{\texttt{OP\_0}}}

\begin{document}
	\title{\huge \emph{MAD-HTLC}: Because \htlc{} is Crazy-Cheap to Attack (extended version)}

\sloppy


\author{\IEEEauthorblockN{Itay Tsabary}
	\IEEEauthorblockA{\textit{Technion, IC3} \\
		\textit{sitay@campus.technion.ac.il}}
	\and
	\IEEEauthorblockN{Matan Yechieli}
	\IEEEauthorblockA{\textit{Technion, IC3} \\
		\textit{matany@campus.technion.ac.il}}
	\and
	\IEEEauthorblockN{Alex Manuskin}
	\IEEEauthorblockA{\textit{ZenGo-X} \\
		\textit{alex@manuskin.org}}
	\and
	\IEEEauthorblockN{Ittay Eyal}
	\IEEEauthorblockA{\textit{Technion, IC3} \\
		\textit{ittay@technion.ac.il}}
}



\maketitle

\begin{abstract}
\emph{Smart Contracts} and \emph{transactions} allow users to implement elaborate constructions on \emph{cryptocurrency} \emph{blockchains} like Bitcoin and Ethereum. 
Many of these constructions, including operational payment channels and atomic swaps, use a building block called \emph{Hashed Time-Locked Contract} (\htlc{}).

In this work, we distill from \htlc{} a specification (\htlcSpec{}), and present an implementation called \emph{Mutual-Assured-Destruction Hashed Time-Locked Contract} (\madhtlc{}). 
\madhtlc{} employs a novel approach of utilizing the existing blockchain operators, called \emph{miners}, as part of the design.
If a user misbehaves, \madhtlc{} incentivizes the miners to confiscate all her funds. 
We prove \madhtlc{}'s security using the UC framework and game-theoretic analysis.
We demonstrate \madhtlc{}'s efficacy and analyze its overhead by instantiating it on Bitcoin's and Ethereum's operational blockchains. 

Notably, current miner software makes only little effort to optimize revenue, since the advantage is relatively small. 
However, as the demand grows and other revenue components shrink, miners are more motivated to fully optimize their fund intake. 
By patching the standard Bitcoin client, we demonstrate such optimization is easy to implement, making the miners natural enforcers of \madhtlc{}. 

Finally, we extend previous results regarding \htlc{} vulnerability to \emph{bribery attacks}. 
An attacker can incentivize miners to prefer her transactions by offering high \emph{transaction fees}. 
We demonstrate this attack can be easily implemented by patching the Bitcoin client, and use game-theoretic tools to qualitatively tighten the known cost bound of such bribery attacks in presence of rational miners. 
We identify bribe opportunities occurring on the Bitcoin and Ethereum main networks where a few dollars bribe could yield tens of thousands of dollars in reward (e.g., \$2 for over \$25K).

\end{abstract}


	\section{Introduction}

Blockchain-based cryptocurrencies like Bitcoin~\cite{nakamoto2008bitcoin} and Ethereum~\cite{buterin2013ethereum} are monetary systems with a market cap of~\$400B~\cite{cryptoslate2020marketCap}.
They enable simple \emph{transactions} of internal tokens and implementation of more elaborate \emph{smart contracts}. 
The transactions create the smart contracts and interact with them. 
Entities called \emph{miners} create data structures called \emph{blocks} that contain transactions.
They publish and order the blocks to form a blockchain, thus \emph{confirming} the included transactions and achieving system progress.
The system state is obtained by parsing the transactions according to the block order. 
Blockchain security relies on \emph{incentives}, rewarding miners with tokens for carrying out their tasks.

A prominent smart-contract design pattern is the \emph{Hashed Time-Locked Contract} (\htlc{}), set up for two participants, Alice,~$\userA$, and Bob,~$\userB$~(\S\ref{sec:previous_work}).
It asserts that~$\userA$ gets tokens for presenting a hash preimage of a specific value before a certain timeout, otherwise~$\userB$ gets them. 
A variety of more elaborate smart-contract designs rely on \htlc{} as a building block.
These include high-frequency payment channels~\cite{poon2016bitcoin,decker2015duplex,green2017bolt,mccorry2016towards,miller2019sprites,dziembowski2018general,dziembowski2017perun}, atomic swaps~\cite{herlihy2018atomic,malavolta2019anonymous,van2019specification,miraz2019atomic,zie2019extending}, contingent payments~\cite{maxwell2016firstZKCP,campanelli2017zero,banasik2016efficient,fuchsbauer2019wi,bursuc2019contingent}, and vaults~\cite{moser2016bitcoin, mccorry2018preventing,bishop2019vaults,zamyatin2019xclaim}.
We identify the specification required by the variety of contracts using~\htlc{} and call it~\htlcSpec{}. 

Unfortunately, \htlc{} is vulnerable to \emph{incentive manipulation attacks}~\cite{bonneau2016buy,mccorry2018smart,judmayer2019pay}. 
Winzer et al.~\cite{winzer2019temporary} showed that~$\userB$ can bribe miners using designated smart contracts to ignore~$\userA$'s transactions until the timeout elapses.
Similarly, Harris and Zohar~\cite{harris2020flood} show that~$\userB$ can delay~$\userA$'s transaction confirmation by overloading the system with his own transactions.
Both of these allow~$\userB$ to obtain the \htlc{} tokens while depriving~$\userA$ of them, even if~$\userA$ published the preimage.

In this work, we provide a secure implementation of~\htlcSpec{}, and further analyze \htlc{}'s susceptibility to bribery.

We begin by describing the model~(\S\ref{sec:model}) for an underlying blockchain mechanism like that of Bitcoin or Ethereum.
The system's state is a set of contracts; each contract comprises a token amount and a predicate; transactions redeem contract tokens by providing inputs that satisfy their predicates.
Users publish transactions initiating new contracts, assigning them with the redeemed tokens while also offering some as fees.
In each round one miner adds a block with a transaction to the chain and receives its fee.

We proceed to present \madhtlc{}, our \htlcSpec{} implementation~(\S\ref{sec:mad_htlc_design}).
\madhtlc{} relies on the fact that miners can participate in a smart contract execution, and thus their interests should be taken into account.
\madhtlc{} utilizes miners as enforcers of its correct execution, allowing and incentivizing them to seize its contract tokens in case of any bribery attempt.
That, in turn, incentivizes~$\userA$ and~$\userB$ to refrain from such attempts and to interact with \madhtlc{} as intended. 
To the best of our knowledge, this is the first work to utilize miner incentives in this way.

In addition to the preimage specified by~\htlcSpec{}, which we denote~\aliceSecret, \madhtlc{} uses a second preimage, \bobSecret, known only to~$\userB$. 
\madhtlc{} comprises a main~\emph{deposit} contract (\depositContract{}) and an auxiliary~\emph{collateral} contract (\collateralContract{}), which work as follows.
\depositContract{} has three so-called \emph{redeem paths}. 
First, it allows~$\userA$ to redeem it with a transaction including the predefined preimage,~$\aliceSecret$. 
Alternatively, it allows~$\userB$ to redeem it after a timeout with a transaction including the other preimage,~$\bobSecret$.
This is essentially the specification, but \depositContract{} provides another option, allowing any user, and specifically any miner, to redeem it herself with a transaction including both~$\aliceSecret$ and~$\bobSecret$.

Now, if both~$\userA$ and~$\userB$ try to redeem \depositContract{}, then their transactions must reveal preimages \aliceSecret and \bobSecret, respectively.
Any miner can then simply take these preimages and issue her own transaction that uses the third redeem path to seize the tokens for herself.
Specifically, if~$\userA$ tries to redeem \depositContract{}, then~$\userB$ is assured that she cannot do so~-- if she tries to redeem the tokens then the miners would get them instead.
Assuming~$\userB$ is benign, i.e., rational but prefers to act honestly for the same reward, then this construction is sufficient to satisfy \htlcSpec{}.
But we can do better.

If~$\userB$ is spiteful, then he will prefer to reduce~$\userA$'s reward if it does not affect his.
When~$\userA$ knows~\aliceSecret and tries to redeem \depositContract{},~$\userB$ cannot redeem it as well, but he can publish a redeeming transaction nonetheless allowing the miners to collect the tokens instead of~$\userA$.

We strengthen \madhtlc{} such that~$\userB$ is strictly incentivized to refrain from such deviant behavior with the auxiliary contract \collateralContract{}.
It can be redeemed only after the same timeout as \depositContract{}, either by a transaction of~$\userB$, or by any miner that provides both~$\aliceSecret$ and~$\bobSecret$.
Now, if~$\userA$ knows~\aliceSecret then she can redeem~\depositContract{} and~$\userB$ can redeem \collateralContract{}. 
If instead~$\userB$ contends with~$\userA$ for \depositContract{}, both still lose the \depositContract{} for the miners; 
but now both~\aliceSecret and~\bobSecret are revealed, allowing miners to seize the \collateralContract{} tokens as well. 
$\userB$ is therefore strictly incentivized not to contend, allowing~$\userA$ to receive the \depositContract{} tokens as required. 

This means the \madhtlc{} construction is secure against the known incentive manipulation attacks~\cite{bonneau2016buy,mccorry2018smart,judmayer2019pay,winzer2019temporary,harris2020flood}~--~$\userB$ cannot incentivize miners to exclude~$\userA$'s transaction and getting his confirmed instead. 

\madhtlc{} utilizes the~\emph{mutual assured destruction}~\cite{deudney1983whole,asgaonkar2019solving} principle: 
If a party misbehaves then all parties lose everything. 
Although penalizing the well-behaved party as well, this mechanism design~\cite{tadelis2013game} technique ensures rational players act as intended.

To prove the security of \madhtlc{}~(\S\ref{sec:mad_htlc_analysis}), we first bound the possible leakage and interactions of system entities using the UC framework~\cite{canetti2007universally}. 
These interactions include the setup, initiation and redeeming of \madhtlc{}.
Then, we formalize \madhtlc{} as a game played by~$\userA$,~$\userB$ and the miners, where the action space comprises the aforementioned possible interactions.
We model all parties as rational non-myopic players, and show the prescribed behavior is~\emph{incentive-compatible}~\cite{roughgarden2010algorithmic}.

We prove the efficacy of~\madhtlc{} by implementing it both in the less expressive Bitcoin Script~\cite{wiki2020bitcoinScript} and in the richer Ethereum Solidity~\cite{solidity} smart-contract languages~(\S\ref{sec:mad_htlc_implementation}).
We deploy it on Bitcoin's and Ethereum's main networks, and show it bears negligible overhead (e.g., 2.2e-6 BTC) compared to the secured amount (e.g., 2.6 BTC).
Specifically for payment-channels~\cite{poon2016bitcoin,decker2015duplex,green2017bolt,mccorry2016towards,miller2019sprites,dziembowski2018general,dziembowski2017perun}, this negligible overhead is only incurred in the abnormal case of a dispute.

\madhtlc{} relies on miners non-myopically optimizing their transaction choices, often referred to as~\emph{Miner Extractable Value} (\emph{MEV})~\cite{daian2020flash,cointelegraph2020mev,felten2020mev}.
While such optimizations are common in the Ethereum network, as of today, Bitcoin's default cryptocurrency client only offers basic optimization. 
Changes in miners' revenue structure will make better optimizations more important. 
To demonstrate miners can easily enhance transaction choice optimization once they choose to do so, we patch the standard Bitcoin client~\cite{coindance2020bitcoinClientDistribution} to create \emph{Bitcoin-MEV infrastructure}, allowing for easy additions of elaborate logic. 
In particular, we implement the logic enabling miners to benefit from enforcing the correct execution of \madhtlc{}.

We then revisit the security of the prevalent \htlc{} implementation and refine previous results~\cite{winzer2019temporary} regarding its vulnerability to bribing attacks~(\S\ref{sec:original_htlc}).
We show that \htlc{} is vulnerable even in blockchains with limited Script-like languages, and bribe can be made using the built-in transaction fee mechanism. 
We analyze miners' behavior as a game for the \htlc{} timeout duration. 
Each suffix of the game can be analyzed as a subgame, and all players have perfect knowledge of the system state. 
$\userB$ can take advantage of this setting to incentivize miners to withhold~$\userA$'s transaction until the timeout, making this the single \emph{subgame perfect equilibrium}~\cite{selten1965spieltheoretische}. 
So, in presence of rational non-myopic miners the required bribe cost is independent of the timeout duration.
This matches the lower bound, and qualitatively tightens the exponential-in-timeout upper bound, both presented by Winzer et al.~\cite{winzer2019temporary}.

In our Bitcoin-compatible attack variation, miners only have to be non-myopic for the attack to succeed, a simple optimization we implement by patching the standard Bitcoin client with merely~150 lines of code.
We identify several potential bribe opportunities on the Bitcoin and Ethereum main networks, including examples of a few dollars bribe would have yielded tens of thousands of dollars in payout (e.g., payment channel where bribe of \$2 could have yielded payout of~\$25K).

We conclude by discussing future directions~(\S\ref{sec:future_directions}), including attacks and mitigations in a weaker model where~$\userA$ or~$\userB$ have mining capabilities, and using \madhtlc{} to reduce latency in systems using \htlcSpec{}.

In summary, we make the following contributions:
\begin{itemize}

	\item We formalize the specification \htlcSpec{} of the prevalent \htlc{} contract;
	
	\item present \madhtlc{} that satisfies \htlcSpec{} utilizing miners as participants;
	
	\item prove \madhtlc{} is secure and incentive compatible;

	\item implement, deploy, and evaluate \madhtlc{} on the Bitcoin and Ethereum main networks;

	\item patch the prevalent Bitcoin Client to create Bitcoin-MEV infrastructure, and to specifically support enforcing correct \madhtlc{} execution;

	\item prove \htlc{} is vulnerable to bribery attacks in limited smart-contract environments; and

	\item qualitatively tighten the bound of Winzer et al.~\cite{winzer2019temporary} and implement the required rational miner behavior.

\end{itemize}

\paragraph*{Open Source and Responsible Disclosure} 
We completed a responsible-disclosure process with prominent blockchain development groups. 
We intend to open source our code, subject to security concerns of the community. 


	\section{Related Work}
	\label{sec:previous_work}

We are not aware of prior work utilizing miners' incentives to use them as enforcers of correct smart contract execution. 
We review previous work on bribing attacks in blockchains~(\S\ref{sec:related_work_bribes}), detail exhibited and postulated mining behavior with respect to transaction selection~(\S\ref{sec:related_work_transaction_optimization}), and present systems and applications using~\htlcSpec{}~(\S\ref{sec:related_work_htlc}). 

		\subsection{Bribery Attacks}
		\label{sec:related_work_bribes}

Winzer et al.~\cite{winzer2019temporary} present attacks that delay confirmation of specific transactions until a given timeout elapses.
Their attacks apply to \htlc{} where~$\userB$ delays the confirmation~of~$\userA$'s redeeming transaction until he can redeem it himself.
Their presented attack requires predicates available only in rich smart contract languages like Ethereum's Solidity~\cite{solidity,dannen2017introducing} and Libra's Move~\cite{blackshear2019move,baudet2018state}, but not Bitcoin's Script~\cite{wiki2020bitcoinScript}.
Specifically, the attack requires setting a \emph{bribe contract} that monitors what blocks are created and rewards participants accordingly. 

In contrast, our attack variation works with Bitcoin's Script as well, as we demonstrate by implementation. 
It therefore applies a wider range of systems~\cite{litecoin2013site,hopwood2016zcash,bitcoinCash2020website}.

Winzer et al.~\cite{winzer2019temporary} present two results regarding the attack costs.
First, they show that~$\userB$'s attack cost for making miner's collaboration with the attack a Nash-equilibrium grows linearly with the \emph{size} (i.e., the  relative mining capabilities) of the smallest miner. 
However, all miners not cooperating with the attack is also a Nash equilibrium. 
Therefore, they analyze~$\userB$'s cost for making the attack a dominant strategy, i.e., to incentivize a to support the attack irrespective of the other miners' strategies. 
This bound grows linearly with relative miner sizes, and exponentially with the \htlc{} timeout. 

Our analysis improves this latter bound by taking into account the miners all know the system state and each others' incentives.
This insight allows us to use the \emph{subgame perfect equilibrium}~\cite{rosenthal1981games,fudenberg1991game,myerson1991game,selten1965spieltheoretische,van2002strategic,watson2002strategy,cerny2014playing,bernheim1984rationalizable,asgaonkar2019solving,roughgarden2010algorithmic} solution concept, a refinement of Nash-equilibrium suitable for games of dynamic nature.
We consider the game played by non-myopic rational participants aware of the game dynamics, and show that a linear-in-miner-size cost (as in~\cite{winzer2019temporary}) suffices for the existence of a \emph{unique} subgame perfect equilibrium.

Other work~\cite{bonneau2016buy,mccorry2018smart,judmayer2019pay} analyzes bribing attacks on the consensus mechanism of cryptocurrency blockchains.
Unlike this work, bribes in these papers compete with the total block reward (not just a single transaction's fee) and lead miners to violate predefined behavior.
These attacks are therefore much more costly and more risky than the bribery we consider, where a miner merely prioritizes transactions for confirmation.

A recent and parallel work~\cite{khabbazian2021timelocked} also suggests using Bitcoin's fee mechanism to attack~$\htlc$.
It assumes miners below a certain hash-rate threshold are myopic  (sub-optimal) while those above it are non-myopic; it presents safe timeout values given Bitcoin's current hash-rate distribution.
In this work, we assume all miners are non-myopic and prove that in this model the attack costs are independent of the timeout.
We also present~$\madhtlc$, which is secure against these attacks with both myopic and non-myopic miners.

		\subsection{Transaction-Selection Optimization}
		\label{sec:related_work_transaction_optimization}

\madhtlc{} incentivizes rational entities to act in a desired way. 
It relies on the premise that all involved parties are rational, and specifically, that they monitor the blockchain state and issue transactions accordingly.

Indeed, previous work~\cite{winzer2019temporary,daian2020flash,robinson2020darkForest,malinova2017market,doweck2020multiparty,eskandari2019sok,zhou2020high,munro2018fomo3ds} shows this premise is prominent, and that system users and miners engage in carefully-planned transaction placing, manipulating their publication times and offered fees to achieve their goals.
Other work~\cite{bentov2019tesseract,prestwich2018minersArentFriends,tsabary2018thegapgame,sliwinskiblockchains,arvindcutoff,tsabary2019heb,easley2019mining} asserts the profitability of such actions is expected to rise as the underlying systems mature, enabling constructions such as \madhtlc{}, which rely on these optimizations.

		\subsection{\htlcSpec{} usage}
		\label{sec:related_work_htlc}

A variety of smart contracts~\cite{maxwell2016firstZKCP,campanelli2017zero,banasik2016efficient,fuchsbauer2019wi,bursuc2019contingent,moser2016bitcoin, mccorry2018preventing,bishop2019vaults,zamyatin2019xclaim} critically rely on \htlcSpec{}. 
To the best of our knowledge, all utilize \htlc{}, making them vulnerable once miners optimize their transaction choices. 
We review some prominent examples. 

\paragraph{Off-chain state channels}
\label{sec:related_work_htlc_offchain_state_channels}
A widely-studied smart contract construction~\cite{decker2015duplex,miller2019sprites,dziembowski2017perun,dziembowski2018general,green2017bolt,khalil2017revive,gudgeon2019sok,malavolta2019anonymous,mccorry2016towards,aumayrgeneralized} with implementations on various blockchains~\cite{poon2016bitcoin,blockstream2020lightningImp,lightningLabs2020lightningImp,acinq2020lightningImp,raiden2020raidenImp,omg2020blockchainDesign} is that of an \emph{off-chain channel} between two parties,~$\userA$ and~$\userB$.

The channel has a \emph{state} that changes as~$\userA$ and~$\userB$ interact, e.g., pay one another by direct communication. 
In the simplest case, the state is represented by a \emph{settlement transaction} that~$\userB$ can place on the blockchain. 
The settlement transaction terminates the channel by placing its final state back in the blockchain. 
The transaction initiates an \htlc{} with a hash digest of~$\userB$'s choice.
$\userB$ can redeem the contract after the timeout or, alternatively,~$\userA$ can redeem it before the timeout if~$\userB$ had shared the preimage with her.

When~$\userA$ and~$\userB$ interact and update the channel state,~$\userB$ \emph{revokes} the previous settlement transaction by sending his preimage to~$\userA$.
This guarantees that if~$\userB$ places a revoked settlement transaction on the blockchain,~$\userA$ can redeem the tokens within the timeout.
Alternatively, if~$\userA$ becomes unresponsive,~$\userB$ can place the transaction on the blockchain and redeem the tokens after the timeout elapses.

Note that this scheme assumes synchronous access to the blockchain~--~$\userA$ should monitor the blockchain, identify revoked-state transactions, and issue her own transaction before the revocation timeout elapses.
To remove this burden, services called \emph{Watchtowers}~\cite{avarikioti2020cerberus,mccorry2019pisa,khabbazian2019outpost} offer to replace~$\userA$ in monitoring the blockchain and issuing transactions when needed.
However, these also require the same synchronous access to the blockchain, and the placement of transactions is still at the hands of bribable miners. 
\madhtlc{} can be viewed as turning the miners themselves into watchtowers~-- watchtowers that directly confirm the transactions, without a bribable middleman.

\paragraph{Atomic swaps}

These contracts enable token exchange over multiple blockchain systems~\cite{herlihy2018atomic,malavolta2019anonymous,van2019specification,miraz2019atomic,zie2019extending,wagner2019dispute}, where a set of parties transact their assets in an atomic manner, i.e., either all transactions occur, or none.

Consider two users,~$\userA$ and~$\userB$, that want to have an atomic swap over two blockchains.
$\userA$ picks a preimage and creates an \htlc{} on the first blockchain with timeout~$\timeout{}_1$.
Then,~$\userB$ creates an \htlc{} requiring the same preimage ($\userB$ knows only its hash digest) and a timeout~$\timeout{}_2 < \timeout{}_1$ on the second blockchain.
$\userA$ publishes a transaction redeeming the \htlc{} on the second blockchain, revealing the preimage and claiming the tokens.
$\userB$ learns the preimage from~$\userA$'s published transaction, and publishes a transaction of his own on the first blockchain.
If~$\userA$ does not publish her transaction before~$\timeout{}_2$ elapses, then the swap is canceled.

	\section{Model}
	\label{sec:model}

We start by describing the system participants and how they form a chain of blocks that contain transactions~(\S\ref{sec:model:blockchain}). 
Next, we explain how the transactions are parsed to define the system state~(\S\ref{sec:model:state}). 
Finally, we detail the required contract specification~\htlcSpec{}~(\S\ref{sec:model:spec}). 

		\subsection{Blockchain, Transactions and Miners} 
		\label{sec:model:blockchain} 

We assume an existing blockchain-based cryptocurrency system, facilitating \emph{transactions} of internal system \emph{tokens} among a set of \emph{entities}.
All entities have access to a digital signature scheme~\cite{badertscher2017bitcoin} with a security parameter~$\secParameter$.
Additionally, they have access to a hash function~$\hashFunction : \rangeGeneral \rightarrow \rangeSecParameter$, mapping inputs of arbitrary length to outputs of length~$\secParameter$.
We assume the value of~$\secParameter$ is sufficiently large such that the standard cryptographic assumptions hold: the digital signature scheme is existentially unforgeable under chosen message attacks (\emph{EU-CMA})~\cite{goldwasser1988digital,badertscher2017bitcoin}, and that~$\hashFunction$ upholds preimage resistance~\cite{rogaway2004cryptographic,dziembowski2018fairswap}.

The blockchain serves as an append-only ledger storing the system state. 
It is implemented as a linked list of elements called \emph{blocks}.
A subset of the entities are called \emph{miners}, who aside from transacting tokens also extend the blockchain by creating new blocks.
We refer to non-mining entities as \emph{users}.

There is a constant set of~$\numberOfMiners$ miners.
Each miner is associated a number representing its relative block-creation rate, or \emph{mining power}. 
Denote the mining power of miner~$i$ by~$\prob{i}$, where~$\sum_{i = 1}^{\numberOfMiners}{\prob{i}} = 1$. 
Denote the minimal mining power by~$\probMin = \min\limits_{i} \prob{i}$.
As in previous work~\cite{winzer2019temporary,eyal2014majority,tsabary2018thegapgame,sapirshtein2016optimal}, these rates are common knowledge, as in practice miners can monitor the blockchain and infer them~\cite{blockchain2020bitcoinPools}.

Block creation is a discrete-time, memoryless stochastic process.
At each time step exactly one miner creates a block.
As in previous work~\cite{dziembowski2017perun,dziembowski2018general,tsabary2018thegapgame,poon2016bitcoin,herlihy2018atomic}, we disregard miners deliberately~\cite{eyal2014majority,Andes2011kryptonite,mirkin2020bdos} or unintentionally~\cite{garay2015backbone,pass2017analysis,kiffer2018better} causing transient inconsistencies (called \emph{forks} in the literature).

Blocks are indexed by their location in the blockchain.
We denote the first block by~\blockIndex{1} and the~$j$'th block by~$\blockIndex{j}$.

Transactions update the system state. 
An entity creates a transaction locally, and can \emph{publish} it to the other entities. 
Transaction publication is instantaneous, and for simplicity we abstract this process by considering published transactions to be part of a publicly-shared data structure called the \emph{mempool}.
As in previous work~\cite{dziembowski2018general,dziembowski2017perun,tsabary2018thegapgame}, all entities have synchronous access to the mempool and the blockchain.

Unpublished and mempool transactions are \emph{unconfirmed}, and are yet to take effect. 
Miners can include unconfirmed transactions of their choice when creating a block, thus \emph{confirming} them and executing the stated token reassignment.

The system limits the number of included transactions per block, and to simplify presentation we consider this limit to be one transaction per block.

The system progresses in steps.
Each step~$j$ begins with system entities publishing transactions to the mempool.
Then, a single miner is selected at random proportionally to her mining power, i.e., miner~$i$ is selected with probability~$\prob{i}$.
The selected miner creates block~$\blockIndex{j}$, either empty or containing a single transaction, and adds it to the blockchain.
This confirms the transaction, reassigning its tokens and awarding that miner with its fee.
The system then progresses to the next step.

		\subsection{System State} 
		\label{sec:model:state} 

The system state is a set of token and \emph{\predicateName{}} pairs called \emph{\contractNamePlural{}}.
Transactions \emph{reassign} tokens from one contract to another. 
We say that a transaction \emph{redeems} a \contractName{} if it reassigns its tokens to one or more new \emph{initiated} \contractNamePlural{}.

To redeem a contract, a transaction must supply input values such that the contract predicate evaluated over them is true. 
Transactions that result in negative predicate value are \emph{invalid}, and cannot be included in a block. 
We simply disregard such transactions.

We say that an entity \emph{owns} tokens if she is the only entity able to redeem their \contractName{}, i.e., the only entity that can provide input data in a transaction that results in positive evaluation of the \contractName{}'s \predicateName{}.

Transactions reassign tokens as follows.
Each transaction lists one or more input \contractNamePlural{} that it redeems, each with its respective provided values.
Each transaction also lists one or more output \contractNamePlural{} that it initiates.
A transaction is only valid if the aggregate amount in the output contracts is not larger than the amount in its redeemed input contracts. 
The difference between the two amounts is the transaction's \emph{fee}. 
The fee is thus set by the entity that creates the transaction.

The system state is derived by parsing the transactions in the blockchain by their order. 
Each transaction reassigns tokens, thus updating the contract set. 
Transaction fees are reassigned to a contract supplied by the confirming miner.

Two transactions \emph{conflict} if they redeem the same contract. 
Both of them might be valid, but only one can be placed in the blockchain. 
Once one of them is confirmed, a block containing the other is invalid. 
We disregard such invalid blocks, and assume miners only produce valid ones. 

There is always at least one unconfirmed valid transaction in the mempool~\cite{blockchain2018mempoolCount,easley2019mining,lavi2019redesigning,arvindcutoff,tsabary2018thegapgame}, and the highest offered fee by any mempool transaction is~$\txFee$, referred to as the~\emph{base} fee.
Miners act rationally to maximize their received fees~(see \S\ref{sec:related_work_transaction_optimization}). 
Users are also rational, and offer the minimal sufficient fee for having their transactions confirmed.

\predicateNamePluralCap{} have access to three primitives of interest: 
\begin{itemize}
	\item $\verifySig{\sig{}}{\publicKey}$: validate that a digital signature~$\sig{}$ provided by the transaction (on the transaction, excluding~$\sig{}$) matches a public key~$\publicKey$ specified in the contract.
	
	\item $\verifyPreImage{\hashInputVar{}}{\hashOutputVar{}}$: validate that a preimage~$\hashInputVar{}$ provided by the transaction matches a hash digest~$\hashOutputVar{}$ specified in the contract, i.e., that~$\hashFunctionWithInput{\hashInputVar{}} = \hashOutputVar{}$.
	
	\item $\verifyTimeout{\timeout}$: validate that the transaction trying to redeem the \contractName{} is in a block at least~$\timeout$ blocks after the transaction initiating it.

\end{itemize}

A predicate can include arbitrary logic composing those primitives. 
In \predicateNamePlural{} that offer multiple redeem options via \emph{or} conditions, we refer to each option as a \emph{redeem path}.

We note that once a transaction is published, its content becomes available to all entities.
We say that an entity \emph{knows} data if it is available to it.

		\subsection{\htlcSpec{} Specification}
        \label{sec:model_desired_functionality}
		\label{sec:model:spec} 

We formalize as~\htlcSpec{} the following \contractName{} specification, used in variety of blockchain-based systems and algorithms~\cite{poon2016bitcoin,mccorry2016towards,miller2019sprites,dziembowski2018general,dziembowski2017perun,campanelli2017zero,banasik2016efficient,fuchsbauer2019wi,herlihy2018atomic,malavolta2019anonymous,van2019specification,miraz2019atomic}.
\htlcSpec{} is specified for two users,~$\userA$ and~$\userB$.
It is parameterized by a hash digest and a timeout, and contains a certain \emph{deposit} amount,~$\depositTokens$. 
$\userA$ gets the deposit if she publishes a matching preimage before the timeout elapses, otherwise~$\userB$ does.

In a blockchain setting,~$\userA$ and~$\userB$ redeem the deposit with a transaction that offers a fee.
We assume the contract token amount~$\depositTokens$ is larger than the base fee~$\txFee$, otherwise the contract is not applicable.

The redeeming transaction by~$\userA$ or~$\userB$ (according to the scenario) should require a fee negligibly larger than the base fee~$\txFee$.
Specifically, the fee amount is independent of~$\depositTokens$.

To construct \htlcSpec{},~$\userA$ and~$\userB$ choose the included hash digest, the timeout, and the token amount,~$\depositTokens$. 
Then either of them issues a transaction that generates the contract with~\depositTokens tokens and the parameterized predicate. 
Either~$\userA$ or~$\userB$ initially knows the preimage, depending on the scenario.

For simplicity, we assume that~$\userA$ either knows the preimage when the transaction initiating \htlcSpec{} is confirmed on the blockchain, or she never does.

	\section{\madhtlc{} Design}
	\label{sec:mad_htlc_design}

We present \madhtlc{}, an implementation of \htlcSpec{}.
\madhtlc{} comprises two sub \contractNamePlural{}\footnote{Separating \madhtlc{} into two sub contracts is for Bitcoin compatibility; these can be consolidated to a single contract in blockchains supporting, richer smart-contract languages, see \S\ref{sec:mad_htlc_implementation}.}~--- \depositContract{}, the core implementation of the \htlcSpec{} functionality, and \collateralContract{}, an auxiliary \contractName{} for collateral, used to disincentivize spiteful behavior by~$\userB$.

\madhtlc{} includes additional variables and parameters along those of \htlcSpec{}, facilitating its realization.
It includes two preimages,~$\aliceSecret$ and~$\bobSecret$; the former corresponds to the preimage of \htlcSpec{}; the latter is an addition in \madhtlc{}, chosen by~$\userB$, used in the various redeem paths.
It also includes the \htlcSpec{} deposit token amount~$\depositTokens$, but also utilizes~$\collateralTokens$ collateral tokens.

Essentially, \depositContract{} lets either~$\userA$ redeem~$\depositTokens$ with preimage~$\aliceSecret$, or~$\userB$ after the timeout with preimage~$\bobSecret$, or any party with both preimages~$\aliceSecret$ and~$\bobSecret$.
\collateralContract{} has~$\collateralTokens$ redeemable only after the timeout, either by~$\userB$, or by any party with both preimages~$\aliceSecret$ and~$\bobSecret$.

We present protocol~$\protocolForBlockchain$ for setup, initiation and redeeming of a~\madhtlc{}~(\S\ref{sec:mad_htlc_design_protocol}), and detail the specifics of \depositContract{}~(\S\ref{sec:mad_htlc_design_alice_bob}) and \collateralContract{}~(\S\ref{sec:mad_htlc_design_deposit}).

		\subsection{Protocol~$\protocolForBlockchain$}
\label{sec:mad_htlc_design_protocol}
Recall that~$\htlcSpec$ is used in several scenarios differing in which party chooses the preimage, when that chosen preimage is shared, and who initiates the contract on the blockchain~(\S\ref{sec:related_work_htlc}).
However, in all scenarios, once the contract is initiated,~$\userA$ can redeem~$\depositTokens$ by publishing the preimage before the timeout elapses, and~$\userB$ can redeem them only after.

So, there are several variants for any protocol that implements~$\htlcSpec$, and we focus on the variant where~$\userB$ picks the first preimage~$\aliceSecret$, potentially shares it with~$\userA$, either~$\userA$ or~$\userB$ can initiate the contract on chain, and either can redeem it using the various redeem paths.
This corresponds to the~\emph{off-chain payment channels} scenario~(\S\ref{sec:related_work_htlc}).

Protocol~$\protocolForBlockchain$~(Protocol~\ref{protocol:madhtlc_presentation}) progresses in phases, and is parameterized by the timeout~$\timeout$ and the token amounts~$\depositTokens$ and~$\collateralTokens$.
First, in the~$\stateSetup$ phase,~$\userB$ randomly draws (denoted by~$\getsRandom$) the two preimages~$\aliceSecret$ and~$\bobSecret$.
He then derives their respective hash digests~$\contractSecretAlice \gets \hashFunctionWithInput{\aliceSecret}$ and~$\contractSecretBob \gets \hashFunctionWithInput{\bobSecret}$, shares~$\contractSecretAlice$ and~$\contractSecretBob$ with~$\userA$.
Upon~$\userA$'s confirmation~$\userB$ creates a transaction~$\initTxName$ that initiates a~$\madhtlc$ with parameters~$\timeout, \contractSecretAlice, \contractSecretBob, \depositTokens, \collateralTokens$ and shares~$\initTxName$ with~$\userA$.

In the following~$\stateInit$ phase,~$\userB$ can share~$\aliceSecret$ with~$\userA$. 
Additionally, either~$\userA$ or~$\userB$ can publish~$\initTxName$ to the mempool, allowing miners to confirm it and initiate the~$\madhtlc$.

In the final~$\stateRedeeming$ phase, once the~$\madhtlc$ is initiated,~$\userA$ and~$\userB$ can redeem~$\depositTokens$ and~$\collateralTokens$ from~$\depositContract$ and~$\collateralContract$, respectively.
Specifically,~$\userA$ redeems~$\depositTokens$ only if she received~$\aliceSecret$ from~$\userB$, and otherwise~$\userB$ redeems~$\depositTokens$.
Either way,~$\userB$ redeems~$\collateralTokens$.

\begin{protocolEnv} 
	
	Protocol~$\protocolForBlockchain$ run by~$\userA$ and~$\userB$ details the setup, initiation and redeeming of a~$\madhtlc$ in the scenario where~$\userB$ picking~$\aliceSecret$.
	It is parameterized by timeout~$\timeout$, and token amounts~$\depositTokens$ and~$\collateralTokens$.

	\headline{$\stateSetup$}
	
	$\userB$ draws~$\aliceSecret \getsRandom \rangeSecParameter, \bobSecret \getsRandom \rangeSecParameter$ and sets~$\contractSecretAlice \gets \hashFunctionWithInput{\aliceSecret},\contractSecretBob \gets \hashFunctionWithInput{\bobSecret}$.
	Then~$\userB$ sends~$\contractSecretAlice$ and~$\contractSecretBob$ to~$\userA$ for confirmation.
	Afterwards~$\userB$, compiles a transaction~$\initTxName$ that initiates a~$\madhtlc$ (both~\depositContract{} and~\collateralContract{}) with~$\contractSecretAlice, \contractSecretBob, \timeout,\depositTokens,\collateralTokens$ as parameters and shares it with~$\userA$.
	$\initTxName$ is not published yet.

	\headline{$\stateInit$}	
	
	$\userB$ can send~$\aliceSecret$ to~$\userA$.
	If so,~$\userA$ expects to receive~$\generalSecret$ such that~$\contractSecretAlice = \hashFunctionWithInput{\generalSecret}$, and ignores other values.

	Either~$\userA$ or~$\userB$ publish~$\initTxName$ to the mempool, and it is eventually included in a block~$\blockIndex{j}$, initiating~$\madhtlc$.
	
	\headline{$\stateRedeeming$}	
	
	If~$\userA$ had received~$\aliceSecret$, she creates and publishes~$\aliceTransactionMadAliceBob$, a transaction redeeming~\depositContract{} using the~$\mhRedeemPathOne$ redeem path.

	$\userB$ waits for the creation of block~$\blockIndex{j+\timeout{} - 1}$.
	If by then~$\userA$ did not publish~$\aliceTransactionMadAliceBob$ then~$\userB$ publishes~$\bobTransactionMadBoth$, redeeming both \depositContract{} and \collateralContract{} through~$\mhRedeemPathTwo$ and~$\mhRedeemPathFour$ redeem paths, respectively.
	If~$\userA$ did publish~$\aliceTransactionMadAliceBob$ then~$\userB$ publishes~$\bobTransactionMadDeposit$, redeeming only \collateralContract{} using the~$\mhRedeemPathFour$ redeem path.

	\caption{$\protocolForBlockchain$}
	\label{protocol:madhtlc_presentation}
\end{protocolEnv} 

		\subsection{\depositContract{}}
		\label{sec:mad_htlc_design_alice_bob}

The \depositContract{} contract is initiated with~$\depositTokens$ tokens. 
Its \predicateName{} is parameterized with~$\userA$'s and~$\userB$'s public keys,~$\alicePublicKey$ and~$\bobPublicKey$, respectively; a hash digest of the predefined preimage~$\contractSecretAlice = \hashFunctionWithInput{\aliceSecret}$ such that any entity other than~$\userA$ and~$\userB$ does not know~$\aliceSecret$, and~$\userA$ or~$\userB$ know~$\aliceSecret$ according to on the specific use case; 
another hash digest~$\contractSecretBob$ such that~$\hashFunctionWithInput{\bobSecret} = \contractSecretBob$, where only~$\userB$ knows~$\bobSecret$; and a timeout~$\timeout$.
The contract has three redeem paths, denoted by~$\mhRedeemPathOne, \mhRedeemPathTwo$ and~$\mhRedeemPathThree$, and presented in Predicate~\ref{alg:modified_redeem_paths}.
Table~\ref{fig:mad_alicebob_table} shows the possible redeeming entities of \depositContract{}.

In the~$\mhRedeemPathOne$ path~(line~\ref{alg:modified_redeem_path_alice}),~$\userA$ can redeem \depositContract{} by creating a transaction including~$\aliceSecret$ and~$\aliceSig$, a signature created using her secret key~$\aliceSecretKey$.
Such a transaction can be included even in the next block~$\blockIndex{j+1}$.
This path is only available to~$\userA$, since only she ever knows~$\aliceSecretKey$.  

In the~$\mhRedeemPathTwo$ path~(line~\ref{alg:modified_redeem_path_bob}),~$\userB$ can redeem \depositContract{} by creating a transaction including~$\bobSecret$ and~$\bobSig$, a signature created using his secret key~$\bobSecretKey$.
Such a transaction can be included in a block at least~$\timeout$ blocks after \depositContract{}'s initiation, that is, not earlier than block~$\blockIndex{j+\timeout}$.
This path is only available to~$\userB$, since only he ever knows~$\bobSecretKey$.

In the~$\mhRedeemPathThree$ path~(line~\ref{alg:modified_redeem_path_miner}), any entity can redeem \depositContract{} by creating a transaction including both~$\aliceSecret$ and~$\bobSecret$.
A transaction taking this redeem path does not require a digital signature, and can be included even in the next block~$\blockIndex{j+1}$.
This path is therefore available to any entity, and specifically to any miner, that knows both~$\aliceSecret$ and~$\bobSecret$.

\begin{algorithm}[t]
	\renewcommand{\algorithmcfname}{Predicate}
	\DontPrintSemicolon
	\SetAlgoNoLine
	\predicateTextSize
	
	\nonl Parameters: $\alicePublicKey, \bobPublicKey, \timeout, \contractSecretAlice, \contractSecretBob$ \\
	
	\nonl $\depositContract{} \left(\hashInputVar{1}, \hashInputVar{2}, \sig{} \right) \coloneqq $ \\
	
	$\depositPredicateMathSpace \left(\verifyPreImage{\hashInputVar{1}}{\contractSecretAlice} \land \verifySig{\sig{}}{\alicePublicKey} \right)$ \label{alg:modified_redeem_path_alice} $\lor$ \tcp*{$\mhRedeemPathOne$} 
	
	$\depositPredicateMathSpace \left(\verifyPreImage{\hashInputVar{2}}{\contractSecretBob} \land \verifySig{\sig{}}{\bobPublicKey} \land \verifyTimeout{\timeout} \right)$ \label{alg:modified_redeem_path_bob} $\lor$ \tcp*{$\mhRedeemPathTwo$}
	
	$\depositPredicateMathSpace \left(\verifyPreImage{\hashInputVar{1}}{\contractSecretAlice} \land \verifyPreImage{\hashInputVar{2}}{\contractSecretBob} \right)$ \label{alg:modified_redeem_path_miner} \tcp*{$\mhRedeemPathThree$}

	\caption{\depositContract{}}
	\label{alg:modified_redeem_paths}
\end{algorithm}

\begin{table}[!t]
	\scriptsize

	\begin{minipage}{\linewidth}
		\centering
		\captionof{table}{Possible redeeming entity of \depositContract{}.}	
		\negspace			
		\begin{tabular}{| c | c | c |} 
			\hline
			\textit{} & \textit{$\bobSecret$ published} & \textit{$\bobSecret$ not published} \\
			\hline
			\textit{$\aliceSecret$ published}		   & Any entity   & $\userA$ \\ \hline
			\textit{$\aliceSecret$ not published}	   & $\userB$   & --- \\ \hline	 
		\end{tabular}

		\label{fig:mad_alicebob_table}
	\end{minipage}%

\vspace{0.5\baselineskip}

	\begin{minipage}{\linewidth}
		\centering
		\captionof{table}{Possible redeeming entity of \collateralContract{}.}
		\negspace		
		\begin{tabular}{| c | c | c |} 
			\hline
			\textit{} & \textit{$\bobSecret$ published} & \textit{$\bobSecret$ not published} \\
			\hline
			\textit{$\aliceSecret$ published}		   & Any entity   & $\userB$ \\ \hline
			\textit{$\aliceSecret$ not published}	   & $\userB$   & $\userB$ \\ \hline	 
		\end{tabular}

		\label{fig:mad_deposit_table}
	\end{minipage} 

\end{table}

		\subsection{\collateralContract{}}
		\label{sec:mad_htlc_design_deposit}

The \collateralContract{} contract is initiated with~$\collateralTokens$ tokens.
Its \predicateName{} is parameterized with~$\userB$'s public key~$\bobPublicKey$; the hash digest of the predefined secret~$\contractSecretAlice = \hashFunctionWithInput{\aliceSecret}$ such that any entity other than~$\userA$ and~$\userB$ does not know~$\aliceSecret$, and~$\userA$ and~$\userB$ know~$\aliceSecret$ based on the specific use case; the hash digest~$\contractSecretBob$ such that~$\hashFunctionWithInput{\bobSecret} = \contractSecretBob$, where only~$\userB$ knows~$\bobSecret$; and a timeout~$\timeout$.
It has two redeem paths, denoted by~$\mhRedeemPathFour$ and~$\mhRedeemPathFive$, and presented in Predicate~\ref{alg:mad_deposit_redeem_paths}.
Table~\ref{fig:mad_deposit_table} shows the possible redeeming entities of \collateralContract{}.

Both paths are constrained by the timeout~$\timeout$, meaning a redeeming transaction can only be included in a block at least~$\timeout$ blocks after the \collateralContract{} initiation~(line~\ref{alg:mad_deposit_timeout}).

In the~$\mhRedeemPathFour$ path~(line~\ref{alg:mad_deposit_redeem_path_bob}),~$\userB$ can redeem \collateralContract{} by creating a transaction including~$\bobSig$, a signature created using his secret key~$\bobSecretKey$.
Only~$\userB$ can redeem \collateralContract{} using this path as he is the only one able to produce such a signature.
This path allows~$\userB$ to claim the collateral tokens in case either he or~$\userA$, but not both, publish a transaction redeeming \depositContract{}.

The~$\mhRedeemPathFive$ path~(line~\ref{alg:mad_deposit_redeem_path_miners}) allows any entity to redeem \collateralContract{} by creating a transaction including both~$\aliceSecret$ and~$\bobSecret$, not requiring any digital signature.
This path allows miners to claim the \collateralContract{} tokens in case~$\userB$ tries contesting~$\userA$ on redeeming \depositContract{}, thus disincentivizing his attempt.

\begin{algorithm}[t]
	\renewcommand{\algorithmcfname}{Predicate}
	\DontPrintSemicolon
	\SetAlgoNoLine
	\predicateTextSize
	
	\nonl Parameters: $\bobPublicKey, \timeout, \contractSecretAlice, \contractSecretBob$ \\
	
	\nonl $\collateralContract{} \left(\hashInputVar{1}, \hashInputVar{2}, \sig{} \right) \coloneqq $ \\
	
	$\collateralPredicateMathSpace \verifyTimeout{\timeout} \land \label{alg:mad_deposit_timeout}$ \\ 
	
	$\; \; \collateralPredicateMathSpace   \big[\verifySig{\sig{}}{\bobPublicKey} $ \label{alg:mad_deposit_redeem_path_bob} $ \lor $ \tcp*{$\mhRedeemPathFour$}
	
	$\; \; \; \collateralPredicateMathSpace  \left( \verifyPreImage{\hashInputVar{1}}{\contractSecretAlice} \land \verifyPreImage{\hashInputVar{2}}{\contractSecretBob} \right)\big] $ \tcp*{$\mhRedeemPathFive$} \label{alg:mad_deposit_redeem_path_miners} 
	
	\caption{\collateralContract{}}
	\label{alg:mad_deposit_redeem_paths}
\end{algorithm}


	\section{\madhtlc{} Security Analysis}
	\label{sec:mad_htlc_analysis}

To prove the security of \madhtlc{} we first show what actions the participants can take to interact with it~(\S\ref{sec:mad_htlc_possible_actions}).
We prove with the UC framework~\cite{canetti2007universally} the security of the setup, initiation and redeeming of a \madhtlc{}.
This analysis yields a set of conditions on which entity can redeem tokens from \madhtlc{}.

Then, we move to analyze how the entities should act to maximize their gains.
We formalize the redeeming of an initiated \madhtlc{} as a game played by~$\userA$,~$\userB$ and the miners~(\S\ref{sec:mad_htlc_game}), and show that they are all incentivized to act as intended~(\S\ref{sec:mad_htlc_proof}).


		\subsection{Setup, Initiation and Redeeming Transactions Security}
		\label{sec:mad_htlc_possible_actions}

Our first goal is to prove the setup and initiation of~$\madhtlc{}$ are secure and to show which valid transactions each participant can generate based on the mempool and blockchain state.
We present an overview of the security claims and proofs, and bring the details in Appendix~\ref{appendix:sec_proof}.

Like prior work~\cite{kiayias2016fair,cheng2019ekiden,kiayias2020composable,badertscher2017bitcoin,dziembowski2018general,luu2016making,wohrer2018smart,huang2019smart,zhou2018security,delmolino2016step,dziembowski2018fairswap,winzer2019temporary,asgaonkar2019solving,atzei2017survey,nikolic2018finding,breidenbach2017parityBug,bhargavan2016formal,tsabary2018thegapgame,poon2016bitcoin}, we assume the blockchain and predicate security holds, including the digital signature scheme and the hash function.

We make the following observation: Transaction invalidity due to~\verifyTimeoutName{} is temporal; this predicate becomes true once sufficiently many blocks are created.
In contrast, two valid transactions can conflict, so only one of them can be confirmed.
We neglect both invalidity reasons and show which valid transactions can be created; clearly, any transaction that is invalid under this relaxation is also invalid without it.
Additionally, we consider only transactions relevant to our protocol, ignoring unrelated transactions.

We formalize parties' ability to redeem the contract under this relaxation using the~$\publishingFunctionName{}\left(\right)$ function~(Eq.~\ref{eq:publishing_function_def}):
Denote~$\redeemPathIndex \in \left\{\mhRedeemPathOne,\mhRedeemPathTwo,\mhRedeemPathThree,\mhRedeemPathFour,\mhRedeemPathFive\right\}$; $\anyParty$ the redeeming party;
$\publishingFunctionHashMatches{a} = 1$ if the redeeming party has a suitable preimage for~$\contractSecretAlice$, and~$0$ otherwise; and
$\publishingFunctionHashMatches{b} = 1$ if the redeeming party has a suitable preimage for~$\contractSecretBob$, and~$0$ otherwise.
Then the relaxed contract predicate is expressed by the function~$\publishingFunctionName{}\left(\redeemPathIndex,\anyParty,\publishingFunctionHashMatches{a},\publishingFunctionHashMatches{b}\right)$.
We note redeeming transactions are published in the mempool, hence publish any included preimages.
\begin{equation}
\label{eq:publishing_function_def}
\begin{aligned}
\publishingFunctionName{}&\left(\redeemPathIndex,\anyParty,\publishingFunctionHashMatches{a},\publishingFunctionHashMatches{b}\right) =  \\
&
\begin{cases}
\left(\anyParty = \userA\right) \land \publishingFunctionHashMatches{a}  & \redeemPathIndex = \mhRedeemPathOne\\
\left(\anyParty = \userB\right) \land \publishingFunctionHashMatches{b}  & \redeemPathIndex = \mhRedeemPathTwo \\
\publishingFunctionHashMatches{a} \land \publishingFunctionHashMatches{b} & \redeemPathIndex \in \left\{\mhRedeemPathThree, \mhRedeemPathFive\right\} \\		
\anyParty = \userB										  & \redeemPathIndex = \mhRedeemPathFour 
\end{cases}
\end{aligned}
\end{equation}

We move to consider the setup, initiation, and redeeming of a single (relaxed) contract with respect to the mempool and the blockchain.
We focus on a \emph{mempool and blockchain projection}~(\emph{mbp}) functionality of a relaxed \madhtlc{}, and we model it as a single ideal functionality,~$\idealFunctionalityMempool$.
This functionality captures the parameter setup of a single contract by~$\userA$ and~$\userB$, its initiation, and redeeming transaction validity due to the~\verifyPreImageName{} and~\verifySigName{} predicates, disregarding conflicts and timeouts.
To facilitate the~$\verifyPreImageName$ predicate and its underlying preimage-resistant hash function~$\hashFunction$, we model the latter as a~\emph{global random oracle} ideal functionality~$\idealFunctionalityHash$~\cite{canetti2007universally,canetti2014practical,camenisch2018wonderful}.
We abstract away digital signatures by considering authenticated channels among parties and functionalities.
We consider an adversary that learns messages sent to~$\idealFunctionalityMempool$ and that messages were sent between parties but not their content.
This modeling is similar to previous work~\cite{dziembowski2018fairswap,dziembowski2018general,miller2019sprites,kiayias2020composable}.

We then define the $\left(\idealFunctionalityHash,\idealFunctionalityMempool\right)$-\emph{hybrid world}~\cite{canetti2007universally} (hereinafter, simply \emph{the hybrid world}), where the~$\idealFunctionalityHash$ and~$\idealFunctionalityMempool$ ideal functionalities reside.
In this hybrid world we then define the \emph{relaxed \madhtlc{}} (\emph{rmh}) protocol~$\protocolForUC$ that is similar to~\protocolForBlockchain{} (Protocol~\ref{protocol:madhtlc_presentation}), but 
(1)~it is defined with~$\idealFunctionalityHash$ and~$\idealFunctionalityMempool$; 
(2)~it considers system entities other than~$\userA$ and~$\userB$, and specifically miners, represented as a third party~$\userO$; and 
(3)~it disregards timeouts and transaction conflicts. 
 
The transition from~$\protocolForUC$ to~\protocolForBlockchain{} is straightforward, and we bring~$\protocolForUC$ in Appendix~\ref{appendix:sec_proof}.
 
Then, our goal is to prove the following lemma, detailing the possible valid transactions the entities can create and publish.

\begin{lemma}
	\label{lemma:valid_transactions}
	Let there be a contract setup and initiated as described by~$\protocolForUC$, let~$\propertyInternalPreimages{a}$ and~$\propertyInternalPreimages{b}$ be indicators whether the preimages~$\aliceSecret$ and~$\bobSecret$ were published in~$\idealFunctionalityMempool$, respectively, and let~$\propertySentPreimageA$ indicate if~$\userB$ shared~$\aliceSecret$ with~$\userA$.
	So, initially~$\propertyInternalPreimages{a} \gets 0$,~$\propertyInternalPreimages{b} \gets 0$ and~$\propertySentPreimageA \gets 0$.
	Then, parties~$\userA$,~$\userB$ and~$\userO$ can only create and publish the following valid redeeming transactions:	
	\begin{itemize}
		\item $\userB$ can publish a valid redeeming transaction using the~$\mhRedeemPathTwo$,~$\mhRedeemPathThree$,~$\mhRedeemPathFour$ or~$\mhRedeemPathFive$ redeem paths.
		Doing so with either~$\mhRedeemPathThree$ or~$\mhRedeemPathFive$ sets~$\propertyInternalPreimages{a} \gets 1$, and with either~$\mhRedeemPathTwo$,~$\mhRedeemPathThree$, or~$\mhRedeemPathFive$ sets~$\propertyInternalPreimages{b} \gets 1$.
		In addition to transaction creation and publication,~$\userB$ can share~$\aliceSecret$ with~$\userA$ (and by doing so sets~$\propertySentPreimageA \gets 1$).
		
		\item If~$\propertyInternalPreimages{a} \lor \propertySentPreimageA = 1$ then $\userA$ can publish a valid redeeming transaction using the~$\mhRedeemPathOne$ redeem path (and by doing so she sets~$\propertyInternalPreimages{a} \gets 1$).
		If~$\left(\propertyInternalPreimages{a} \lor \propertySentPreimageA \right) \land \propertyInternalPreimages{b} = 1$, then $\userA$ can publish a valid redeeming transaction with either the~$\mhRedeemPathThree$ or~$\mhRedeemPathFive$ redeem paths (and by doing so sets~$\propertyInternalPreimages{a} \gets 1$ and~$\propertyInternalPreimages{b} \gets 1$).
		
		\item If~$\propertyInternalPreimages{a} \land \propertyInternalPreimages{b} = 1$ then $\userO$ can publish a valid redeeming transaction with either~$\mhRedeemPathThree$ or~$\mhRedeemPathFive$. 
	\end{itemize}
\end{lemma}

To prove Lemma~\ref{lemma:valid_transactions} we consider an ideal world, where we define a \emph{relaxed \madhtlc{}} ideal functionality~$\idealFunctionalityMadhtlc$~(Functionality~\ref{functionality:game}) that implements the setup, initiation and redeeming of a relaxed~$\madhtlc$ contract.
%
%

$\idealFunctionalityMadhtlc$ maintains indicators~$\mhFSetupA$, $\mhFSetupB$, $\mhFSentPreimageA$, $\mhFPublished$, $\mhFInit$, and~$\mhFInternalPreimages{1}$ and~$\mhFInternalPreimages{2}$, corresponding to execution of~$\madhtlc$~(Protocol~\ref{protocol:madhtlc_presentation}):~$\mhFSetupA$ and~$\mhFSetupB$ correspond to 
to~$\userA$ and~$\userB$ completing their setup;
$\mhFSentPreimageA$ is set if~$\userB$ shared~$\aliceSecret$ with~$\userB$;
$\mhFPublished$ and $\mhFInit$ indicate if the execution reached the~$\stateInit$ and~$\stateRedeeming$ phases, respectively; 
and~$\mhFInternalPreimages{1}$ and~$\mhFInternalPreimages{2}$ are set if~$\aliceSecret$ and~$\bobSecret$ are published with a transaction.

$\idealFunctionalityMadhtlc$ leaks messages to~$\simulator$ and receives a special~$\mhUpdate$ instruction that sets either~$\mhFInternalPreimages{1}$ or~$\mhFInternalPreimages{2}$.
Looking ahead, this allows the simulator to notify~$\idealFunctionalityMadhtlc$ of a publication by a corrupted party.

\begin{functionalityEnv} 
	
	Ideal functionality~$\idealFunctionalityMadhtlc$ in the ideal world represents the setup, initiation and redeeming transaction publication of the contract for session id~$\id$.
	It interacts with parties~$\userA, \userB, \userO$, and simulator~$\simulator$.
	It internally stores indicators~$\mhFSetupA,\mhFSetupB,\mhFSentPreimageA,\mhFPublished,\mhFInit$,~$\mhFInternalPreimages{1}$ and~$\mhFInternalPreimages{2}$, all with initial value of~$0$.
	
	\begin{itemize}
		
		\item Upon receiving~$\left(\mhSetupB,\id\right)$ from~$\userB$ when~$\mhFSetupB = 0$, set~$\mhFSetupB \gets 1$ and leak~$\left(\mhSetupB,\id\right)$ to~$\simulator$. 
		
		\item Upon receiving~$\left(\mhSetupA,\id\right)$ from~$\userA$ when~$\mhFSetupB = 1 \land \mhFSetupA = 0$, set~$\mhFSetupA \gets 1$ and leak~$\left(\mhSetupA,\id\right)$ to~$\simulator$. 
		
		\item Upon receiving~$\left(\mhSharePreimage,\id\right)$ from~$\userB$ when~$\mhFSetupA = 1 \land \mhFSentPreimageA = 0$, set~$\mhFSentPreimageA \gets 1$, and leak~$\left(\mhSharePreimage,\id\right)$ to~$\simulator$.
		
		\item Upon receiving~$\left(\mhPublish,\id\right)$ from either~$\userA$ or~$\userB$ when~$\mhFSetupA = 1 \land \mhFPublished = 0$, set~$\mhFPublished \gets 1$, and leak~$\left(\mhPublish,\id\right)$ to~$\simulator$.
		
		\item Upon receiving~$\left(\mhInit,\id\right)$ from~$\userO$ when~$\mhFPublished = 1 \land \mhFInit = 0$, set~$ \mhFInit \gets 1$, and leak~$\left(\mhInit,\id\right)$ to~$\simulator$.
		
		\item Upon receiving~$\left(\mhTransaction,\id,\redeemPathIndex\right)$ from any party~$\anyParty$ such that~$\redeemPathIndex\in \left\{\mhRedeemPathOne,\mhRedeemPathTwo,\mhRedeemPathThree,\mhRedeemPathFour,\mhRedeemPathFive\right\}$ when~$\mhFInit = 1$, set~$\mhFInternalPreimages{1} \gets \mhFInternalPreimages{1} \lor \left(\left(\anyParty = \userA\right) \land \mhFSentPreimageA \land \left(\redeemPathIndex = \mhRedeemPathOne\right) \right) \lor \left(\left(\anyParty = \userB\right) \land \redeemPathIndex \in \left\{\mhRedeemPathOne,\mhRedeemPathThree, \mhRedeemPathFive\right\} \right)$ and~$\mhFInternalPreimages{2} \gets \mhFInternalPreimages{2} \lor \left(\left(\anyParty = \userB\right) \land \redeemPathIndex \in \left\{\mhRedeemPathTwo, \mhRedeemPathThree, \mhRedeemPathFive\right\} \right)$, denote~$\mhFResult \gets \publishingFunctionName{}\left(\redeemPathIndex,\anyParty,\mhFInternalPreimages{1},\mhFInternalPreimages{2}\right)$, leak~$\left(\mhTransaction,\id,\redeemPathIndex,\anyParty\right)$ to~$\simulator$, and return~$\mhFResult$ to~$\anyParty$.
		
		\item Upon receiving~$\left(\mhUpdate,\id,i\right)$ for~$i\in \left\{0,1\right\}$ from~$\simulator$ through the influence port, set~$\mhFInternalPreimages{i} \gets 1$.
		
	\end{itemize}

	\caption{$\idealFunctionalityMadhtlc$ in the ideal world.}
	\label{functionality:game}
\end{functionalityEnv} 

The construction of~$\idealFunctionalityMadhtlc$ and the definition of~$\publishingFunctionName{}$~(Eq.~\ref{eq:publishing_function_def}) imply that the properties described by Lemma~\ref{lemma:valid_transactions} trivially hold in the ideal world.

We then prove~$\protocolForUC$ UC-realizes~$\idealFunctionalityMadhtlc$, i.e., for any PPT adversary~$\adversary$, there exists a PPT simulator~$\simulator$ such that for any PPT environment~$\funcEnv$, the execution of~$\protocolForUC$ in the hybrid world with~$\adversary$ is computationally indistinguishable from the execution of~$\idealFunctionalityMadhtlc$ in the ideal world with~$\simulator$.

We prove the aforementioned by showing how to construct such a~$\simulator$ for any~$\adversary$, and the derived indistinguishability towards~$\funcEnv$.
$\simulator$ internally manages two preimages on its own, which are indistinguishable from to those chosen by~$\userB$ in the hybrid world: for an honest~$\userB$,~$\simulator$ draws these two preimages from the same distribution as in the real world; for a corrupted~$\userB$,~$\simulator$ learns the chosen preimages throughout the execution.
Additionally,~$\simulator$ internally-simulates~\idealFunctionalityHash{} and~\idealFunctionalityMempool{}, and interacts with~$\idealFunctionalityMadhtlc$ through leakage and influence ports.

The existence of these simulators shows Lemma~\ref{lemma:valid_transactions} applies to the hybrid world as well, meaning it details the possible valid redeeming transactions of a relaxed \madhtlc{}.

Recall the relaxed version considers only the~$\verifyPreImageName$ and~$\verifySigName$ predicates while disregarding $\verifyTimeoutName$ and transactions conflicts, which we now move to consider.

		\subsection{\madhtlc{} Game}
		\label{sec:mad_htlc_game}

The \madhtlc{} construction within the blockchain system gives rise to a game: the participants are~$\userA$,~$\userB$ and the system miners; their utilities are their tokens; and the action space is detailed by Lemma~\ref{lemma:valid_transactions} while considering the timeout constraints and transaction conflicts.

Note that Lemma~\ref{lemma:valid_transactions} considers party~$\userO$ representing any system miner, while the upcoming analysis considers all the miners and their individual rewards.

The \madhtlc{} game begins when the \depositContract{} and \collateralContract{} contracts are initiated in some block~$\blockIndex{j}$.
The game, which we denote by~$\gameDefNameDHTLC{}$, comprises~$\timeout$ rounds, representing the creation of blocks~$\blockIndex{j +1},...,\blockIndex{j +\timeout}$.
Each round begins with~$\userA$ and~$\userB$ publishing redeeming transactions, followed by a miner creating a block including a transaction of her choice.

$\userA$ and~$\userB$'s strategies are their choices of published transactions~-- which transactions to publish, when, and with what fee.
Miner strategies are the choices of which transaction to include in a block if they are chosen to create one.

To accommodate for the stochastic nature of the game~\cite{mertens1981stochastic} we consider entity utilities as the expected number of tokens they own at game conclusion, i.e., after the creation of~$\timeout$ blocks.
$\userA$ and~$\userB$'s utilities depend on the inclusion of their transactions and their offered fees, and miner utilities on their transaction inclusion choices.

We present the game details~(\S\ref{sec:mad_htlc_game_details}) and the suitable solution concept~(\S\ref{sec:mad_htlc_game_sc}).

			\subsubsection{Game Details}
\label{sec:mad_htlc_game_details}

The game progresses in rounds, where each round comprises two steps. 
First,~$\userA$ and~$\userB$ alternately publish transactions, until neither wishes to publish any more.

Note that all published transactions of the current and previous rounds are in the mempool.
Since miners prefer higher fees, for the analysis we ignore any transaction~$\transactionOne$ if there is another transaction~$\transactionTwo$ such that both were created by the same entity, both redeem the same contracts, and~$\transactionTwo$ pays a higher fee than~$\transactionOne$ or arrives before~$\transactionOne$.

Tokens are discrete, hence there is a finite number of fees~$\userA$ and~$\userB$ may offer, meaning the publication step is finite.

Then, a single miner is picked at random proportionally to her mining power and gets to create a block including a transaction of her choice, receiving its transaction fees.
She can also create a new transaction and include it in her block.

\paragraph{Subgames}
The dynamic and turn-altering nature of the game allows us to define \emph{subgames}, representing suffixes of~$\gameDefNameDHTLC{}$.
For any~$\blockCountSpecificVal \in \left[1,\timeout{}\right]$ we refer to the game starting just before round~$\blockCountSpecificVal$ as the~$\blockCountSpecificVal$'th subgame (Fig.~\ref{fig:generic_flow_chart}). 

Note that as miners create blocks and confirm transactions, the system state, including the state of \madhtlc{}, changes.
Specifically, if the \depositContract{} is already redeemed, future blocks do not allow inclusion of conflicting transactions that redeem the \depositContract{} as well.

Hence, when considering \madhtlc{} states we distinguish whether \depositContract{} is redeemable or irredeemable, which we denote by~$\trueConst$ and~$\falseConst$, respectively.
We also note that \collateralContract{} cannot be redeemed until the very last~$\timeout$'th subgame.

Consequently, each subgame~$\blockCountSpecificVal\in \left[1,\timeout{}\right]$ is defined by the number of remaining blocks to be created~$\blockCountSpecificVal$, and the \depositContract{} state~$\contractState \in \{ \trueConst, \falseConst \}$.
We denote such a subgame by~$\gameDefDHTLC{\blockCountSpecificVal}{\contractState}$.

We use~$\subgameIndistinguishable$ to denote sets of subgames, e.g.,~$\gameDefDHTLC{\subgameIndistinguishable}{\trueConst}$ denotes the set of subgames where the contract state~$\contractState$ is~$\trueConst$.

\begin{figure}[!t]
	\centering
	\resizebox{ 0.5 \textwidth }{!}{{\includegraphics[trim=0 0 0 0,clip,angle=0]{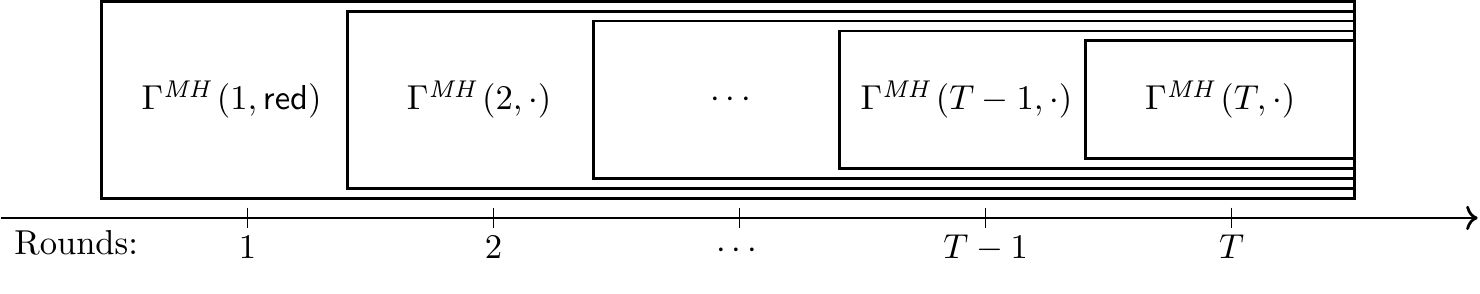}}}
	\caption{$\gameDefNameDHTLC{}$ subgames. }
	\label{fig:generic_flow_chart}
\end{figure}

We refer to~$\gameDefDHTLC{\timeout}{\subgameIndistinguishable}$ as the \emph{final} subgames, as once played, the full game~$\gameDefNameDHTLC{}$ is complete.
We refer to all other subgames as \emph{non-final}.

The game begins when there are~$\timeout$ blocks to be created,~$\userA$ and~$\userB$ did not publish any transactions, and the \depositContract{} is redeemable.
Thus, the initial, complete game is~$\gameDefDHTLC{1}{\trueConst}$.

Once the first round of a non-final subgame is complete, the system transitions to the \emph{subsequent} subgame. 

\paragraph{Actions}
\label{sec:model_generic_game_actions}

$\userA$ and~$\userB$'s actions are the publication of transactions in any~$\gameDefDHTLC{\subgameIndistinguishable}{\subgameIndistinguishable}$ subgame.

$\userA$ can only redeem \depositContract{} and only if she has~$\aliceSecret$ ($\propertySentPreimageA = 1$), hence has a single transaction of interest~$\aliceTransactionMadAliceBob$, offering fee of~$\aliceFeeMadAliceBob$ tokens.
Note~$\aliceTransactionMadAliceBob$ has to outbid unrelated transactions and thus has to offer a fee~$\aliceFeeMadAliceBob > \txFee$, however, cannot offer more tokens than the redeemed ones, so~$ \aliceFeeMadAliceBob < \depositTokens$.
This transaction utilizes the~$\mhRedeemPathOne$ redeem path of \depositContract{}, hence publishing it also publishes~$\aliceSecret$.

$\userB$ can redeem \depositContract{}, \collateralContract{} or both.
We thus consider three transactions of interest:~$\bobTransactionMadAliceBob$, redeeming~$\depositContract{}$ while offering fee~$\bobFeeMadAliceBob$;~$\bobTransactionMadDeposit$, redeeming \collateralContract{} while offering fee~$\bobFeeMadDeposit$; and~$\bobTransactionMadBoth$, redeeming both \depositContract{} and \collateralContract{} while offering fee~$\bobFeeMadBoth$.
To redeem \depositContract{}~$\userB$ uses the~$\mhRedeemPathTwo$ redeem path, hence publishing transactions~$\bobTransactionMadAliceBob$ or~$\bobTransactionMadBoth$ also publishes~$\bobSecret$.
Redeeming \collateralContract{} is by the~$\mhRedeemPathFour$ redeem path.
Similarly to~$\userA$'s fee considerations,~$\userB$'s transactions have to outbid unrelated transactions, and cannot offer more tokens than they redeem, so~$\txFee < \bobFeeMadAliceBob < \depositTokens$,~$\txFee < \bobFeeMadDeposit < \collateralTokens$ and~$\txFee < \bobFeeMadBoth < \depositTokens + \collateralTokens$.

A miner's action is the choice of a transaction to include if she is chosen to create a block.
First, she can include a transaction unrelated to \madhtlc{} in any~$\gameDefDHTLC{\subgameIndistinguishable}{\subgameIndistinguishable}$ subgame.

She can also include any of the following transactions, assuming they were previously published by~$\userA$ or~$\userB$, and as a function of the contract state:
$\aliceTransactionMadAliceBob$ if~\depositContract{} is redeemable, that is, in any~$\gameDefDHTLC{\subgameIndistinguishable}{\trueConst}$;
$\bobTransactionMadDeposit$ if the timeout has elapsed, that is, in any~$\gameDefDHTLC{\timeout}{\subgameIndistinguishable}$; and  
$\bobTransactionMadAliceBob$ or~$\bobTransactionMadBoth$ if the timeout has elapsed and~\depositContract{} is redeemable, that is, in~$\gameDefDHTLC{\timeout}{\trueConst}$.

Conditioned on knowing~$\aliceSecret$ and~$\bobSecret$ through published transactions, a miner can also create and include the following transactions, redeeming the contracts herself:

\begin{itemize}	
	\item Transaction~$\minerTransactionMadAliceBob$ redeeming \depositContract{}, using the~$\mhRedeemPathThree$ redeem path, and getting the~$\depositTokens$ tokens of \depositContract{} as reward.
	This action is only available if the miner knows both~$\aliceSecret$ and~$\bobSecret$, and if \depositContract{} is redeemable, that is, in any~$\gameDefDHTLC{\subgameIndistinguishable}{\trueConst}$ subgame where~$\aliceTransactionMadAliceBob$ and either of~$\bobTransactionMadAliceBob$ or~$\bobTransactionMadBoth$ were published.
	
	\item Transaction~$\minerTransactionMadDeposit$ redeeming \collateralContract{}, using the~$\mhRedeemPathFive$ redeem path, and getting the~$\collateralTokens$ tokens of \collateralContract{} as reward.
	This action is only available if the miner knows both~$\aliceSecret$ and~$\bobSecret$, and the timeout has elapsed, that is, in any~$\gameDefDHTLC{\timeout}{\trueConst}$ subgame where~$\aliceTransactionMadAliceBob$ and either of~$\bobTransactionMadAliceBob$ or~$\bobTransactionMadBoth$ were published.

	\item Transaction~$\minerTransactionMadBoth$ redeeming both \depositContract{} and \collateralContract{}, using the~$\mhRedeemPathThree$ and~$\mhRedeemPathFive$ redeem paths, and getting the~$\depositTokens + \collateralTokens$ tokens of \depositContract{} and \collateralContract{} as reward.
	This action is only available if the miner knows both~$\aliceSecret$ and~$\bobSecret$, the \depositContract{} is redeemable, and the timeout has elapsed, that is, in subgame~$\gameDefDHTLC{\timeout}{\trueConst}$ where~$\aliceTransactionMadAliceBob$ and either of~$\bobTransactionMadAliceBob$ or~$\bobTransactionMadBoth$ were published.

\end{itemize}

We disregard actions that are trivially dominated~\cite{watson2002strategy}, such as~$\userA$ and~$\userB$ sharing their secret keys or publishing the relevant preimages not via a transaction; a miner including a transaction of another entity that redeems either of the contracts using the two preimages instead of redeeming it herself; and a miner creating an empty block instead of including an unrelated transaction.

\paragraph{Strategy}
A \emph{strategy}~$\strategyOfEntity$ is a mapping from each subgame to a respective feasible action, stating that an entity takes that action in the subgame.
We call the strategy vector of all entities in a game a \emph{strategy profile}, denoted by~$\strategyProfile$.

\paragraph{Utility}
Recall an entity's utility is her expected accumulated token amount at game conclusion.
We define the utility of an entity in a subgame as the expected token amount she accumulates within the subgame until its conclusion.
We denote the utility of entity~$i$ when all entities follow~$\strategyProfile$ in subgame~$\gameDefDHTLC{\blockCountSpecificVal}{\contractState}$ by~$\utilityOfEntity{i}{\strategyAllMiners}{\gameDefDHTLC{\blockCountSpecificVal}{\contractState}}$.

			\subsubsection{Solution Concept}
			\label{sec:mad_htlc_game_sc}


Note that block-creation rates, entity utilities and their rationality are all common knowledge, and that when choosing an action an entity is aware of the current system state.
That means any subgame~$\gameDefDHTLC{\blockCountSpecificVal}{\contractState}$ is of \emph{perfect information}~\cite{shoham2008multiagent,osborne1994course}. 
We are thus interested in strategy profiles that are \emph{subgame perfect equilibria}~\cite{rosenthal1981games,fudenberg1991game,myerson1991game,selten1965spieltheoretische,van2002strategic,watson2002strategy,cerny2014playing,bernheim1984rationalizable}.

A strategy profile~$\strategyProfile$ is a subgame perfect equilibrium in~$\gameDefDHTLC{\blockCountSpecificVal}{\contractState}$ if, for any subgame, no entity can increase her utility by deviating to a different strategy, where it knows how the other players would react based on their perfect knowledge.
This implies that for each subgame, the actions stated by~$\strategyProfile$ are a Nash equilibrium.

We say that a prescribed strategy profile is \emph{incentive compatible}~\cite{roughgarden2010algorithmic} if it is a subgame perfect equilibrium, and the utility of each player is not lower than her utility in any other subgame perfect equilibrium.
So an entity cannot deviate to increase her utility, and there are no other more favorable equilibria.

Our analysis utilizes the common technique of \emph{backward induction}~\cite{zermelo1913anwendung,aumann1995backward,kaminski2017backward,bernheim1984rationalizable}, suitable for perfect-information finite games.
Intuitively, to determine her best action, a player analyzes the game outcome for each possible action, repeating the process recursively for each possible game suffix. 


		\subsection{\madhtlc{} Incentive Compatibility}
		\label{sec:mad_htlc_proof}

We now show the \madhtlc{} prescribed behavior~(Protocol~\ref{protocol:madhtlc_presentation}) is incentive compatible and implements \htlcSpec{}.

We first analyze~$\userA$'s and~$\userB$'s utilities when both follow the prescribed strategy, starting with the scenario where~$\userA$ knows the preimage~$\aliceSecret$ (i.e., when~$\propertySentPreimageA = 1$).

\begin{lemma}
	\label{lemma:mad_proof_prescribed_utilities_when_alice_knows}
	In~$\gameDefDHTLC{1}{\trueConst}$, if~$\userA$ knows~$\aliceSecret$ and~$\userA$ and~$\userB$ both follow the prescribed strategies, then miners' best-response strategy leads to~$\userA$ redeeming \depositContract{} for~$\depositTokens -\aliceFeeMadAliceBob$ tokens, and~$\userB$ redeeming \collateralContract{} for~$\collateralTokens - \bobFeeMadDeposit$ tokens.
\end{lemma}

\begin{proof}
	The prescribed strategy states that~$\userA$ publishes~$\aliceTransactionMadAliceBob$ during the first~$\timeout-1$ rounds, and that~$\userB$ publishes~$\bobTransactionMadDeposit$ in round~$\timeout$.
	
	Note that~$\userB$ does not publish~$\bobTransactionMadAliceBob$ and~$\bobTransactionMadBoth$, hence miners do not know~$\bobSecret$.
	The transactions~$\aliceTransactionMadAliceBob$ and~$\bobTransactionMadDeposit$ offer~$\aliceFeeMadAliceBob$ and~$\bobFeeMadDeposit$ fees, respectively, both greater than the base fee~$\txFee$.

	The induced subgames therefore enable miners to include~$\aliceTransactionMadAliceBob$ in one of the first~$\timeout -1$ blocks, and including~$\bobTransactionMadDeposit$ in the last one.
	Using backward induction shows the subgame perfect equilibrium is to include~$\aliceTransactionMadAliceBob$ in its published round, and~$\bobTransactionMadDeposit$ in the last.
	
	So both~$\aliceTransactionMadAliceBob$ and~$\bobTransactionMadDeposit$ are included in blocks, and~$\userA$ and~$\userB$ get~$\depositTokens -\aliceFeeMadAliceBob$ and~$\collateralTokens - \bobFeeMadDeposit$ tokens, respectively.
\end{proof}

We now consider~$\userA$ that does not know~$\aliceSecret$ (i.e., when~$\propertySentPreimageA = 0$).
\begin{lemma}
	\label{lemma:mad_proof_prescribed_utilities_when_alice_doesnt_know}
	In~$\gameDefDHTLC{1}{\trueConst}$, if~$\userA$ does not know~$\aliceSecret$ and~$\userA$ and~$\userB$ both follow the prescribed strategies, then miners' best-response strategy leads to~$\userB$ redeeming both \depositContract{} and \collateralContract{} for~$\depositTokens + \collateralTokens - \bobFeeMadBoth$ tokens, and~$\userA$ gets none.
\end{lemma}

\begin{proof}
	As~$\userA$ does not know~$\aliceSecret$ she does not publish any transaction, hence redeems no contract and receives no tokens.
	
	By the prescribed strategy~$\userB$ publishes~$\bobTransactionMadBoth$, offering fee~$\bobFeeMadBoth > \txFee$ and revealing~$\bobSecret$.
	However,~$\aliceSecret$ is not published, so miners cannot redeem \depositContract{} and \collateralContract{} themselves.
	Therefore, miners maximize their utility by including~$\bobTransactionMadBoth$ in the last round.
	
	That means~$\bobTransactionMadBoth$ is included in a block, and~$\userA$ and~$\userB$ get~$0$ and~$\depositTokens + \collateralTokens - \bobFeeMadBoth$ tokens, respectively.
\end{proof}

We now present three lemmas, considering potential deviations from the prescribed strategy, and showing that any such deviation is strictly dominated.
We provide the gist of the proofs, with the details deferred to Appendix~\ref{appendix:additional_proofs}.

We first show that if~$\userA$ and~$\userB$ contend then the miners do not take their transactions in the last round. 
\begin{lemma}
	\label{lemma:mad_proof_last_round_mutual_destruction}
	In the last round of the game, i.e. subgame~$\gameDefDHTLC{\timeout}{\subgameIndistinguishable}$, if~$\aliceTransactionMadAliceBob$ and either~$\bobTransactionMadAliceBob$ or~$\bobTransactionMadBoth$ are published then miners' best-response strategy is not to include any of~$\userA$'s or~$\userB$'s transactions in this round.
\end{lemma}

This holds because in the described scenario any miner can simply redeem all the tokens herself.
Then we show~$\userA$'s cannot deviate to increase her utility.
\begin{lemma}
	\label{lemma:mad_proof_alice_plays_prescribed}
	In~$\gameDefDHTLC{1}{\trueConst}$,~$\userA$ cannot increase her utility by deviating from the prescribed strategy.
\end{lemma}

This holds as publishing at the last round or not publishing at all results with~$\userA$ not getting any tokens.
Similarly, we claim~$\userB$ does not gain from deviating.
\begin{lemma}
	\label{lemma:mad_proof_bob_plays_prescribed}
	In~$\gameDefDHTLC{1}{\trueConst}$,~$\userB$ cannot increase his utility by deviating from the prescribed strategy.
\end{lemma}

Intuitively, if~$\userB$ publishes when~$\userA$ also does then~$\userB$ loses all the tokens, whilst refraining from doing so earns him the collateral.

Following directly from Lemma~\ref{lemma:mad_proof_alice_plays_prescribed} and Lemma~\ref{lemma:mad_proof_bob_plays_prescribed}), we obtain:

\begin{corollary}
	\label{corollary:mad_proof_incentive_compatible}
	The prescribed strategy of \madhtlc{} is a unique subgame perfect equilibrium, and as such, incentive compatible.
\end{corollary}

We are now ready to prove our main theorem:

\begin{theorem}
	\label{theorem:mad_satisfies_tlc}
	\madhtlc{} satisfies \htlcSpec{} with rational PPT participants.
\end{theorem}

\begin{proof}

	Lemma~\ref{lemma:valid_transactions} shows the possible redeeming transactions for PPT participants, disregarding invalidity due to timeouts and transaction conflicts.
	Consequently, the game description considering the timeouts and conflicts~(\S\ref{sec:mad_htlc_game}) captures the possible redeeming transactions of PPT participants. 
	
	The game analysis~(Corollary~\ref{corollary:mad_proof_incentive_compatible}) shows the prescribed strategy~(Protocol~\ref{protocol:madhtlc_presentation}) is incentive compatible, and
	Lemma~\ref{lemma:mad_proof_prescribed_utilities_when_alice_knows} and 
	Lemma~\ref{lemma:mad_proof_prescribed_utilities_when_alice_doesnt_know} show the prescribed strategy matches \htlcSpec{}.
	Note that matching \htlcSpec{}, Protocol~\ref{protocol:madhtlc_presentation} states the redeeming transaction fee should be negligibly larger than~$\txFee$, and is independent of~$\depositTokens$.
\end{proof}

\paragraph*{Myopic Miners}
\madhtlc{}'s design deters~$\userB$ from bribe attempts as he knows rational non-myopic miners will seize his funds if he acts dishonestly.

However, even in the presence of unsophisticated, myopic miners, \madhtlc{} still satisfies \htlcSpec{}.
The common transaction selection logic~\cite{bitcoincore2020website,geth2020website,parity2020website,ethstackexchange2020txordering} as of today has miners myopically optimize for the next block. 
Since~$\userB$'s transaction can only be confirmed in the last round, these miners will simply include~$\userA$'s transaction, achieving the desired outcome. 

Only miners that are sophisticated enough to be non-myopic but not sophisticated enough to take advantage of the~$\mhRedeemPathThree$ path would cooperate with the attack. 
But even in the presence of such miners, it is sufficient for one miner (or user) to take advantage of the~$\mhRedeemPathThree$ path during the~$\timeout$ rounds in order to thwart the attack.

	\section{\madhtlc{} Implementation}
	\label{sec:mad_htlc_implementation}	

We demonstrate the efficacy of \madhtlc{} by evaluating it in Bitcoin and Ethereum. 
We discuss the deployment of \madhtlc{} and its overhead~(\S\ref{sec:implementaiton_overheard_deployment}), and our implementation of a framework for implementing MEV infrastructure~\cite{daian2020flash,cointelegraph2020mev,felten2020mev} on Bitcoin~(\S\ref{sec:rational_miner_implementation}), used to facilitate \madhtlc{} guarantees.

		\subsection{Contract Implementation, Overhead and Deployment}
		\label{sec:implementaiton_overheard_deployment}

We implement \depositContract{} and \collateralContract{} in Bitcoin's Script~\cite{wiki2020bitcoinScript} and Ethereum's Solidity~\cite{solidity,dannen2017introducing} smart contract languages.
We also implement a version of the standard \htlc{} for reference.
We bring the code in Appendix~\ref{appendix:bitcoin_implementation}.

We briefly discuss these implementations, show their transaction-fee overhead is negligible compared to the secured amounts, and present main network deployments.

\paragraph*{Bitcoin implementation}
Bitcoin's transaction fees are determined by the transaction sizes. 
Our contracts use \emph{P2SH}~\cite{p2sh2013bip} (non SegWit~\cite{segwit2017bip}) addresses, so the initiating transactions contain only the hashes of the scripts, and each contract initiation within a transaction requires~28 bytes.
The redeeming transactions provide the full predicate script along with its inputs.
Table~\ref{tab:contract_comparison_bitcoin} presents the script and redeeming transaction sizes of \htlc{}, \depositContract{} and \collateralContract{}.

A transaction redeeming \depositContract{} is about~50 bytes larger than one redeeming~\htlc{}. 
At the current Bitcoin common fees~\cite{bitcoinTxFees2020} and exchange rate~\cite{coindesk2020bitcoinPrice} implies an additional cost of a mere \$0.02.
Including the auxiliary \collateralContract{} implies an additional cost of about \$0.10. 

Further size-reduction optimizations such as using SegWit transactions and merging multiple transactions can also be made, but are outside the scope of this work.
\begin{table}[t]
	\scriptsize
	\captionof{table}{Bitcoin contract and redeeming transaction sizes.}	
	\negspace
	\negspace	
	\begin{center}
		\begin{tabular}{| >{\centering\arraybackslash}p{0.09\linewidth} | >{\centering\arraybackslash}m{0.16\linewidth} | >{\centering\arraybackslash}m{0.19\linewidth} | >{\centering\arraybackslash}m{0.31\linewidth} |} 
			\hline
			\textbf{Contract} & \textbf{Size [bytes]} & \textbf{Redeem path} & \textbf{Redeeming tx [bytes]} \\
			\hline
			
			\multicolumn{1}{|l|}{\multirow{2}{*}{\htlc{}}} & \multirow{2}{*}{99} & $\htlcRedeemPathOne$ & 291  \\ 
			
			& 					   & $\htlcRedeemPathTwo$   & 259 \\ \hline
			
			\multicolumn{1}{|l|}{\multirow{3}{*}{\depositContract{}}} & \multirow{3}{*}{129} & $\mhRedeemPathOne$ & 323  \\ 
			
			& 					       & $\mhRedeemPathTwo$   & 322 \\ 
			
			& 					       & $\mhRedeemPathThree$ & 282 \\ \hline			 
			
			\multicolumn{1}{|l|}{\multirow{2}{*}{\collateralContract{}}} & \multirow{2}{*}{88} & $\mhRedeemPathFour$   & 248  \\ 
			
			& 					     & $\mhRedeemPathFive$ & 241 \\ \hline
			
		\end{tabular}
		\label{tab:contract_comparison_bitcoin}
		\negspace
		\negspace
	\end{center}
\end{table}

\paragraph*{Ethereum implementation}

Compared to Bitcoin's Script, Solidity~\cite{solidity,dannen2017introducing} is a richer smart contract language, allowing \madhtlc{} to be expressed as a single contract consolidating \depositContract{} and \collateralContract{}.

On the Ethereum platform transactions pay fees according to their so-called~\emph{gas} usage, an inner form of currency describing the cost of each operation the transaction performs.
We compare the initiation and redeeming costs of \htlc{} and \madhtlc{}.
Note that \madhtlc{} contains about twice the code of \htlc{}, and as expected, its operations are more gas-consuming.
We bring the details in Table~\ref{tab:contract_comparison_ethereum}.

We stress these numbers regard the most basic, straight-forward implementation, and that Ethereum and Solidity enable further optimizations~-- for example, deploying a more elaborate~\emph{library contract} once~\cite{ethFoundation2020smartcontractlibraries}, and simpler~\emph{contract instances} that use the library, achieving significantly reduced amortized costs.
More importantly, the additional fee costs are independent of (e.g., \$0.2~\cite{liquality2020atomicSwapEth}), and can be negligible compared to, the secured amounts (e.g.,~\$6.2K~\cite{liquality2020atomicSwapEth}).

\begin{table}[t]
	\scriptsize
	\captionof{table}{Ethereum gas for contract initiation and redeeming.}
	\negspace
	\negspace	
	\begin{center}
		\begin{tabular}{| >{\centering\arraybackslash}p{0.09\linewidth} | >{\centering\arraybackslash}m{0.21\linewidth} | >{\centering\arraybackslash}m{0.19\linewidth} | >{\centering\arraybackslash}m{0.23\linewidth} |} 
			\hline
			\textbf{Contract} & \textbf{Initiation [gas]} & \textbf{Redeem path} & \textbf{Redeeming [gas]} \\
			\hline
			
			\multicolumn{1}{|l|}{\multirow{2}{*}{\htlc{}}} & \multirow{2}{*}{362,000} & $\htlcRedeemPathOne$ & 34,734  \\ 
			
			& 					   & $\htlcRedeemPathTwo$   & 32,798 \\ \hline
			
			\multicolumn{1}{|l|}{\multirow{5}{*}{\madhtlc{}}} & \multirow{5}{*}{600,000} & $\mhRedeemPathOne$ & 58,035  \\ 
			
			& 					       & $\mhRedeemPathTwo$   & 58,885 \\ 
			
			& 					       & $\mhRedeemPathThree$ & 59,043 \\
			
			&                          & $\mhRedeemPathFour$  & 41,175  \\ 
			
			& 					       & $\mhRedeemPathFive$ & 44,887 \\ \hline
			
		\end{tabular}
		\label{tab:contract_comparison_ethereum}
		\negspace
		\negspace		
	\end{center}
\end{table}

Recall that for off-chain channels this overhead is incurred only in the abnormal unilateral channel closure.

\paragraph*{Main network deployment}

We deployed \madhtlc{} on both blockchains (Appendix~\ref{appendix:main_network_deployment} details the transaction IDs).

For Bitcoin, we deployed three \depositContract{} instances on the main network and redeemed them using its three redeem paths.
We also deployed two \collateralContract{} instances and redeemed using its two redeem paths.

For Ethereum, we deployed a consolidated~\madhtlc{}, and posted transactions redeeming the~$\depositTokens$ through both~$\mhRedeemPathOne$ and~$\mhRedeemPathTwo$.
These transactions offered relatively low fees, so were not included in a block by any miner, and only revealed~$\aliceSecret$ and~$\bobSecret$.
At this point there were no other transactions trying to redeem the \madhtlc{}, although users and miners monitoring the blockchain could have created a transaction redeeming the~$\depositTokens$ using the~$\mhRedeemPathThree$ with the revealed~$\aliceSecret$ and~$\bobSecret$.
We deduce this optimization currently does not take place on the Ethereum main network.

Then, we published a transaction of our own using~$\mhRedeemPathThree$, revealing (again)~$\aliceSecret$ and~$\bobSecret$, offering a relatively-high fee. 
Nevertheless, our transaction was slightly out-bid by another transaction, which also used~\mhRedeemPathThree, and took the deposit. 
It was likely published by a front-running bot~\cite{robinson2020darkForest,zhou2020high,daian2020flash},
presenting yet another example of entities monitoring the blockchain looking for MEV opportunities~\cite{robinson2020darkForest,zhou2020high,daian2020flash}, as required for \madhtlc{} security. 


		\subsection{Bitcoin-MEV Infrastructure}
		\label{sec:rational_miner_implementation}

By default, cryptocurrency clients~\cite{bitcoincore2020website,geth2020website,parity2020website,ethstackexchange2020txordering} only perform myopic transaction-inclusion optimizations, trying to generate a single maximal-fee block each time. 
As recently shown~\cite{robinson2020darkForest,zhou2020high,daian2020flash} (including in our deployment above), miners and other entities perform more sophisticated optimizations on the Ethereum network.

In contrast, we are not aware of similar optimizations taking place on the Bitcoin network.
Specifically, Bitcoin Core, which is used by roughly 97\% of current Bitcoin nodes~\cite{coindance2020bitcoinClientDistribution}, maintains a local \emph{mempool} data structure that only contains unconfirmed transactions whose timeouts (if any) have elapsed.
This implementation prevents miners from optimizing based on transaction pending on a timeout. 
However, this limitation is not a consensus rule, but an implementation choice. 
Taking more elaborate considerations into account when choosing transactions is not a violation of legitimate miner behavior. 

As noted~(\S\ref{sec:related_work_transaction_optimization}), optimizing transaction revenue is becoming more important for miners over time. 
To demonstrate the ease of achieving broader optimizations, including non-myopic considerations, we implemented \emph{Bitcoin-MEV}, an infrastructure allowing to easily incorporate any logic over Bitcoin Core received transactions.

Bitcoin-MEV's main design goal is to enable users to deploy their own optimization algorithms.
It comprises a patched (140~LoC) C++ Bitcoin Core node with additional RPCs, and a Python script (Fig.~\ref{fig:rmi_design}, new components shaded), working as follows.

When the node receives a new transaction, instead of directly placing it in its mempool, it pushes the transaction to a designated \emph{new transaction queue}.
The Python script monitors this queue with a dedicated RPC, fetches new transactions and parses them.
Then, based on the implemented optimization algorithm, it can instruct the node how to handle the transaction~-- insert it to the mempool, discard it, or keep it for future use.
The Python script can also generate new transactions and send them to the node.

We implemented and locally tested a Python script (350~LoC) for enforcing \madhtlc{} by taking advantage of the opportunities it provides miners to increase their revenue.
We screen received transactions, tease out~$\userA$ and~$\userB$'s preimages, and create a transaction redeeming the \madhtlc{} contracts using the extracted preimages.

\begin{figure}[!t]
	\centering
	\resizebox{ 0.6 \textwidth }{!}{{\includegraphics[trim=85 212 20 75,clip]{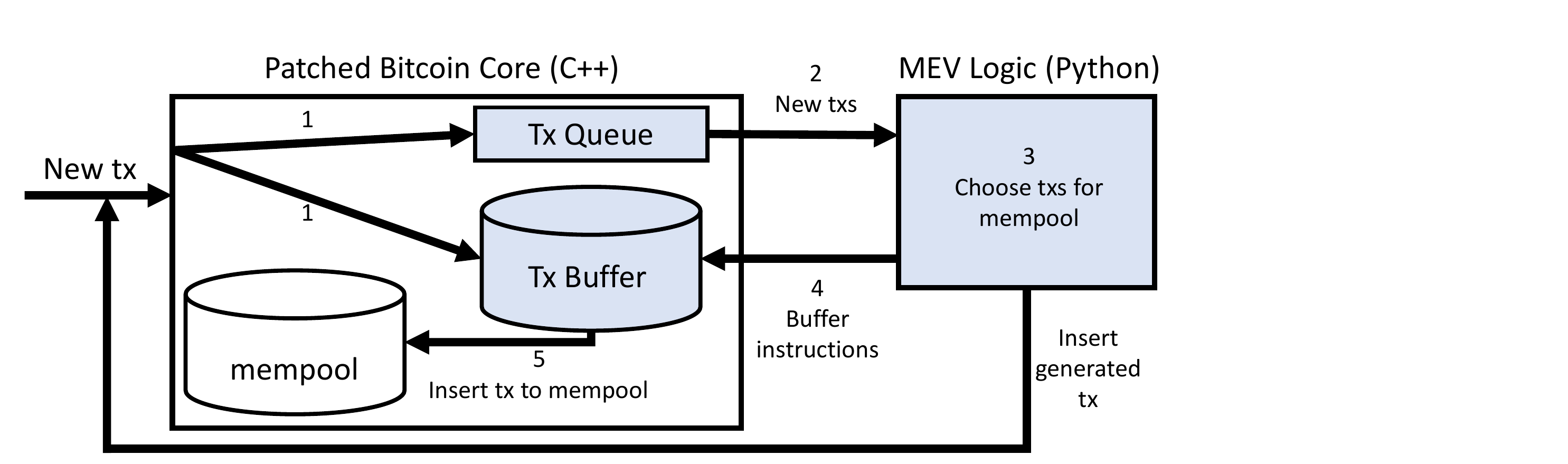}}}
	\caption{Bitcoin-MEV, new components shaded. }
	\label{fig:rmi_design}
\end{figure}


	\section{\htlc{}}
	\label{sec:original_htlc}

The prevalent implementation of~\htlcSpec{} is a direct translation of the specification to a single contract called~\htlc{}~(\S\ref{sec:original_htlc_definition}). 

It relies on the premise that miners benevolently enforce the desired execution, namely include~$\userA$'s transaction in a block before the timeout elapses. 
However, this assumption contradicts the core principle of cryptocurrency permissionless systems~--- miners operate to make profit~\cite{malinova2017market,tsabary2018thegapgame,arvindcutoff,doweck2020multiparty,tsabary2019heb,bentov2019tesseract,prestwich2018minersArentFriends,sliwinskiblockchains}, and include transactions that benefit their personal gains~\cite{daian2020flash,eskandari2019sok,munro2018fomo3ds,robinson2020darkForest}.
Specifically,~$\userB$ can incentivize miners with a bribe~\cite{winzer2019temporary,mccorry2018smart,judmayer2019pay} to exclude~$\userA$'s transaction until the timeout elapses, and then redeem the \htlc{} himself.

We analyze the security of~\htlc{} by formalizing the game played by the entities~(\S\ref{sec:original_htlc_game}), and showing how cheap~$\userB$'s required bribe is~(\S\ref{sec:original_htlc_attack}). 
We show miner fee optimization is easy by implementing a bribery-accepting (i.e., rational and non-myopic) miner~(\S\ref{sec:htlc_bribe_implementation}), and conclude by estimating the actual attack cost using numbers from operational contracts~(\S\ref{sec:original_htlc_examples}).

		\subsection{Construction}
\label{sec:original_htlc_definition}		

$\userA$ and~$\userB$ execute \htlcSpec{} by having an~\htlc{} contracted with~\depositTokens tokens placed in some block~$b_j$. 
The \htlc{}'s \predicateName{} is parameterized with~$\userA$'s and~$\userB$'s public keys,~$\alicePublicKey$ and~$\bobPublicKey$, respectively; a hash digest of the predefined secret~$\contractSecretAlice = \hashFunctionWithInput{\aliceSecret}$ such that any entity other than~$\userA$ and~$\userB$ does not know~$\aliceSecret$ ($\userA$ or~$\userB$ know~$\aliceSecret$ based on the specific use case); and a timeout~$\timeout$.

\htlc{} has two redeem paths, denoted~$\htlcRedeemPathOne$ and~$\htlcRedeemPathTwo$, and presented in Predicate~\ref{alg:original_redeem_paths}.
In~$\htlcRedeemPathOne$~(line~\ref{alg:original_redeem_path_alice}),~$\userA$ can redeem with a transaction including~$\aliceSecret$ and~$\aliceSig$, a signature with her secret key~$\aliceSecretKey$.
In~$\htlcRedeemPathTwo$~(line~\ref{alg:original_redeem_path_bob}),~$\userB$ can redeem with a transaction including~\bobSig, a signature with his secret key~$\bobSecretKey$.
This transaction can only be included in a block at least~$\timeout$ blocks after \htlc{}'s initiation, that is, block~$\blockIndex{j+\timeout}$.

As only~$\userA$ and~$\userB$ know their respective secret keys, other entities cannot redeem the contract.

\begin{algorithm}[t]
	\renewcommand{\algorithmcfname}{Predicate}
	\DontPrintSemicolon
	\SetAlgoNoLine
	\predicateTextSize
	
	\nonl Parameters: $\alicePublicKey, \bobPublicKey, \timeout, \contractSecretAlice$ \\

	\nonl $\htlc{} \left(\hashInputVar{}, \sig{} \right) \coloneqq $ \\
	
	\Numberline $\htlcPredicateMathSpace  \left( \verifyPreImage{\hashInputVar{}}{\contractSecretAlice} \land \verifySig{\sig{}}{\alicePublicKey} \right) \lor$ \label{alg:original_redeem_path_alice} \tcp*{$\htlcRedeemPathOne$}
	
	\Numberline $\htlcPredicateMathSpace \left( \verifySig{\sig{}}{\bobPublicKey} \land \verifyTimeout{\timeout} \right)$ \label{alg:original_redeem_path_bob}  \tcp*{$\htlcRedeemPathTwo$}

	\caption{\htlc{}}
	\label{alg:original_redeem_paths}
\end{algorithm}		

The intended way~$\userA$ and~$\userB$ should interact with \htlc{} is as follows.
If~$\userA$ knows the predefined preimage~$\aliceSecret$, she publishes a transaction~$\aliceTransactionHTLC$ offering a fee~$\aliceFeeHTLC > \txFee$ that redeems the \htlc{}.
She publishes this transaction right after the creation of block~$\blockIndex{j}$, that is, before the creation of block~$\blockIndex{j+1}$.
If~$\userA$ does not know the predefined preimage~$\aliceSecret$ she does not publish any transactions.

$\userB$ observes the published transactions in the mempool, watching for~$\aliceTransactionHTLC$.
If by block~$\blockIndex{j+\timeout{} - 1}$~$\userA$ did not publish~$\aliceTransactionHTLC$ then~$\userB$ publishes~$\bobTransactionHTLC$ with a fee~$\bobFeeHTLC > \txFee$, redeeming the \htlc{}.
If~$\userA$ did publish~$\aliceTransactionHTLC$ by block~$\blockIndex{j+\timeout{} - 1}$ then~$\userB$ does not publish any transactions.

		\subsection{\htlc{} Game}
\label{sec:original_htlc_game}		

\htlc{} operation gives rise to a game, denoted by~$\gameDefNameHTLC$, played among~$\userA$,~$\userB$ and the miners.
It is similar to that of the \madhtlc{} game~(\S\ref{sec:mad_htlc_game}), so we present the differences.

\paragraph{Subgames}

The game state is simply the number of blocks~($k$) created so far and state of the \htlc{}, which can be either redeemable~(\trueConst) or irredeemable~(\falseConst), so denoted~$\gameDefHTLC{k}{\trueConst/\falseConst}$. 

The game begins when one block (initiating the~\htlc{}) was created,~$\userA$ and~$\userB$ did not publish any transactions, and the \htlc{} is redeemable.
Thus, the initial, complete game is~$\gameDefHTLC{1}{\trueConst}$.

\paragraph{Actions}

$\userA$ can redeem the \htlc{} with a transaction~$\aliceTransactionHTLC$, offering~$\aliceFeeHTLC$ tokens as fee.
Note~$\aliceTransactionHTLC$ has to outbid unrelated transactions and thus has to offer a fee~$\aliceFeeHTLC > \txFee$, however, cannot offer more tokens than the redeemed ones, so~$ \aliceFeeHTLC < \depositTokens$.
$\userA$ redeems \htlc{} using the~$\htlcRedeemPathOne$ redeem path, so~$\aliceTransactionHTLC$ can be confirmed in any round. 

$\userB$ can redeem \htlc{} with a transaction~$\bobTransactionHTLC$, offering~$\bobFeeHTLC$ tokens as fee.
Similarly,~$\bobFeeHTLC$ is bounded such that~$\txFee < \bobFeeHTLC < \depositTokens$.
$\userB$ redeems \htlc{} using the~$\htlcRedeemPathTwo$ redeem path, so~$\bobTransactionHTLC$ can only be confirmed in the last round.

Any miner can include the following transactions:
an unrelated transaction in any~$\gameDefHTLC{\subgameIndistinguishable}{\subgameIndistinguishable}$ subgame;
$\aliceTransactionHTLC$ in any~$\gameDefHTLC{\subgameIndistinguishable}{\trueConst}$ subgame; and
$\bobTransactionHTLC$ in the~$\gameDefHTLC{\timeout}{\trueConst}$ subgame.

		\subsection{Bribe Attack Analysis}
        \label{sec:original_htlc_attack}		

We now show the \htlc{} prescribed strategy~(\S\ref{sec:original_htlc_definition}) is not incentive compatible. 
Specifically, we show that if~$\userA$ commits to the prescribed strategy, then~$\userB$ strictly gains by publishing a conflicting transaction, outbidding~$\userA$'s fee, thus incentivizing miners to exclude~$\userA$'s transaction and include his instead.

Let~$\userA$ publish~$\aliceTransactionHTLC$ with fee~$\aliceFeeHTLC$ in the first round, and~$\userB$ publish a transaction~$\bobTransactionHTLC$ with fee~$\bobFeeHTLC > \tfrac{\aliceFeeHTLC - \txFee}{\probMin} + \txFee$.
Focusing on miner actions, we show through a series of lemmas they are incentivized to include~$\bobTransactionHTLC$ and to exclude~$\aliceTransactionHTLC$, resulting with lower utility for~$\userA$, higher utility for~$\userB$, and a violation of the \htlcSpec{}.

First, we show miner utilities for subgames where the HTLC is irredeemable. 
Denote by~$\strategyProfile$ the best response strategy of all miners in this setting.

\begin{lemma}
	\label{lemma:htlc_include_unrelated_in_irredeemable}
	For any~$\blockCountSpecificVal \in \left[1,\timeout\right]$, the utility of miner~$i$ in subgame~$\gameDefHTLC{\blockCountSpecificVal}{\falseConst}$ is~$\utilityOfEntity{i}{\strategyProfile}{\gameDefHTLC{\blockCountSpecificVal}{\falseConst}} = \prob{i} \left(\timeout - \blockCountSpecificVal + 1 \right) \txFee$.
\end{lemma}
		
\begin{proof}
	Since \htlc{} is irredeemable, the only available action for miners is to include an unrelated transaction, yielding a reward of~$\txFee$.
	
	Consider any~$\gameDefHTLC{\blockCountSpecificVal}{\falseConst}$ subgame.
	There are~$\timeout - \blockCountSpecificVal + 1$ remaining blocks to be created, and miner~$i$ creates any of them with probability~$\prob{i}$.
	This scenario can be viewed as a series of~$\timeout - \blockCountSpecificVal + 1$ Bernoulli trials with success probability~$\prob{i}$. 
	The number of successes is therefore Binomially distributed, and the expected number of blocks miner~$i$ creates is~$\prob{i} \left(\timeout - \blockCountSpecificVal + 1\right)$.
	The reward for each block is~$\txFee$, so miner~$i$'s utility is~$\utilityOfEntity{i}{\strategyProfile}{\gameDefHTLC{\blockCountSpecificVal}{\falseConst}} = \prob{i} \left(\timeout - \blockCountSpecificVal + 1\right) \txFee$.
\end{proof}

We now consider miner utilities for~$\gameDefHTLC{\subgameIndistinguishable}{\trueConst}$ subgames, where the \htlc{} is redeemable. 
We begin with the final subgame~$\gameDefHTLC{\timeout}{\trueConst}$, creating block~$B_{j+\timeout}$.

\begin{lemma}
	\label{lemma:htlc_include_bob_in_final_subgame}
	Choosing to include~$\bobTransactionHTLC$ is a unique subgame perfect equilibrium in~$\gameDefHTLC{\timeout}{\trueConst}$, and miner~$i$'s utility when doing so  is~$\utilityOfEntity{i}{\strategyProfile}{\gameDefHTLC{\timeout}{\trueConst}} = \prob{i} \bobFeeHTLC$. 

\end{lemma}
		
\begin{proof}
	
	In the~$\gameDefHTLC{\timeout}{\trueConst}$ subgame, the miner that creates the block has three transactions to pick from: an unrelated transaction for the base fee~$\txFee$, $\aliceTransactionHTLC$ for~$\aliceFeeHTLC$, or~$\bobTransactionHTLC$ for~$\bobFeeHTLC$.
	
	As~$\bobFeeHTLC > \tfrac{\aliceFeeHTLC - \txFee}{\probMin} + \txFee$,~$ 0<\probMin<1$ and~$\aliceFeeHTLC > \txFee$, it follows that~$\bobFeeHTLC > \aliceFeeHTLC$ and~$\bobFeeHTLC > \txFee$.
	That means including~$\bobTransactionHTLC$ yields strictly greater reward than all other actions, thus being a unique subgame perfect equilibrium in this subgame.
	
	Miner~$i$ creates the block with probability~$\prob{i}$, and so her expected profit, i.e. utility, is~$\utilityOfEntity{i}{\strategyProfile}{\gameDefHTLC{\timeout}{\trueConst}} = \prob{i} \bobFeeHTLC$. 
\end{proof}

We now move on to consider any earlier ($\blockCountSpecificVal \in \left[1,\timeout - 1\right]$) subgame (Blocks~$B_{j+1}$ to~$B_{j+\timeout - 1}$ ) where the \htlc{} is redeemable.  

\begin{lemma}
	\label{lemma:htlc_include_regular_in_nonfinal_subgame}	
	For any~$\blockCountSpecificVal \in \left[1,\timeout - 1\right]$, the unique subgame perfect equilibrium is that every miner includes an unrelated transaction in~$\gameDefHTLC{\blockCountSpecificVal}{\trueConst}$, and miner~$i$'s utility when doing so is~$\utilityOfEntity{i}{\strategyProfile}{\gameDefHTLC{\blockCountSpecificVal}{\trueConst}} = \prob{i} \left( \left(\timeout - \blockCountSpecificVal\right) \txFee + \bobFeeHTLC\right)$.
\end{lemma}

To prove this lemma we show that for any~$\blockCountSpecificVal \in \left[1,\timeout - 1\right]$, including~$\userA$'s transaction in subgame~$\gameDefHTLC{\blockCountSpecificVal}{\trueConst}$ results with lower overall utility at game conclusion -- intuitively, it redeems the contract, so in the last subgame miners cannot include~$\userB$'s transaction. 
The proof is by induction on~$\blockCountSpecificVal$, and we bring it in full in Appendix~\ref{appendix:htlc_additional_proofs}.

We conclude with the main theorem regarding \htlc{} susceptibility to bribing attacks:
\begin{theorem}
	\label{theorem:htlc_accepting_bribe_is_spe}
	Alice's prescribed behavior of \htlc{} allows~$\userB$ to bribe miners at a cost of~$\tfrac{\aliceFeeHTLC - \txFee}{\probMin} + \txFee$.
\end{theorem}

\begin{proof}
	The proof follows directly from Lemma~\ref{lemma:htlc_include_bob_in_final_subgame} and Lemma~\ref{lemma:htlc_include_regular_in_nonfinal_subgame}, both showing that if~$\userA$ naively follows the prescribed strategy then subgame perfect equilibrium of the initial subgame is for all miners to place unrelated transactions until round~$\timeout{}$ and then place~$\userB$'s transaction.
\end{proof}

Note that by Theorem~\ref{theorem:htlc_accepting_bribe_is_spe}, the bribing cost required to attack \htlc{} is independent in~$\timeout$, meaning that simply increasing the timeout does contribute to \htlc{}'s security.

Of course once~$\userA$ sees an attack is taking place she can respond by increasing her fee.
In turn, this could lead to~$\userB$ increasing his fee as well, and so forth.
Instead of focusing on these bribe and counter-bribe dynamics, we conclude by showing that~$\userA$ can preemptively prevent the attack, or assure winning with a counter-bribe, by paying a high fee dependent on~$\depositTokens$,.
We note that such a high fee is in violation of the \htlcSpec{}.

\begin{corollary} 
	\label{corollary:htlc_safety_bound}
	$\userB$ cannot bribe the miners in this manner if~$\userA$'s~$\aliceTransactionHTLC$ offers at least~$\aliceFeeHTLC  > \probMin \left(\depositTokens - \txFee\right) + \txFee$.
\end{corollary}

\begin{proof}
	In order to achieve the attack,~$\userB$ ought to make placing unrelated transactions until~$\timeout{}$ and placing his transaction at~$\timeout{}$ a subgame perfect equilibrium. 
	As shown~(Theorem~\ref{theorem:htlc_accepting_bribe_is_spe}), the threshold to incentivize the smallest miner is~$\bobFeeHTLC > \tfrac{\aliceFeeHTLC - \txFee}{\probMin} + \txFee$. 
	Recall the fee~$\bobFeeHTLC$ of the bribing transaction~$\bobTransactionHTLC$ is upper bounded by the \htlc{} tokens~$\depositTokens$.
	Therefore, to achieve the attack it must hold that~$\depositTokens > \tfrac{\aliceFeeHTLC - \txFee}{\probMin} + \txFee$. 
	By choosing~$\aliceFeeHTLC > \probMin \left(\depositTokens - \txFee\right) + \txFee$,~$\userA$ can prevent~$\userB$ from paying a fee adhering to the bounds. 
\end{proof}

\paragraph*{Myopic Miners} 
This bribery attack variant relies on all miners being rational, hence considering their utility at game conclusion instead of myopically optimizing for the next block.
If a portion of the miners are myopic and any of them gets to create a block during the first~$\timeout-1$ rounds, that miner would include~$\userA$'s transaction and~$\userB$'s bribery attempt would have failed.

In such scenarios the attack succeeds only with a certain probability~-- only if a myopic miner does not create a block in the first~$\timeout-1$ rounds.
The success probability therefore decreases exponentially in~$\timeout$. 
Hence, to incentivize miners to support the attack,~$\userB$ has to increase his offered bribe exponentially in~\timeout.


The analysis relies on assumptions on the mining power distribution, and is outside the scope of this work.
Notably, for the simpler case when all other miners are myopic, miner~$i$ is incentivized to support the attack only when it is her dominant strategy, matching the upper bound of Winzer et al.~\cite{winzer2019temporary}. 

		\subsection{Non-Myopic Bribery-Accepting Miner Implementation}
		\label{sec:htlc_bribe_implementation}

Aside from the Bitcoin-MEV infrastructure~(\S\ref{sec:rational_miner_implementation}), we also implemented a simpler Bitcoin Core patch supporting the mentioned bribe attack on \htlc{}.

When the patched client receives transactions with an unexpired timeout (\emph{waiting} transactions) it stores them in a data structure instead of discarding them. 
When creating a new block, the client first checks if any of the timeouts have elapsed, and if so, moves the relevant transactions to the mempool.
When receiving conflicting transactions, instead of accepting the first and discarding the second, it accepts the transaction that offers a higher fee.
In case of a conflict with a waiting transaction, it chooses based on the condition described in Theorem~\ref{theorem:htlc_accepting_bribe_is_spe}.

The simplicity of this patch (150 LoC, no external modules) demonstrates that miners can trivially achieve non-myopic transaction selection optimization. 
		
		\subsection{Real-World Numbers}
		\label{sec:original_htlc_examples}		

We conclude this section by presenting three examples of \htlc{} being used in running systems, and show the substantial costs to make them resistant against bribery attacks. 

Table~\ref{tab:htlc_examples} presents for each example the \htlc{} tokens~$\depositTokens$, the base fee~$\txFee$, and the ratio of required tokens for bribery resistance (Theorem~\ref{theorem:htlc_accepting_bribe_is_spe}) and the base fee~$ \tfrac{\probMin\left(\depositTokens - \txFee\right) + \txFee}{\txFee}$. 
To estimate the base fee we conservatively take the actual paid fee, which is an upper bound. 
We conservatively estimate~$\probMin=0.01$~\cite{blockchain2020bitcoinPools}; miners with lower mining power are less likely due to economy-of-scale~\cite{arnosti2018bitcoin}.


The first example is of a Bitcoin Lightning channel~\cite{lightning2020topCapacityChannel,lightning2020topCapacityClosingTx}, where the required fee to secure the contract against a bribery is~1.34e4 times the actual fee.
Plugging in \$10K as the average Bitcoin price at the time~\cite{coindesk2020bitcoinPrice}, we get that an attack requires about a \$2 bribe for a payoff of over \$25K.
Note this is just an arbitrary example, and there are plenty of such low-fee, high-capacity channels, in all a few dollars bribe is sufficient to yield tens of thousands of dollars as reward~\cite{lightning2020topCapacityChannel}.
The second example is of a Litecoin atomic swap~\cite{litecoin2017atomicSwap}, requiring 436 times higher fee to be secured against bribes.
The last two examples are the two sides of a BTC-ETH atomic swap conducted by Liquality~\cite{liquality2020site}, requiring more than 300X and 480X fees to be secure, respectively.

\begin{table}[t]
	\scriptsize
	\caption{\htlc{} bribe resistance cost examples.}
	\negspace
	\negspace	
	\begin{center}
		\begin{tabular}{|  >{\centering\arraybackslash} m{0.33\linewidth} |  >{\centering\arraybackslash} m{0.09\linewidth} |  >{\centering\arraybackslash} m{0.1\linewidth} | >{\centering\arraybackslash} m{0.29\linewidth} |} 
			\hline
			
			\textbf{Name} &~$\depositTokens$ & $\txFee$ & $\dfrac{\probMin\left(\depositTokens - \txFee\right) + \txFee}{\txFee}$   \\[11pt]
			\hline

			Lightning channel $\left(\text{BTC}\right)$~\cite{lightning2020topCapacityChannel,lightning2020topCapacityClosingTx} & 2.684 & 2.22e-6 & 1.34e4  \\ \hline

			Litecoin atomic swap $\left(\text{LTC}\right)$~\cite{litecoin2017atomicSwap} & 1.337 & 3.14e-4 & 435.7  \\ \hline

								
			Liquality~\cite{liquality2020site} atomic swap $\left(\text{ETH}\right)$~\cite{liquality2020atomicSwapEth} & 12 & 0.0004 & 301 \\ \hline
			
			Liquality~\cite{liquality2020site} atomic swap $\left(\text{BTC}\right)$~\cite{liquality2020atomicSwapBtc} & 0.278 & 5.76e-6 & 483.63 \\ \hline
														
		\end{tabular}
	\label{tab:htlc_examples}
	\end{center}
\end{table}	


	\section{Future Directions}
	\label{sec:future_directions}

We briefly present two future research directions. 
First, we discuss attacks and mitigations in a weaker model, where either~$\userA$ or~$\userB$ have significant mining power~(\S\ref{sec:mad_htlc_colluding_miners}).
Then, we discuss how using \madhtlc{} can reduce latency in systems utilizing \htlcSpec{}~(\S\ref{sec:latency_reduction})

		\subsubsection{Mining~$\userA$ or~$\userB$}
		\label{sec:mad_htlc_colluding_miners}

As in previous work~\cite{winzer2019temporary,mccorry2018smart}, the security analysis of \madhtlc{} assumes that~$\userA$ and~$\userB$ have no mining capabilities and do not collude with any miner.
Indeed, acquiring mining capabilities (or forming collusion agreements) requires a significant investment, substantially higher than necessary for a simple bribe.
Removing this assumption extends the game space considerably, and brings in timing and probability considerations that are outside the scope of this work.
Nevertheless, we briefly present the issue and a potential low-overhead modification that disincentivizes such attacks.

\userA with mining capabilities that knows~$\aliceSecret$ can stall until the timeout elapses and \userB publishes~$\bobSecret$, and then redeem both~$\depositContract$ (using either~$\mhRedeemPathOne$ or~$\mhRedeemPathThree$) and~$\collateralContract$ (using~$\mhRedeemPathFive$).
This requires \userA to create the block right after the timeout elapses, otherwise another miner would include $\userB$'s transactions.
The potential profit is the \collateralTokens tokens, whose number is in the order of a transaction fee.

\userB with mining capabilities can redeem~$\depositContract$ (using~$\mhRedeemPathThree$) if he knows~$\aliceSecret$.
This requires \userB to create the first block after the~$\madhtlc$ initiation, otherwise another miner would include $\userA$'s transaction.
The potential damage for this case is similar to the \htlc{} bribery~(Winzer et al.~\cite{winzer2019temporary} and~\S\ref{sec:original_htlc_attack}), and note that any miner will be able to redeem~$\collateralContract$ once the timeout elapses.

Both variants require the miner to reveal~$\aliceSecret$ and~$\bobSecret$ by creating a block at a specific height, meaning they only succeed with some probability.
As such, their profitability depends on the relative mining size of the miner, the deposit and collateral amounts, and the transaction fees.

Nevertheless, these are vulnerabilities of \madhtlc{}, and we propose the following countermeasure: Instead of having a single \depositContract{} and a single \collateralContract{}, have multiple of each, all with the same~$\contractSecretAlice$ and~$\contractSecretBob$, but each with a different timeout~$\timeout$, and split~$\depositTokens$ and~$\collateralTokens$ among them.

As one of the timeouts elapse, if the miner attacks then she loses her advantage, as once she exposes~\aliceSecret and~\bobSecret, any miner can compete for the remaining contracts.
Therefore, this mechanism diminishes the attack profitability.

This adjustment's overhead is only due to the fees for creating and redeeming more contracts.
However, those can be small, independent of the secured amount.

		\subsubsection{Latency Reduction}
		\label{sec:latency_reduction}

Systems utilizing \htlcSpec{} must set the timeout parameter~$\timeout$, facing a trade-off.
Too short timeouts result in a security risk~-- $\userB$ might get the tokens unjustly because~$\userA$'s transaction was not yet confirmed. 
Too long timeouts imply an opportunity cost due to the unavailability of the locked coins, and increase susceptibility to various attacks~\cite{tikhomirov2020quantitative,mizrahi2020congestion,malavolta2019anonymous}. 

\madhtlc{} can allow for significantly reduced timeouts compared to \htlc{}, since instead of waiting for confirmation, it now suffices to consider transaction publication. 
The analysis depends on mempool and congestion properties that are outside the scope of this work. 


	\section{Conclusion}

We introduce a novel approach of utilizing miner's rationality to secure smart contracts, and use it to design \madhtlc{}, a contract implementing \htlcSpec{}. 
We show using the UC framework and with game-theoretic analysis that \madhtlc{} is secure.
We also show the prevalent \htlc{} is vulnerable to cheap bribery attacks in a wider variety of systems, and qualitatively tighten the known cost bound in presence of rational miners. 
We demonstrate the efficacy of our approach by implementing and executing \madhtlc{} on Bitcoin and Ethereum. 
We also demonstrate the practicality of implementing a rational miner by patching the standard Bitcoin client. 

Both the attack against \htlc{} and the secure alternative \madhtlc{} have direct impact on a variety of contracts using the \htlcSpec{} design pattern.
As miners' incentives to act rationally increase, those systems will become vulnerable and can directly adopt \madhtlc{} as a plug-in alternative.  


	\section{Acknowledgments}

We thank the anonymous reviewers, Sebastian Faust, and our shepherd Dominique Schröder for their valuable feedback and guidance.
This research was supported by the Israel Science Foundation (grant No. 1641/18), an IC3 research grant, the US-Israel Binational Science Foundation (BSF), and the Technion Hiroshi Fujiwara cyber-security research center.


\bibliographystyle{./IEEEtran}
\bibliography{./IEEEabrv,./btc}

\appendices
\renewcommand\thefigure{\arabic{figure}}    

	\section{Security Proof}
\label{appendix:sec_proof}

This section is dedicated to proving Lemma~\ref{lemma:valid_transactions}, and is organized as follows.

We focus on the scenario where~$\userB$ draws~$\aliceSecret$ and~$\bobSecret$, generates~$\contractSecretAlice = \hashFunctionWithInput{\aliceSecret}$ and~$\contractSecretBob=\hashFunctionWithInput{\bobSecret}$, and potentially shares~$\aliceSecret$ with~$\userA$ during the execution.
The alternative scenario where~$\userA$ generates~$\aliceSecret$ is similar, but nuanced differently, hence we omit its formalization in favor of concise presentation.

We prove security under static corruptions.
Additionally, we describe some of the protocols and ideal functionalities to operate according to~\emph{phases}, in which they except to receive certain messages.
Honest parties and ideal functionalities both ignore unknown messages and messages they expect to receive at other phases than their current one.

We begin by presenting the model for the UC framework, formalizing the relaxed \madhtlc{} predicate as a function, and presenting~$\idealFunctionalityHash$ and~$\idealFunctionalityMempool$ global ideal functionalities~(\S\ref{appendix:uc_model}).
We follow by formalizing~$\protocolForUC$~(\S\ref{appendix:uc_protocol}), detailing~$\idealFunctionalityMadhtlc$ and proving it satisfies the properties described in Lemma~\ref{lemma:valid_transactions}~(\S\ref{appendix:uc_ideal_func}), and conclude by showing~$\protocolForUC$ UC-realizes~$\idealFunctionalityMadhtlc$ by presenting suitable simulators~(\S\ref{appendix:indistinguishability_proof}).

\subsection{Model and Definitions} 
\label{appendix:uc_model}		
\subsubsection{Communication Model}
\label{appendix:uc_communication_model}

We assume a synchronous network where the participants have secure, FIFO channels among themselves, and the adversary can observe a message being sent over but not its content.
Communication with ideal functionalities is instantaneous, and leaks information explicitly detailed by the ideal functionality.

The environment can instruct ideal functionalities in the hybrid world to delay messages up to the synchrony bound.
However, we consider a simulator that can delay these messages in a similar manner in the ideal world.
Both of these can be realized by specific instructions through the~\emph{influence port} to the functionalities~\cite{dziembowski2018fairswap}, which we omit for simplicity.

\subsubsection{Hash Function~$\hashFunction$}
\label{appendix:uc_hash_func}

For the security analysis we model~$\hashFunction$ as a global random oracle ideal functionality~$\idealFunctionalityHash$~\cite{canetti2007universally,canetti2014practical,camenisch2018wonderful}.
Functionality~$\idealFunctionalityHash$ takes queries of variable length and responds with values of length~$\secParameter$ such that responses for new queries are chosen uniformly at random, and that additional queries of the same value always result with the same response\footnote{Our construction does not require marking queries as illegitimate~\cite{canetti2014practical} or programmability~\cite{camenisch2018wonderful,dziembowski2018fairswap}.}.
We present~$\idealFunctionalityHash$ in Functionality~\ref{functionality:hash}.

\begin{functionalityEnv} 
	
	Ideal functionality~$\idealFunctionalityHash$ is a global random oracle, receiving queries from any party or ideal functionality.
	It maps values~$\hashVarQuery \in \rangeGeneral$ to~$\hashVarResponse \in \rangeSecParameter$, and internally store mappings as an initially empty set~$\hashPrivatehSet$.
	
	\headline{Query} 
	
	Upon receiving~$\left(\id,\hashVarQuery\right)$: 
	If~$\exists \hashVarResponse : \left(\hashVarQuery,\hashVarResponse\right)\in \hashPrivatehSet$ then return~$\hashVarResponse$. 
	Otherwise, draw~$\hashVarResponse \getsRandom \rangeSecParameter$, add~$\left(\hashVarQuery,\hashVarResponse\right)$ to~$\hashPrivatehSet$, and return~$\hashVarResponse$.
	
	\caption{$\idealFunctionalityHash$ in the hybrid world.}
	\label{functionality:hash}
\end{functionalityEnv} 

\subsubsection{Mempool and Blockchain Projection Functionality~$\idealFunctionalityMempool$}
\label{appendix:uc_mempool}

We now present the MBP functionality~$\idealFunctionalityMempool$ (Functionality~\ref{functionality:mh_mempool}), detailing the setup, initiation and redeeming of a single contract in session id~$\id$.

For the setup, it allows~$\userA$ and~$\userB$ to agree on the contract parameters. 
This represents the fact both parties should be aware of the contract, and possibly one or both of them should sign the setup transaction, according to the use-case (e.g.,~\cite{herlihy2018atomic,malavolta2019anonymous,van2019specification,miraz2019atomic,zie2019extending,wagner2019dispute,poon2016bitcoin,blockstream2020lightningImp,lightningLabs2020lightningImp,acinq2020lightningImp,raiden2020raidenImp,omg2020blockchainDesign,maxwell2016firstZKCP,campanelli2017zero,banasik2016efficient,fuchsbauer2019wi,bursuc2019contingent,moser2016bitcoin, mccorry2018preventing,bishop2019vaults,zamyatin2019xclaim}).
Note that the analysis thus disregards other contract instances that exclude either of the parties.
The result of the setup is an initiating transaction, which either~$\userA$ or~$\userB$ can publish in the mempool. 

Then, a miner~$\userO$ can confirm the transaction, thus initiating the contract.
From that point onward, any party~$\anyParty \in \left\{\userA,\userB,\userO\right\}$ can publish transactions trying to redeem the initiated contract, which are evaluated based on the provided input data and the contract predicate.

We emphasize that the setup, publication and redeeming of a contract by parties~$\userA$ and~$\userB$ all correspond to exogenous events, dictated by the system using the contract as a building block. 

$\idealFunctionalityMempool$ internally maintains two hash digests~$\biInternalImages{1},\biInternalImages{2}$ of the contract, both are initially~$\bot$.
It also maintains an indicator~$\biPublished$ representing if the initiation transaction is published, and 
an indicator~$\biInit$ if it is confirmed.
Finally, it also holds $\biInternalPreimages{1}$ and $\biInternalPreimages{2}$ indicators representing if the preimages of~$\hashOutputVar{1}$ and $\hashOutputVar{2}$ were revealed, respectively.
Initial indicator values are all~$0$.

$\idealFunctionalityMempool$ accepts the following queries.
First,~$\userB$ can set the~$\biInternalImages{1}$ and~$\biInternalImages{2}$ by providing these values.
Then,~$\userA$ can acknowledge the provided digests, resulting with~$\idealFunctionalityMempool$ creating an initiating transaction~$\biTx$, and sending it to both~$\userA$ and~$\userB$.
Following that, it accepts queries from either~$\userA$ or~$\userB$ to publish~$\biTx$.
Later on,~$\userO$ can confirm the initiating transaction, initiating the contract.
Finally, any party~$\anyParty \in \left\{\userA, \userB, \userO\right\}$ can publish a transaction trying to redeem the initiated contract using a redeem path~$\redeemPathIndex$.
The redeeming attempt is evaluated by~$\publishingFunctionName$, considering the provided redeem path~$\redeemPathIndex$, the provided preimages, the redeeming party, and the values of~$\biInternalPreimages{1}$ and $\biInternalPreimages{2}$.
The latter two are also updated according to the provided preimages.

\begin{functionalityEnv} 
	
	Ideal functionality~$\idealFunctionalityMempool$ in the hybrid world represents the mempool projection of the contract, including its setup, initiation and redeeming transaction publication for session id~$\id$.
	It interacts with parties~$\userA, \userB, \userO$, the adversary~$\adversary$, and functionality~$\idealFunctionalityHash$.
	
	It maintains the following inner variable with the following default values:~$\biInternalImages{1} \gets \bot,\biInternalImages{2} \gets \bot, \biTx \gets \bot, \biPublished \gets 0, \biInit \gets 0 ,\biInternalPreimages{1} \gets 0$ and $\biInternalPreimages{2} \gets 0$.

	It accepts queries of the following types:
	
	\begin{itemize}

		\item Upon receiving~$\left(\mhSetupB,\id,\hashOutputVar{1},\hashOutputVar{2}\right)$ from~$\userB$ when~$\biInternalImages{1} = \bot$ and~$\biInternalImages{2} = \bot$, set~$\biInternalImages{1} = \hashOutputVar{1}$ and~$\biInternalImages{2} = \hashOutputVar{2}$, and leak~$\left(\mhSetupB,\id,\hashOutputVar{1},\hashOutputVar{2}\right)$ to~$\adversary$.
		
		\item Upon receiving~$\left(\mhSetupA,\id\right)$ from~$\userA$ when~$\biInternalImages{1} \ne \bot$ and~$\biInternalImages{2} \ne \bot$,
		create a transaction initiating the contract with parameters~$\biInternalImages{1}$ and~$\biInternalImages{2}$, store it as~$\biTx$, leak~$\left(\mhSetupA,\id\right)$ to~$\adversary$, and send~$\biTx$ to~$\userA$ and~$\userB$.
		
		\item Upon receiving~$\left(\mhPublish,\id,\biTx\right)$ from any party~$\anyParty \in \left\{\userA,\userB,\userO\right\}$ when~$\biPublished = 0$, set~$\biPublished \gets 1$, leak~$\left(\mhPublish,\id,\biTx\right)$ to~$\adversary$ and send~$\left(\mhPublish,\id,\biTx\right)$ to all parties.

		\item Upon receiving~$\left(\mhInit,\id,\biTx\right)$ from~$\userO$ when~$\biPublished = 1$, 
		set~$\biInit \gets 1$, leak~$\left(\mhInit,\id,\biTx\right)$ to~$\adversary$ and send~$\left(\mhInit,\id,\biTx\right)$ to all parties.

		\item Upon receiving~$\left(\mhTransaction,\id,\biTx,\hashInputVar{1},\hashInputVar{2},\redeemPathIndex\right)$ from any party~$\anyParty \in \left\{\userA, \userB, \userO\right\}$ such that~$\redeemPathIndex\in \left\{\mhRedeemPathOne,\mhRedeemPathTwo,\mhRedeemPathThree,\mhRedeemPathFour,\mhRedeemPathFive\right\}$, set~$\biInternalPreimages{1} \gets \biInternalPreimages{1} \lor \left(\left( \hashInputVar{1} \ne \bot \right) \land \left( \invokeFunctionality{\idealFunctionalityHash}{\id,\hashInputVar{1}} = \contractSecretAlice \right)\right)$ and~$\biInternalPreimages{2} \gets \biInternalPreimages{2} \lor \left(\left( \hashInputVar{2} \ne \bot \right) \land \left( \invokeFunctionality{\idealFunctionalityHash}{\id,\hashInputVar{2}} = \contractSecretBob\right)\right)$, denote~$\biResult \gets \publishingFunctionName{}\left(\redeemPathIndex,\anyParty,\biInternalPreimages{1},\biInternalPreimages{2}\right)$, leak~$\left(\mhTransaction,\id,\hashInputVar{1},\hashInputVar{2},\hashOutputVar{1},\hashOutputVar{2},\redeemPathIndex,\anyParty,\biInternalPreimages{1},\biInternalPreimages{2},\biResult\right)$ to~$\adversary$, and send~$\biResult$ to~$\anyParty$.
		
	\end{itemize}
	
	\caption{$\idealFunctionalityMempool$ in the hybrid world.}
	\label{functionality:mh_mempool}
\end{functionalityEnv} 

\subsection{Protocol~$\protocolForUC$}
\label{appendix:uc_protocol}

We now present protocol~$\protocolForUC$ (Protocol~\ref{protocol:madhtlc}), the formalization of~Protocol~\ref{protocol:madhtlc_presentation} with~$\idealFunctionalityHash$ and~$\idealFunctionalityMempool$.
$\protocolForUC$ is run by~$\userA$ and~$\userB$, but also by~$\userO$, representing any system miner.
We do not include any other non-miner entities as part of this formalization as all their possible actions are covered by~$\userO$.
The participants running the protocol proceed in phases, corresponding to the setup, initiation and publication of redeeming transactions.

Recall~$\madhtlc{}$ is used as a building-block in more elaborate constructions by~$\userA$ and~$\userB$, and~$\userA$ and~$\userB$ execute steps in the protocol due to exogenous events.
For example, for payment channels,~$\userA$ and~$\userB$ setup an initiating transaction to create a new channel state,~$\userB$ shares the first preimage with~$\userA$ to revoke the state, either~$\userA$ or~$\userB$ publish that transaction to unilaterally close the channel, and either~$\userA$ or~$\userB$ redeem it to retrieve tokens after the channel is closed.
Similarly,~$\userO$ can include transactions only when she gets to create a block, which corresponds to the random process of mining.
Finally, any party can publish transactions trying to redeem an initiated contract at will~-- providing any input it chooses, and using any redeem path.

To capture all of these we model the protocol steps as invocations by environment~$\funcEnv$.
Note that different invocation orders yield different results, and parties are not bound to terminate at all.
Specifically, correct executions may not include some of the protocol steps.
For example,~$\userB$ shares~$\aliceSecret$ to revoke a previous channel state, which occurs only if~$\userB$ and~$\userA$ agree on advancing to a next channel state.
Similarly, publication of the initiating transaction corresponds to one of the parties unilaterally closing the channel, which often enough is avoided by the parties agreeing on closing the channel with an alternative transaction.
In that case the contract does not make it on the blockchain.

Protocol~$\protocolForUC$ (Protocol~\ref{protocol:madhtlc}) operates as follows.
Each party~$\anyParty \in \left\{\userA,\userB,\userO\right\}$ maintains a state variable~$\protocolState{\anyParty}$, registering its view on the protocol~\emph{phase}.
Parties act based on their current~$\protocolState{\anyParty}$ value. 

First, in the~$\stateSetup$ phase,~$\userB$ draws~$\aliceSecret \getsRandom \rangeSecParameter, \bobSecret \getsRandom \rangeSecParameter$, to be used as the secret preimages.
He also uses~$\idealFunctionalityHash$ to derive~$\contractSecretAlice \gets \invokeFunctionality{\idealFunctionalityHash}{\id,\aliceSecret},\contractSecretBob \gets \invokeFunctionality{\idealFunctionalityHash}{\id,\bobSecret}$, thus obtaining all the relevant contract parameters.
He then sends~$\left(\mhSetupB,\id,\contractSecretAlice,\contractSecretBob\right)$ to~$\idealFunctionalityMempool$.
When he receives~$\simpleTxName$ from~$\idealFunctionalityMempool$ it stores it, and proceeds to the~$\stateInit$ phase.

$\userA$ waits to receive~$\left(\stateSetup,\id,\contractSecretAlice,\contractSecretBob\right)$ from~$\idealFunctionalityMempool$ and, once she does, she stores~$\contractSecretAlice$ and~$\contractSecretBob$.
Then, upon receiving input from~$\funcEnv$, she sends~$\left(\mhSetupA,\id\right)$ to~$\idealFunctionalityMempool$, waits to receives~$\simpleTxName$ from~$\idealFunctionalityMempool$, stores it, and proceeds to the next~$\stateInit$ phase.
Both parties have now agreed on the setup transaction. 

In the~$\stateInit$ phase,~$\userB$ can share~$\hashInputVar{}$ with~$\userA$, who accepts~$\hashInputVar{}$ only if~$\contractSecretAlice = \invokeFunctionality{\idealFunctionalityHash}{\id,\hashInputVar{}}$.
Additionally, either~$\userA$ or~$\userB$ can publish an initiating by sending~$\left(\mhPublish,\id,\simpleTxName\right)$ to~$\idealFunctionalityMempool$.
$\idealFunctionalityMempool$ notifies all parties of this invocation, causing~$\userA$ and~$\userB$ to proceed to the next~$\stateRedeeming$ phase. 
The transaction is now in the mempool. 

$\userO$ awaits to see the published initiating transaction, and then invokes~$\idealFunctionalityMempool$ to initiate it.
She then proceed to the next~$\stateRedeeming$ phase. 
The transaction is now confirmed, i.e., on the blockchain. 

In the~$\stateRedeeming$ phase, any party can publish transactions attempting to redeem the initiated contract using the various redeem paths (using~$\idealFunctionalityMempool$).

\begin{protocolEnv} 
	
	Protocol~$\protocolForUC$ for session id~$\id$ run by~$\userA, \userB, \userO$, describes the contract setup, initiation, and publication of redeeming transactions.
	The participants interact with~$\idealFunctionalityHash$ and~$\idealFunctionalityMempool$ global ideal functionalities.
	Each party has a local state variable~$\protocolState{}$ where initially~$\protocolState{\userA} = \protocolState{\userB} = \stateSetup$ and~$\protocolState{\userO} = \stateInit$.
	
	\headline{For party~$\anyParty \in \left\{\userA,\userB,\userO\right\}$ when $\protocolState{\anyParty} = \stateSetup$}
	
	\textbf{$\userB$:} 
	Upon receiving input~$\left(\mhSetupB,\id\right)$ from~$\funcEnv$, draw~$\aliceSecret \getsRandom \rangeSecParameter, \bobSecret \getsRandom \rangeSecParameter$, and set~$\contractSecretAlice \gets \invokeFunctionality{\idealFunctionalityHash}{\id,\aliceSecret},\contractSecretBob \gets \invokeFunctionality{\idealFunctionalityHash}{\id,\bobSecret}$.
	Send~$\left(\mhSetupB,\id,\contractSecretAlice,\contractSecretBob\right)$ to~$\idealFunctionalityMempool$ and wait until receiving~$\simpleTxName$ back from~$\idealFunctionalityMempool$.
	Then, store~$\simpleTxName$ and set~$\protocolState{\userB} \gets \stateInit$.

	\textbf{$\userA$:} 
	\begin{itemize}
		
		\item Upon receiving~$\left(\mhSetupB,\id,\contractSecretAlice,\contractSecretBob\right)$ from~$\idealFunctionalityMempool$, store~$\contractSecretAlice,\contractSecretBob$.
		
		\item Upon receiving input~$\left(\mhSetupA,\id\right)$ from~$\funcEnv$ after previously receiving~$\left(\mhSetupB,\id,\contractSecretAlice,\contractSecretBob\right)$, send~$\left(\mhSetupA,\id\right)$ to~$\idealFunctionalityMempool$.
		Wait to receive~$\simpleTxName$ from~$\idealFunctionalityMempool$, store it, and set~$\protocolState{\userA} \gets \stateInit$.
		
	\end{itemize}

	\headline{For party~$\anyParty \in \left\{\userA,\userB,\userO\right\}$ when $\protocolState{\anyParty} = \stateInit$}	
	
	\textbf{$\userB$:} 
	Upon receiving input~$\left(\mhSharePreimage,\id\right)$ from~$\funcEnv$, send~$\left(\mhSharePreimage,\id,\aliceSecret\right)$ to~$\userA$.

	\textbf{$\userA$:} 
	Upon receiving~$\left(\mhSharePreimage,\id,\hashInputVar{}\right)$ from~$\userB$ such that~$\contractSecretAlice = \invokeFunctionality{\idealFunctionalityHash}{\id,\hashInputVar{}}$, set~$\aliceSecret \gets \hashInputVar{}$.

	\textbf{$\anyParty \in \left\{\userA, \userB\right\}$:} 
	
	\begin{itemize}
		
		\item Upon receiving input~$\left(\mhPublish,\id\right)$ from~$\funcEnv$, send~$\left(\mhPublish,\id,\simpleTxName\right)$ to~$\idealFunctionalityMempool$.
		
		\item Upon receiving~$\left(\mhInit,\id,\simpleTxName\right)$ from~$\idealFunctionalityMempool$, set~$\protocolState{\anyParty} \gets \stateRedeeming$.
		
	\end{itemize}

	\textbf{$\userO$:} 
	
	\begin{itemize}
		
		\item Upon receiving~$\left(\mhPublish,\id,\simpleTxName\right)$ from~$\idealFunctionalityMempool$, store~$\simpleTxName$, and set~$\protocolReceivedPublishedParams \gets 1$.
		
		\item Upon receiving input~$\left(\mhInit,\id\right)$ from~$\funcEnv$ when $\protocolReceivedPublishedParams = 1$, send~$\left(\mhInit,\id,\simpleTxName\right)$ to~$\idealFunctionalityMempool$, and set~$\protocolState{\userO} \gets \stateRedeeming$.
		
	\end{itemize}
	
	\headline{For party~$\anyParty \in \left\{\userA,\userB,\userO\right\}$ when $\protocolState{\anyParty} = \stateRedeeming$}

	\textbf{$\anyParty \in \left\{\userA, \userB, \userO\right\}$:} 
	Upon receiving input~$\left(\mhTransaction,\id,\redeemPathIndex\right)$ from~$\funcEnv$ such that~$\redeemPathIndex\in \left\{\mhRedeemPathOne,\mhRedeemPathTwo,\mhRedeemPathThree,\mhRedeemPathFour,\mhRedeemPathFive\right\}$:
	
	\begin{enumerate}
		
		\item If~$\redeemPathIndex\in \left\{\mhRedeemPathOne,\mhRedeemPathThree,\mhRedeemPathFive\right\}$ and $\aliceSecret$ was previously stored, set~$\protocolPreimagePassedValue{1} \gets \aliceSecret$, and~$\bot$ otherwise.
		
		\item If~$\redeemPathIndex\in \left\{\mhRedeemPathTwo,\mhRedeemPathThree,\mhRedeemPathFive\right\}$ and $\bobSecret$ was previously stored, set~$\protocolPreimagePassedValue{2} \gets \bobSecret$, and~$\bot$ otherwise.
		
		\item 
		Send~$\left(\mhTransaction,\id,\biTx,\protocolPreimagePassedValue{1},\protocolPreimagePassedValue{2},\redeemPathIndex\right)$ to~$\idealFunctionalityMempool$, and output the result.
		
	\end{enumerate}

	\caption{$\protocolForUC$ in the hybrid world.}
	\label{protocol:madhtlc}
\end{protocolEnv} 

\subsection{Indistinguishability Proof}
\label{appendix:indistinguishability_proof}

Therefore, to prove these properties also hold for~$\protocolForUC$ in the hybrid world, and thus proving Lemma~\ref{lemma:valid_transactions} itself, we only need to show that Protocol~$\protocolForUC$ UC-realizes~$\idealFunctionalityMadhtlc$.
For that, we need to show indistinguishability between runs in the hybrid and in the ideal worlds.

\begin{lemma}
	\label{lemma:indistinguishability}
	For any PPT environment~$\funcEnv$ and any PPT adversary~$\adversary$ that corrupts any subset of~$\left\{\userA,\userB,\userO\right\}$, there exists a PPT simulator~$\simulator$, such that an execution of~$\protocolForUC$ in the hybrid world with~$\adversary$ is computationally indistinguishable from an execution of~$\idealFunctionalityMadhtlc$ in the ideal world with~$\simulator$.
\end{lemma}

We prove by explaining how to construct the simulator~$\simulator$.
The gist of the approach is as follows. 
We note that in the hybrid world, both the honest and the corrupted parties interact with each other and with the ideal functionalities~$\idealFunctionalityHash$ and~$\idealFunctionalityMempool$.

However, in the ideal world the honest parties interact only with~$\idealFunctionalityMadhtlc$.
So,~$\simulator$ has two main responsibilities.
First, it has to portray honest party interactions towards the corrupted ones~-- this can be achieved by learning of the former actions through the leakage of~$\idealFunctionalityMadhtlc$, and then sending messages to the corrupted parties as if the honest parties were running protocol~$\protocolForUC$.

Additionally, it has to portray corrupted party messages towards the honest ones.
$\simulator$ receives the corrupted party messages sent towards the ideal functionalities and honest parties, and achieves this effect by passing inputs to~$\idealFunctionalityMadhtlc$ as the corrupted party, or through~$\idealFunctionalityMadhtlc$'s influence port.

\begin{proof}
	We detail how to design a simulator~$\simulator$ for any subset of corrupted parties. 
	$\simulator$ internally simulates~$\idealFunctionalityHash$ and~$\idealFunctionalityMempool$ to be used as detailed below. 
	We explaining how~$\simulator$ responds to the various possible events, and show why the response results in a run that is indistinguishable from that of a hybrid-world. 
	
	\paragraph*{Communication leakage} 
	In the hybrid world, the environment~$\funcEnv$ learns through~$\adversary$ about any message sent among parties, but not its content.
	Note that honest parties send messages after receiving inputs from~$\funcEnv$ according to~$\protocolForUC$, and corrupted parties according to~$\adversary$.
	Moreover,~$\adversary$ receives all messages sent to corrupted parties, and can leak them to~$\funcEnv$.
	
	In the ideal world honest parties do not send messages at all, and instead just forward their inputs to~$\idealFunctionalityMadhtlc$.
	So, to maintain indistinguishability,~$\simulator$ needs to create the same leakage, and in case the destination is a corrupted party, also actually send the message.
	However, this is easily achieved as honest parties forward their inputs to~$\idealFunctionalityMadhtlc$, which leaks these inputs to~$\simulator$.
	So,~$\simulator$ leaks as if the matching hybrid-world message was sent, and in case it is destined for a corrupted party, also sends that message.
	
	It remains to show the messages themselves are indistinguishable.

	\paragraph*{State indistinguishability} 
	In the hybrid world, all parties, honest and corrupted, communicate with~$\idealFunctionalityHash$ and~$\idealFunctionalityMempool$, thus affecting the inner states of these functionalities.
	However, in the ideal world, only the corrupted parties communicate with (the internally-simulated by~$\simulator$)~$\idealFunctionalityHash$ and~$\idealFunctionalityMempool$, while the honest parties communicate only with~$\idealFunctionalityMadhtlc$.
	The respective inner states determine the output values, and the aforementioned message content, hence~$\simulator$ has to keep in sync the states of internally-simulated~$\idealFunctionalityMempool$ and the ideal functionality~$\idealFunctionalityMadhtlc$.

	The first task of~$\simulator$ is to set~$\simContractSecretAlice$,~$\simContractSecretBob$ and~$\simTx$, representing the hash s and the initiating transaction that exist in the hybrid world.
	When~$\simulator$ operates based on leaked inputs of honest parties, it considers them with respect to ~$\simContractSecretAlice$,~$\simContractSecretBob$ and~$\simTx$.
	
	$\simulator$ sets~$\simContractSecretAlice$ and~$\simContractSecretBob$ differently, based on whether~$\userB$ is honest or corrupted.
	
	If~$\userB$ is honest, when it receives input~$\left(\mhSetupB,\id\right)$ from~$\funcEnv$
	in the hybrid world then according to~$\protocolForUC$ she draws~$\aliceSecret \getsRandom \rangeSecParameter, \bobSecret \getsRandom \rangeSecParameter$, and sets~$\contractSecretAlice \gets \invokeFunctionality{\idealFunctionalityHash}{\id,\aliceSecret},\contractSecretBob \gets \invokeFunctionality{\idealFunctionalityHash}{\id,\bobSecret}$.
	She additionally sends~$\left(\mhSetupB,\id,\contractSecretAlice,\contractSecretBob\right)$ to~$\idealFunctionalityMempool$.
	
	In the ideal world~$\left(\mhSetupB,\id\right)$ is passed to~$\idealFunctionalityMadhtlc$, which then leaks~$\left(\mhSetupB,\id\right)$ to~$\simulator$.
	So,~$\simulator$ internally-simulates~$\userB$ by drawing~$\simAliceSecret \getsRandom \rangeSecParameter, \simBobSecret \getsRandom \rangeSecParameter$, setting~$\simContractSecretAlice \gets \invokeFunctionality{\idealFunctionalityHash}{\id,\simAliceSecret},\simContractSecretBob \gets \invokeFunctionality{\idealFunctionalityHash}{\id,\simBobSecret}$ (through the internally simulated~$\idealFunctionalityHash$), internally-simulate sending~$\left(\mhSetupB,\id,\simContractSecretAlice,\simContractSecretBob\right)$ to~$\idealFunctionalityMempool$, and leaking similarly to~$\adversary$.
	
	Note that~$\simAliceSecret$ and~$\simBobSecret$ are indistinguishable from~$\aliceSecret$ and~$\bobSecret$ as they are drawn from the same distribution.
	Similarly,~$\simContractSecretAlice$ and~$\simContractSecretBob$ are indistinguishable from~$\contractSecretAlice$ and~$\contractSecretBob$.
	
	If~$\userB$ is corrupted, then~$\simulator$ sets~$\simContractSecretAlice \gets \contractSecretAlice$ and~$\simContractSecretBob \gets \contractSecretBob$ when it receives~$\left(\mhSetupB,\id,\contractSecretAlice,\contractSecretBob\right)$ sent towards~$\idealFunctionalityMempool$.
	Note that in this case,~$\simContractSecretAlice=\contractSecretAlice$ and~$\simContractSecretBob=\contractSecretBob$, but~$\simAliceSecret$ and~$\simBobSecret$ remain unset.
	
	The transaction~$\simTx$ is set by internally-simulating~$\idealFunctionalityMempool$ on inputs~$\simContractSecretAlice$ and~$\simContractSecretBob$.

	As stated, from this point onward~$\simulator$ keeps in sync values of (the internally simulated)~$\idealFunctionalityMempool$ and~$\idealFunctionalityMadhtlc$.
	We show what the respective variables are, and how~$\simulator$ manages to keep them in sync.
	
	We begin with~$\simContractSecretAlice,\simContractSecretBob$ representing in the ideal world~$\contractSecretAlice,\contractSecretBob$ set in the hybrid world.
	We now list possible steps the environment can take, show their effects, and how~$\simulator$ can maintain the required indistinguishability.
	
	\begin{itemize}
		
		\item \textit{$\userA$ setup:} In the hybrid world~$\userA$ can send~$\left(\mhSetupA,\id\right)$ to~$\idealFunctionalityMempool$, either because she is honest and received input~$\left(\mhSetupA,\id\right)$, or because she is corrupted and is instructed to do so by the adversary.
		
		In the ideal world, an input~$\left(\mhSetupA,\id\right)$ sets~$ \mhFSetupA \gets 1$ and leaks~$\left(\mhSetupB,\id\right)$ to~$\simulator$, and a message from corrupted~$\userA$ is received by~$\simulator$.
		
		In the former case,~$\simulator$ internally-simulates~$\userA$ sending~$\left(\mhSetupA,\id\right)$ to (the internally-simulated)~$\idealFunctionalityMempool$.
		
		In the latter case,~$\simulator$ applies the message~$\left(\mhSetupA,\id\right)$ on~$\idealFunctionalityMempool$, and passes input~$\left(\mhSetupA,\id\right)$ as~$\userA$ to~$\idealFunctionalityMadhtlc$.
		
		Either way, both set~$ \mhFSetupA \gets 1$ and result with~$\idealFunctionalityMempool$ simulating the sending of~$\simpleTxName$ to~$\userA$ and~$\userB$ (and also actually sending the message if they are corrupted).
		
		
		\item \textit{$\userB$ sharing preimage:} In the hybrid world~$\userB$ can send~$\left(\mhSharePreimage,\id,\hashInputVar{}\right)$ to~$\userA$, either because she is honest and received input~$\left(\mhSharePreimage,\id\right)$, or because she is corrupted and is instructed to do so by the adversary.
		
		In the ideal world, an input~$\left(\mhSharePreimage,\id\right)$ sets~$ \mhFSentPreimageA \gets 1$, and a message from corrupted~$\userB$ is received by~$\simulator$.
		
		In the former case,~$\simulator$ sends~$\simAliceSecret$ as~$\userB$ to~$\userA$ (or simply leaks that a message was sent if~$\userA$ is honest).
		Note that since~$\userB$ is honest,~$\simulator$ picked~$\simAliceSecret$ and therefore can send it.
		
		In the latter case,~$\simulator$ checks if~$\simContractSecretAlice = \invokeFunctionality{\idealFunctionalityHash}{\hashInputVar{}}$ (just like a honest~$\userA$ would in the hybrid world), and if so, stores~$\simAliceSecret \gets \hashInputVar{}$ and passes~$\left(\mhSharePreimage,\id\right)$ as an honest~$\userB$ to~$\idealFunctionalityMadhtlc$.
		
		Either way, both set~$ \mhFPublished \gets 1$ and result with~$\userA$ having~$\hashInputVar{}$ such that~$\simContractSecretAlice = \invokeFunctionality{\idealFunctionalityHash}{\hashInputVar{}}$.
		
		\item \textit{$\userA$ or~$\userB$ publishing:} In the hybrid world either~$\userA$ or~$\userB$ can send~$\left(\mhPublish,\id,\simpleTxName\right)$ to~$\idealFunctionalityMempool$, either because they are honest and received input~$\left(\mhPublish,\id\right)$, or because they are corrupted and are instructed to do so by the adversary.
		
		In the ideal world, an input~$\left(\mhPublish,\id\right)$ sets~$ \mhFPublished \gets 1$, and a message from corrupted~$\userA$ or~$\userB$ is received by~$\simulator$.
		In the former case,~$\simulator$ internally simulates sending of~$\left(\mhPublish,\id,\simTx\right)$ to~$\idealFunctionalityMempool$.
		In the latter case,~$\simulator$ internally simulates~$\idealFunctionalityMempool$, verifies that~$\simpleTxName =\simTx$, and if so, passes~$\left(\mhPublish,\id\right)$ as an honest~$\userA$ or~$\userB$ to~$\idealFunctionalityMadhtlc$.
		
		Either way, it follows that~$ \mhFPublished \gets 1$ and that~$\biPublished \gets 1$.

		\item \textit{$\userO$ initiating:} In the hybrid world~$\userO$ can send~$\left(\mhPublish,\id,\simpleTxName\right)$ to~$\idealFunctionalityMempool$, either because she is honest and receives input~$\left(\mhInit,\id\right)$ (after already setting~$\protocolReceivedPublishedParams = 1$), or because she is corrupted and is instructed to do so by the adversary.
		
		In the ideal world, an input~$\left(\mhInit,\id\right)$ sets~$\mhFInit \gets 1$, and a message from corrupted~$\userO$ is received by~$\simulator$.
		In the former case,~$\simulator$ internally simulates sending of~$\left(\mhInit,\id,\simTx\right)$ to~$\idealFunctionalityMempool$.
		In the latter case, internally simulates~$\idealFunctionalityMempool$, verifies that~$\simpleTxName =\simTx$, and if so, passes~$\left(\mhInit,\id\right)$ as an honest~$\userO$ to~$\idealFunctionalityMadhtlc$.
		Either way, it follows that~$\mhFInit \gets 1$ and~$\biInit \gets 1$.

		\item \textit{Any party~$\anyParty$ redeeming:} In the hybrid world any party~$\anyParty \in \left\{\userA,\userB\userO\right\}$ can send~$\left(\mhTransaction,\id,\simpleTxName,\hashInputVar{1},\hashInputVar{2},\redeemPathIndex\right)$ to~$\idealFunctionalityMempool$, either because she is honest and receives input~$\left(\mhTransaction,\id,\redeemPathIndex\right)$, or because she is corrupted and instructed to do so by the adversary.
		
		In the ideal world, an input~$\left(\mhTransaction,\id,\redeemPathIndex\right)$ updates~$\mhFInternalPreimages{1}$ and~$\mhFInternalPreimages{2}$ based on their current values,~$\redeemPathIndex$, the party~$\anyParty$ and~$\mhSharePreimage$, and a message from corrupted~$\anyParty$ is received by~$\simulator$.
		In the former case,~$\simulator$ internally simulates sending~$\left(\mhTransaction,\id,\simTx,\hashInputVar{1},\hashInputVar{2}\redeemPathIndex\right)$ to~$\idealFunctionalityMempool$, where~$\hashInputVar{1}$ and~$\hashInputVar{2}$ are set as~$\protocolPreimagePassedValue{1}$ and~$\protocolPreimagePassedValue{2}$ in~$\protocolForUC$.
		Specifically:
		\begin{itemize}
			
			\item For honest~$\userB$,~$\hashInputVar{1} \gets \simAliceSecret$ and~$\hashInputVar{2} \gets \simBobSecret$.
			Recall that for honest~$\userB$,~$\simulator$ picks~$\simAliceSecret$ and~$\simBobSecret$, hence can set~$\hashInputVar{1}$ and~$\hashInputVar{2}$ accordingly.
			
			\item For an honest~$\userA$,~$\hashInputVar{1} \gets \simAliceSecret$ if~$\userA$ was previously sent~$\simAliceSecret$ (when~$\mhSharePreimage = 1$) and~$\hashInputVar{1} \gets \bot$.
			
			\item For honest~$\userO$,~$\hashInputVar{1} \gets \bot$ and~$\hashInputVar{2} \gets \bot$.
			
		\end{itemize}
		This applies the effect of honest party publishing a transaction on the internally-simulated~$\idealFunctionalityMempool$. 
		In the latter case,~$\simulator$ starts by internally simulating~$\idealFunctionalityMempool$ on the received message~$\left(\mhTransaction,\id,\simpleTxName,\hashInputVar{1},\hashInputVar{2},\redeemPathIndex\right)$, updating~$\biInternalPreimages{1}$ and~$\biInternalPreimages{2}$, and deriving~$\biResult$ to be later sent to~$\anyParty$.
		Then~$\simulator$ checks if~$\simContractSecretAlice = \invokeFunctionality{\idealFunctionalityHash}{\id,\hashInputVar{1}}$ and~$\simContractSecretBob = \invokeFunctionality{\idealFunctionalityHash}{\id,\hashInputVar{2}}$.
		If~$\simContractSecretAlice = \invokeFunctionality{\idealFunctionalityHash}{\id,\hashInputVar{1}}$ then~$\simulator$ instructs~$\idealFunctionalityMadhtlc$ to set~$\mhFInternalPreimages{1} \gets \mhFInternalPreimages{1} \lor \biInternalPreimages{1}$ through the designated influence port. 
		Similarly, if~$\simContractSecretBob = \invokeFunctionality{\idealFunctionalityHash}{\id,\hashInputVar{2}}$ then~$\simulator$ instructs~$\idealFunctionalityMadhtlc$ to set~$\mhFInternalPreimages{2} \gets \mhFInternalPreimages{2} \lor \biInternalPreimages{2}$ through the designated influence port. 
		
		Given the functions by which~$\idealFunctionalityMadhtlc$ and~$\idealFunctionalityMempool$ update their respective~$ \mhFInternalPreimages{i}$ and~$\biInternalPreimages{i}$ for~$i \in \left\{1,2\right\}$ variables, it follows that either way~$ \mhFInternalPreimages{i} = \biInternalPreimages{i}$.

	\end{itemize}

	\paragraph*{Output indistinguishability} 
	It remains to show that the outputs in both worlds are indistinguishable.
	However, given that for~$\mhFInternalPreimages{1} = \biInternalPreimages{1}$ and~$\mhFInternalPreimages{2} = \biInternalPreimages{2}$, and that the outputs are calculated based on either~$\biResult \gets \publishingFunctionName{}\left(\redeemPathIndex,\anyParty,\biInternalPreimages{1},\biInternalPreimages{2}\right)$ or~$\mhFResult \gets \publishingFunctionName{}\left(\redeemPathIndex,\anyParty,\mhFInternalPreimages{1},\mhFInternalPreimages{2}\right)$, then it trivially follows they are the same.
	Note that in the ideal world, corrupted parties receive their output from the internally-simulated~$\idealFunctionalityMempool$, and honest parties from~$\idealFunctionalityMadhtlc$.

	To conclude,~$\simulator$ internally simulates~$\idealFunctionalityHash$ and~$\idealFunctionalityMempool$, while interacting with~$\idealFunctionalityMadhtlc$.
	It syncs the state of the single contract honest parties interact with  in~$\idealFunctionalityMempool$ and~$\idealFunctionalityMadhtlc$, making the ideal world indistinguishable from the hybrid one.
	This proves the existence of a required~$\simulator$, thus completing the proof.
\end{proof}

For example, we present~$\simulator$ instances for specific~$\adversary$ corruption sets. 
We note that when~$\simulator$ internally invokes~$\idealFunctionalityHash$ and~$\idealFunctionalityMempool$, these invocations include leaking to~$\funcEnv$ through the~$\adversary$ as the functionality does, but also leaking as if all the messages specified by functionality were sent.

\paragraph*{Honest~$\userA$,~$\userB$ and~$\userO$}
This~$\simulator$ (Simulator~\ref{sim:a_honest_b_honest_o_honest}) instance is straight-forward: No party is corrupted, so~$\simulator$ only receives messages from~$\idealFunctionalityMadhtlc$.
Therefore, its main responsibility is to generate an indistinguishable leakage towards~$\funcEnv$, which it does by picking preimages in an indistinguishable manner on its own, and then internally-simulating the hybrid-world ideal functionalities~$\idealFunctionalityHash$ and~$\idealFunctionalityMempool$.

\begin{simulatorEnv}
	
	$\simulator$ for honest~$\userA$,~$\userB$ and~$\userO$ operates in the ideal world, receives leaked messages from~$\idealFunctionalityMadhtlc$.
	It internally simulates~$\idealFunctionalityHash$ and~$\idealFunctionalityMempool$, and can leak to~$\funcEnv$ through~$\adversary$.

	It internally holds the following variables set with the following initial values:~$\simAliceSecret \gets \bot, \simBobSecret \gets \bot,  \simContractSecretAlice \gets \bot, \simContractSecretBob \gets \bot, \simTx \gets \bot, \simSentPreimageA \gets 0$.
	It operates based on receiving leaked messages from~$\idealFunctionalityMadhtlc$:
	
	\begin{itemize}
		
		\item Upon receiving leaked~$\left(\mhSetupB,\id\right)$ from~$\idealFunctionalityMadhtlc$, draw~$\simAliceSecret \getsRandom \rangeSecParameter, \simBobSecret \getsRandom \rangeSecParameter$ and set~$\simContractSecretAlice \gets \invokeFunctionality{\idealFunctionalityHash}{\id,\simAliceSecret},\simContractSecretBob \gets \invokeFunctionality{\idealFunctionalityHash}{\id,\simBobSecret}$.
		Then, internally simulate sending~$\left(\mhSetupB,\id,\simContractSecretAlice,\simContractSecretBob\right)$ to~$\idealFunctionalityMempool$, set~$\simTx$ as returned from~$\idealFunctionalityMempool$, and leak all according messages to~$\funcEnv$ through~$\adversary$.
		
		\item Upon receiving leaked~$\left(\mhSetupA,\id\right)$ from~$\idealFunctionalityMadhtlc$, internally simulate sending~$\left(\mhSetupA,\id\right)$ to~$\idealFunctionalityMempool$, and leak all according messages to~$\funcEnv$ through~$\adversary$.
		
		\item Upon receiving leaked~$\left(\mhSharePreimage,\id\right)$ from~$\idealFunctionalityMadhtlc$, set~$\simSentPreimageA \gets 1$.
		Leak to~$\funcEnv$ through~$\adversary$ that~$\left(\mhSharePreimage,\id,\simAliceSecret\right)$ is sent from~$\userB$ to~$\userA$.
		
		\item Upon receiving leaked~$\left(\mhPublish,\id,\anyParty\right)$ from~$\idealFunctionalityMadhtlc$, internally simulate sending~$\left(\mhPublish,\id,\simTx\right)$ as~$\anyParty$ to~$\idealFunctionalityMempool$, and leak all according messages to~$\funcEnv$ through~$\adversary$.
		
		\item Upon receiving leaked~$\left(\mhInit,\id\right)$ from~$\idealFunctionalityMadhtlc$, internally simulate sending~$\left(\mhInit,\id,\simTx\right)$ as~$\userO$ to~$\idealFunctionalityMempool$, and leak all according messages to~$\funcEnv$ through~$\adversary$.

		\item Upon receiving leaked~$\left(\mhTransaction,\id,\redeemPathIndex,\anyParty\right)$ from~$\idealFunctionalityMadhtlc$:
		\begin{enumerate}
			
			\item If~$\redeemPathIndex\in \left\{\mhRedeemPathOne,\mhRedeemPathThree,\mhRedeemPathFive\right\}$ and $\left(\anyParty = \userB\right) \lor \left(\anyParty = \userA \land \simSentPreimageA = 1\right) $, let~$\simPreimagePassedValue{1} \gets \simAliceSecret$, and~$\bot$ otherwise.
			
			\item If~$\redeemPathIndex\in \left\{\mhRedeemPathTwo,\mhRedeemPathThree,\mhRedeemPathFive\right\}$ and $\anyParty = \userB$, let~$\simPreimagePassedValue{2} \gets \simBobSecret$, and~$\bot$ otherwise.
			
			\item 
			Internally simulate sending~$\left(\mhTransaction,\id,\simPreimagePassedValue{1},\simPreimagePassedValue{2}\right)$ as~$\anyParty$ to~$\idealFunctionalityMempool$, and leak all according messages to~$\funcEnv$ through~$\adversary$.
			
		\end{enumerate}

	\end{itemize}
	
	\caption{$\simulator$ for honest~$\userA$,~$\userB$ and~$\userO$.}
	\label{sim:a_honest_b_honest_o_honest}
\end{simulatorEnv}

\paragraph*{Honest~$\userA$ and~$\userO$, corrupted~$\userB$}

This~$\simulator$ (Simulator~\ref{sim:a_honest_b_corrupted_o_honest}) interacts with a corrupted~$\userB$, meaning it sets~$\simContractSecretAlice \gets \contractSecretAlice$ and~$\simContractSecretBob \gets \contractSecretBob$ when it receives~$\left(\mhSetupB,\id,\contractSecretAlice,\contractSecretBob\right)$ sent towards~$\idealFunctionalityMempool$.

It maintains state variables corresponding to those of~$\idealFunctionalityMempool$,~$\idealFunctionalityMadhtlc$ and honest participants~$\userA$ and~$\userO$, and ignores messages from~$\userB$ that would be ignored by~$\userA$ and~$\userO$.
Specifically, it lets~$\userB$ setup his part, publish and redeem transactions only when the honest parties~$\userA$ and~$\userO$ have taken their required steps.

When honest parties try to redeem the contract,~$\simulator$ reflects that publication on the internally-simulated~$\idealFunctionalityMempool$.
When corrupted~$\userB$ tries redeeming, his attempt is evaluated based on the internally-simulated~$\idealFunctionalityMempool$, yet its effects are also applied to~$\idealFunctionalityMadhtlc$ through the influence port.
Specifically, if~$\userB$ reveals a correct preimage then this is reflected in~$\idealFunctionalityMadhtlc$.

\begin{simulatorEnv}
	
	$\simulator$ for honest~$\userA$ and~$\userO$, and corrupted~$\userB$ operates in the ideal world.
	It receives messages from~$\userB$ directed at either functionality~$\idealFunctionalityHash$ or~$\idealFunctionalityMempool$, or any honest party~$\userA$ and~$\userO$.
	Additionally, it receives leaked messages from~$\idealFunctionalityMadhtlc$.
	
	$\simulator$ can send messages to~$\userB$, pass inputs as honest~$\userB$ to~$\idealFunctionalityMadhtlc$, and affect it through its influence port.
	Finally, it can also leak messages to~$\funcEnv$ through~$\adversary$. 
	
	It internally holds the following variables set with the following initial values:~$\simAliceSecret \gets \bot, \simBobSecret \gets \bot,  \simContractSecretAlice \gets \bot, \simContractSecretBob \gets \bot, \simTx \gets \bot, \simSetupB \gets 0, \simSetupB \gets 1, \simSentPreimageA \gets 0, \simPublished \gets 0, \simInit \gets 0$.
	It operates based on receiving leaked messages from~$\idealFunctionalityMadhtlc$ and messages from~$\userB$:
	
	\begin{itemize}
		
		\item Upon receiving~$\left(\id,\hashInputVar{}\right)$ from corrupted~$\userB$ directed at~$\idealFunctionalityHash$, internally-simulate~$\idealFunctionalityHash$, and return its result to~$\userB$.		
		
		\item Upon receiving~$\left(\mhSetupB,\id,\hashOutputVar{1},\hashOutputVar{2}\right)$ from corrupted~$\userB$ directed at~$\idealFunctionalityMempool$ when~$\simSetupB = 0$, set~$\simSetupB \gets 1$, store~$\simContractSecretAlice \gets \hashOutputVar{1}$ and~$\simContractSecretBob \gets \hashOutputVar{2}$.
		Proceed to internally simulate sending~$\left(\mhSetupB,\id,\simContractSecretAlice,\simContractSecretBob\right)$ to~$\idealFunctionalityMempool$, set~$\simTx$ as returned from~$\idealFunctionalityMempool$, pass as~$\userB$ input~$\left(\mhSetupB,\id\right)$ to~$\idealFunctionalityMadhtlc$, and leak all according messages to~$\funcEnv$ through~$\adversary$.
		
		\item Upon receiving leaked~$\left(\mhSetupA,\id\right)$ from~$\idealFunctionalityMadhtlc$, set~$\simSetupA \gets 1$, internally simulate sending~$\left(\mhSetupA,\id\right)$ to~$\idealFunctionalityMempool$, and leak all according messages to~$\funcEnv$ through~$\adversary$.
		
		\item Upon receiving~$\left(\mhSharePreimage,\id,\hashInputVar{}\right)$ from corrupted~$\userB$ directed at~$\idealFunctionalityMempool$ when~$\simSetupA = 1 \land \simSentPreimageA = 0$ and such that~$\simContractSecretAlice = \invokeFunctionality{\idealFunctionalityHash}{\hashInputVar{}}$, set~$\simSentPreimageA \gets 1$ and store~$\simAliceSecret \gets \hashInputVar{}$.
		Pass as~$\userB$ input~$\left(\mhSharePreimage,\id\right)$ to~$\idealFunctionalityMadhtlc$.
		
		\item Upon receiving leaked~$\left(\mhPublish,\id,\userA\right)$ from~$\idealFunctionalityMadhtlc$, internally simulate sending~$\left(\mhPublish,\id,\simTx\right)$ as~$\userA$ to~$\idealFunctionalityMempool$, and leak all according messages to~$\funcEnv$ through~$\adversary$.
		
		\item Upon receiving~$\left(\mhPublish,\id,\simpleTxName\right)$ from corrupted~$\userB$ directed at~$\idealFunctionalityMempool$ when~$\simSetupA = 1 \land \simPublished = 0$ and such that~$\simTx = \simpleTxName$, set~$\simPublished \gets 1$.
		Pass as~$\userB$ input~$\left(\mhPublish,\id\right)$ to~$\idealFunctionalityMadhtlc$.		
		
		\item Upon receiving leaked~$\left(\mhInit,\id\right)$ from~$\idealFunctionalityMadhtlc$ when~$\simPublished = 1$, set~$\simInit \gets 1$, internally simulate sending~$\left(\mhInit,\id,\simTx\right)$ as~$\userO$ to~$\idealFunctionalityMempool$, and leak all according messages to~$\funcEnv$ through~$\adversary$.
		
		\item Upon receiving leaked~$\left(\mhTransaction,\id,\redeemPathIndex,\anyParty\right)$ from~$\idealFunctionalityMadhtlc$:
		\begin{enumerate}
			
			\item If~$\redeemPathIndex\in \left\{\mhRedeemPathOne,\mhRedeemPathThree,\mhRedeemPathFive\right\}$ and $\left(\anyParty = \userB\right) \lor \left(\anyParty = \userA \land \simSentPreimageA = 1\right) $, let~$\simPreimagePassedValue{1} \gets \simAliceSecret$, and~$\bot$ otherwise.
			
			\item If~$\redeemPathIndex\in \left\{\mhRedeemPathTwo,\mhRedeemPathThree,\mhRedeemPathFive\right\}$ and $\anyParty = \userB$, let~$\simPreimagePassedValue{2} \gets \simBobSecret$, and~$\bot$ otherwise.
			
			\item 
			Internally simulate sending~$\left(\mhTransaction,\id,\simPreimagePassedValue{1},\simPreimagePassedValue{2}\right)$ as~$\anyParty$ to~$\idealFunctionalityMempool$, and leak all according messages to~$\funcEnv$ through~$\adversary$.
			
		\end{enumerate}

		\item Upon receiving~$\left(\mhTransaction,\id,\simpleTxName,\hashInputVar{1},\hashInputVar{2}\redeemPathIndex\right)$ from corrupted~$\userB$ directed at~$\idealFunctionalityMempool$ when~$\simInit = 1$ and such that~$\simTx = \simpleTxName$:
		\begin{enumerate}
			
			\item If~$\simContractSecretAlice = \invokeFunctionality{\idealFunctionalityHash}{\hashInputVar{1}}$, send~$\left(\mhUpdate,\id,1\right)$ to~$\idealFunctionalityMadhtlc$ through the influence port.
			
			\item If~$\simContractSecretBob = \invokeFunctionality{\idealFunctionalityHash}{\hashInputVar{2}}$, send~$\left(\mhUpdate,\id,2\right)$ to~$\idealFunctionalityMadhtlc$ through the influence port.
			
			\item 
			Internally simulate sending~$\left(\mhTransaction,\id,\simpleTxName,\hashInputVar{1},\hashInputVar{2}\redeemPathIndex\right)$ as~$\userB$ to~$\idealFunctionalityMempool$, leak all according messages to~$\funcEnv$ through~$\adversary$, and return the value from~$\idealFunctionalityMempool$ to~$\userB$.
			
		\end{enumerate}

	\end{itemize}
	
	\caption{$\simulator$ for honest~$\userA$ and~$\userO$, and corrupted~$\userB$.}
	\label{sim:a_honest_b_corrupted_o_honest}
\end{simulatorEnv}

\paragraph*{Corrupted~$\userA$,~$\userB$ and~$\userO$}
This~$\simulator$ (Simulator~\ref{sim:a_corrupt_b_corrupt_o_corrupt}) instance is straight-forward: All parties are corrupted, so~$\idealFunctionalityMadhtlc$ is not invoked at all.
Therefore,~$\simulator$ simply acts as~$\idealFunctionalityHash$ and~$\idealFunctionalityMempool$ towards any message it receives from any corrupted party.

\begin{simulatorEnv}
	
	$\simulator$ for corrupted~$\userA$,~$\userB$ and~$\userO$ operates in the ideal world, receives messages from and send messages to the corrupted parties.
	It internally simulates~$\idealFunctionalityHash$ and~$\idealFunctionalityMempool$, and can leak to~$\funcEnv$ through~$\adversary$.
	
	It operates based on receiving leaked messages from~$\idealFunctionalityMadhtlc$:
	
	\begin{itemize}
		\item Upon any message from a corrupted party~$\anyParty \in \left\{\userA,\userB,\userO\right\}$, follow the functionality code according to its internal state and received message, including changing its state and sending messages.
	\end{itemize}
	
	\caption{$\simulator$ for honest~$\userA$,~$\userB$ and~$\userO$.}
	\label{sim:a_corrupt_b_corrupt_o_corrupt}
\end{simulatorEnv}

	\section{\madhtlc{} Incentive Compatibility Lemma Proofs}
	\label{appendix:additional_proofs}
	
	This section presents the proofs for Lemma~\ref{lemma:mad_proof_last_round_mutual_destruction}, Lemma~\ref{lemma:mad_proof_alice_plays_prescribed} and Lemma~\ref{lemma:mad_proof_bob_plays_prescribed}.
	We recall the lemmas for readability.
	
	\begin{lemmaNo}[\ref{lemma:mad_proof_last_round_mutual_destruction}]
		In the last round of the game, i.e. subgame~$\gameDefDHTLC{\timeout}{\subgameIndistinguishable}$, if~$\aliceTransactionMadAliceBob$ and either~$\bobTransactionMadAliceBob$ or~$\bobTransactionMadBoth$ are published then miners' best-response strategy is not to include any of~$\userA$'s or~$\userB$'s transactions in this round.
	\end{lemmaNo}

	\begin{proof}
		Since~$\userA$ and~$\userB$ published their transactions, both~$\aliceSecret$ and~$\bobSecret$ are available to all miners.
		Therefore, any miner can create a transaction redeeming \depositContract{} and \collateralContract{} herself.
		
		If \depositContract{} is irredeemable~$\left(\gameDefDHTLC{\timeout}{\falseConst}\right)$, then miners can create~$\minerTransactionMadDeposit$ and redeem~$\collateralContract$ themselves in round~$\timeout$, getting~$\collateralTokens$ tokens as reward.
		Alternatively, if~$\bobTransactionMadDeposit$ is published they can include it in a block, getting a fee of~$\bobFeeMadDeposit $ tokens.
		As~$\bobFeeMadDeposit < \collateralTokens$, including~$\bobTransactionMadDeposit$ is strictly dominated by including~$\minerTransactionMadDeposit$.
		In this case miners can also not include~$\aliceTransactionMadAliceBob$ as the \depositContract{} is irredeemable.
		
		If \depositContract{} is redeemable~$\left(\gameDefDHTLC{\timeout}{\trueConst}\right)$, miners can also create~$\minerTransactionMadBoth$, include it in a block, and get~$\depositTokens + \collateralTokens$ in reward.
		Alternatively, they can include either~$\aliceTransactionMadAliceBob$,~$\bobTransactionMadAliceBob$,~$\bobTransactionMadDeposit$ or~$\bobTransactionMadBoth$ (whichever was published). 
		However, any of these offers fees lower than~$\depositTokens + \collateralTokens$, making them strictly-dominated by including~$\minerTransactionMadBoth$.
		
		Either way, including any of~$\userA$'s or~$\userB$'s transactions results with a strictly lower reward, hence miners avoid doing so. 
	\end{proof}

	\begin{lemmaNo}[\ref{lemma:mad_proof_alice_plays_prescribed}]
		In~$\gameDefDHTLC{1}{\trueConst}$,~$\userA$ cannot increase her utility by deviating from the prescribed strategy.
	\end{lemmaNo}
	
	\begin{proof}
		
		First, if~$\userA$ does not know~$\aliceSecret$, she can take no action, hence trivially complies with the prescribed strategy.
		
		If~$\userA$ does know~$\aliceSecret$, then her possible deviations are not publishing~$\aliceTransactionMadAliceBob$ at all, or publishing it only in the last round~$\timeout$.
		
		Not publishing~$\aliceTransactionMadAliceBob$ at all is strictly dominated~--- she gets no tokens; if she instead abides by the prescribed strategy then she cannot get a lower revenue but can get more, e.g., if~$\userB$ also follows the prescribed strategy~(Lemma~\ref{lemma:mad_proof_prescribed_utilities_when_alice_knows}). 
		
		The inclusion of~$\aliceTransactionMadAliceBob$ in the last block depends on what transactions~$\userB$ publishes throughout the game (Lemma~\ref{lemma:mad_proof_last_round_mutual_destruction}).
		That is, if~$\userB$ published either~$\bobTransactionMadAliceBob$ or~$\bobTransactionMadBoth$ then miners' best-response is not to include~$\aliceTransactionMadAliceBob$, and~$\userA$ gets no tokens.
		Otherwise, miners' best response is to include the transaction that offers the highest fee, which can be either~$\aliceTransactionMadAliceBob$ or another, resulting with~$\userA$ receiving~$\depositTokens -\aliceFeeMadAliceBob$ and 0 tokens, respectively.
		
		So,~$\userA$ cannot gain, and in several scenarios strictly lose, by deviating from her prescribed strategy.
	\end{proof}

	\begin{lemmaNo}[\ref{lemma:mad_proof_bob_plays_prescribed}]
		In~$\gameDefDHTLC{1}{\trueConst}$,~$\userB$ cannot increase his utility by deviating from the prescribed strategy.
	\end{lemmaNo}
	
	\begin{proof}
		
		Consider all of~$\userB$'s possible actions.
		His potential maximal utility is from having~$\bobTransactionMadBoth$ included, which he obtains by following the prescribed strategy in the scenario where~$\userA$ does not know~$\aliceSecret$ (Lemma~\ref{lemma:mad_proof_prescribed_utilities_when_alice_doesnt_know}).
		So, he has no incentive to deviate in this case.
		
		Now, consider the case where~$\userA$ knows~$\aliceSecret$, hence according to Lemma~\ref{lemma:mad_proof_alice_plays_prescribed} publishes~$\aliceTransactionMadAliceBob$ in the first~$\timeout - 1 $ rounds.
		
		$\userB$ can publish~$\bobTransactionMadAliceBob$,~$\bobTransactionMadBoth$ and~$\bobTransactionMadDeposit$ throughout the game.
		If he publishes~$\bobTransactionMadAliceBob$ or~$\bobTransactionMadBoth$ in any round then none of his transactions are included~(Lemma~\ref{lemma:mad_proof_last_round_mutual_destruction}) and he gets no reward.
		However, if he only publishes~$\bobTransactionMadDeposit$ then by Lemma~\ref{lemma:mad_proof_prescribed_utilities_when_alice_knows} he receives~$\collateralTokens - \bobFeeMadDeposit > 0$ tokens.

		Not publishing~$\bobTransactionMadDeposit$ at all results with the minimal utility of~$0$, and an earlier publication still leads miners to include both~$\bobTransactionMadDeposit$ and~$\aliceTransactionMadAliceBob$ (cf.~\ref{lemma:mad_proof_bob_plays_prescribed}), obtaining the same utility as of the prescribed behavior.
	\end{proof}

	\section{\madhtlc{} Bitcoin and Ethereum Implementations}
	\label{appendix:bitcoin_implementation}

			\subsubsection*{Bitcoin}

\begin{table*}[t]
	\scriptsize
	\begin{tabular}{| p{0.32\linewidth} | p{0.30\linewidth} | p{0.30\linewidth} | }

		\hline
		\textbf{\depositContract{}} & \textbf{\collateralContract{}} & \textbf{\htlc{}} \\
		\hline
		\renewcommand{\arraystretch}{0.9}
		\begin{tabular}{  m{0.08\linewidth}  m{0.08\linewidth}  m{0.02\linewidth} }
			$ \OPHASH $  & & \\
			$ \contractSecretAlice $  & & \\
			$ \OPEQUAL $ & & \\
			$ \OPSWAP $ & & \\
			$ \OPHASH $	& & \\
			$ \contractSecretBob $	& & \\
			$ \OPEQUAL $ & & \\
			$ \OPIF $ & & \\
			& $ \OPIF $ & \\	
			& & $ \OPONE $ \\
			& $ \OPELSE$ & \\
			& & $ \timeout $ \\
			& & $ \OPCHECKSEQUENCEVERIFY $ \\
			& & $ \OPDROP $ \\
			& & $ \bobPublicKey $ \\
			& & $ \OPCHECKSIG $ \\
			& $ \OPENDIF $ & \\
			$ \OPELSE $	&& \\
			& $ \OPVERIFY $	& \\
			& $ \alicePublicKey $ & \\       
			& $ \OPCHECKSIG~$ & \\
			$ \OPENDIF~$ & & \\
		\end{tabular}
		&   
		\begin{tabular}{ m{0.08\linewidth}  m{0.08\linewidth} }
			$ \timeout $  &   \\ 
			$ \OPCHECKSEQUENCEVERIFY $ &   \\  
			$ \OPDROP $  &   \\  	
			$ \OPHASH $  &   \\  
			$ \contractSecretAlice $ &   \\  
			$ \OPEQUAL $ &   \\  	
			$ \OPIF $ 	&   \\  
			& $ \OPHASH $	 \\
			& $ \contractSecretBob $	\\
			& $ \OPEQUAL $ \\	
			$ \OPELSE$ & \\
			& $ \bobPublicKey $ \\
			& $ \OPCHECKSIG~$ \\ 
			$ \OPENDIF~$ & 
		\end{tabular}
		&
		\begin{tabular}{ m{0.08\linewidth}  m{0.08\linewidth} }
			$ \OPHASH $  &   \\ 
			$ \contractSecretAlice $ &   \\  
			$ \OPEQUAL $  &   \\  	
			$ \OPIF $ 	 &   \\  
			&	$ \alicePublicKey $ \\
			$ \OPELSE$	& \\
			& $ \timeout $ \\
			& $ \OPCHECKSEQUENCEVERIFY $ \\
			& $ \OPDROP $	\\
			& $ \bobPublicKey $ \\
			$ \OPENDIF~$	& \\
			$ \OPCHECKSIG~$ &
		\end{tabular} 
		\\ \hline
		
		\begin{tabular}{  c l  }
			\textbf{Redeem path}  & \textbf{Input data}  \\ 
			1 &  $ \aliceSig \;\; \OPZERO \;\; \aliceSecret $ \\  
			2 &  $ \bobSig \;\; \bobSecret \;\; \OPZERO$ \\
			3 &  $ \bobSecret \;\; \aliceSecret$
		\end{tabular}
		&
		\begin{tabular}{  c l }
			\textbf{Redeem path}  & \textbf{Input data}  \\ 
			1 &  $ \bobSig \;\; \OPZERO$ \\  
			2 &  $ \bobSecret \;\; \aliceSecret$
		\end{tabular}
		&
		\begin{tabular}{  c l }
			\textbf{Redeem path}  & \textbf{Input data}  \\ 
			1 &  $ \aliceSig \;\; \aliceSecret$ \\  
			2 &  $ \bobSig \;\; \OPZERO$
		\end{tabular} \\ \hline

	\end{tabular}
	\captionof{figure}{\depositContract{}, \collateralContract{} and \htlc{} Bitcoin Script implementations.}
	\label{fig:bitcoin_script_implementation}
	\negspace
	\negspace
\end{table*}

Fig.~\ref{fig:bitcoin_script_implementation} shows the Bitcoin Script implementation of \depositContract{}, \collateralContract{} and \htlc{}.
It also presents the required input data for each redeem path.

Script is stack-based, and to evaluate input data and a contract the latter is concatenated to the former, and then executed: constants are pushed into the stack, instructions operate on the stack.
For a successful evaluation the stack must hold exactly one element with value~$1$ after all operations are executed.

\paragraph{\depositContract{}}
The script expects either two or three data elements.
It hashes the first two and checks if they match~$\contractSecretAlice$ and~$\contractSecretBob$.

If the first matches~$\contractSecretAlice$ but the second does not match~$\contractSecretBob$ ($\mhRedeemPathOne$), then the script verifies the existence of a third data element, and that it is a signature created with~$\userA$'s secret key.

If the first does not match~$\contractSecretAlice$ but the second matches~$\contractSecretBob$ ($\mhRedeemPathTwo$), then the script verifies the existence of a third data element, and that it is a signature created with~$\userB$'s secret key.
It also verifies the timeout has elapsed.

If both the first and the second data elements match~$\contractSecretAlice$ and~$\contractSecretBob$ ($\mhRedeemPathThree$), respectively, then the script expects no third data element and evaluates successfully.

\paragraph{\collateralContract{}}
The script expects exactly two data elements.
It begins by verifying timeout has elapsed, and then hashes the first element and checks if it matches~$\contractSecretAlice$.

If not ($\mhRedeemPathFour$), the script then verifies the second data is a signature created with~$\userB$'s secret key.
Otherwise ($\mhRedeemPathFive$), the script hashes the second data element and verifies it matches~$\contractSecretBob$.

			\subsubsection*{Ethereum}

We present an Ethereum Solidity implementation of \madhtlc{} in Fig.~\ref{fig:ethereum_solidity_implementation}.
We also present \htlc{}~\cite{functionalfoundry2020htlcImplementation} implementation for comparison.

\begin{table*}[t]
	\scriptsize
	
	\begin{tabular}{| p{0.45\linewidth} | p{0.51\linewidth} |}

		\hline
		\textbf{\htlc{}} & \textbf{\madhtlc{}} \\
		\hline
		\renewcommand{\arraystretch}{0.9}
		\begin{lstlisting}[language=Java]
contract HTLC is ReentrancyGuard {
  using Address for address payable;

  // Participants in the exchange
  address sender;
  address recipient;

  // Secret hashed by sender
  bytes32 image;

  // Expiration timestamp
  uint256 expires;

  constructor(
    address _sender,
    address _recipient,
    bytes32 _image,
    uint256 _expirationTime
  ) public payable {
    // Define internal state
    sender = _sender;
    recipient = _recipient;
    image = _image;
    expires = now + _expirationTime;
  }

  function claimDepositRecipient(bytes32 _pre)
    			public nonReentrant {
    require(msg.sender == recipient);
    require(hash(_pre) == image);

    msg.sender.sendValue(address(this).balance);
  }

  function claimDepositSender() public nonReentrant {
    require(msg.sender == sender);
    require(now > expires);
    msg.sender.transfer(address(this).balance);
  }

  function hash(bytes32 _preimage)
    internal
    pure
    returns (bytes32 _image) {
      return sha256(abi.encodePacked(_preimage));
  }
}

		\end{lstlisting} 
		
		& 

		\begin{lstlisting}[language=Java]
contract MADHTLC is ReentrancyGuard {
  using Address for address payable;
  address payable sender;
  address payable recipient;
  bytes32 imageA;
  bytes32 imageB;
  uint256 expires;
  uint256 collateral;
  uint256 deposit;
  bool depositClaimed = false;
  bool collateralClaimed = false;

  constructor(
    address payable _sender,
    address payable _recipient,
    bytes32 _imageA,
    bytes32 _imageB,
    uint256 _expirationTime,
    uint256 _collateral
  ) public payable {
    require(_collateral < msg.value);
    sender = _sender;
    recipient = _recipient;
    collateral = _collateral;
    deposit = msg.value - _collateral;
    imageA = _imageA;
    imageB = _imageB;
    expires = now + _expirationTime;
  }
  function claimDepositRecipient(bytes32 _preimageA)
    public nonReentrant {
    require(msg.sender == recipient);
    require(hash(_preimageA) == imageA);
    require(!depositClaimed);
    depositClaimed = true;
    recipient.sendValue(deposit);
  }
  function claimDepositSender(bytes32 _preimageB)
    public nonReentrant {
    require(msg.sender == sender);
    require(hash(_preimageB) == imageB);
    require(now > expires);
    require(!depositClaimed);
    depositClaimed = true;
    sender.sendValue(deposit);
  }
  function claimCollateralSender() public nonReentrant {
    require(msg.sender == sender);
    require(now > expires);
    require(!collateralClaimed);
    collateralClaimed = true;
    sender.sendValue(collateral);
  }
  function claimDepositAnyone(
    bytes32 _preimageA, bytes32 _preimageB
  ) public nonReentrant {
    require(hash(_preimageA) == imageA);
    require(hash(_preimageB) == imageB);
    require(!depositClaimed);
    depositClaimed = true;
    msg.sender.sendValue(deposit);
  }
  function claimCollateralAnyone(
    bytes32 _preimageA, bytes32 _preimageB
  ) public nonReentrant {
    require(now > expires);
    require(hash(_preimageA) == imageA);
    require(hash(_preimageB) == imageB);
    require(!collateralClaimed);
    collateralClaimed = true;
    msg.sender.sendValue(collateral);
  }
  function hash(bytes32 _preimage)
    internal pure returns (bytes32 _image) {
    return sha256(abi.encodePacked(_preimage));
  }
		\end{lstlisting}
\\ \hline
	\end{tabular}
	\captionof{figure}{\htlc{} and \madhtlc{} Ethereum Solidity implementations.}
	\label{fig:ethereum_solidity_implementation}
\end{table*}

	\section{\madhtlc{} Bitcoin and Ethereum Deployment}
	\label{appendix:main_network_deployment}

Tables~\ref{tab:transaction_ids_bitcoin} and~\ref{tab:transaction_ids_ethereum} show the transaction IDs in our Bitcoin and Ethereum deployments~(\S\ref{sec:implementaiton_overheard_deployment}), respectively. 
Their details can be viewed with online block explorers. 

%

\begin{table}[t]
	\scriptsize
	\captionof{table}{Bitcoin main-net experiment transaction IDs.}
	\begin{center}
		\begin{tabular}{| >{\centering\arraybackslash}p{0.25\linewidth} | >{\centering\arraybackslash}m{0.64\linewidth} |} 
			\hline
			\textbf{Description} &  \textbf{Transaction ID}  \\
			\hline
			\hline
			Initiate \depositContract{} & \texttt{d032175260145055860296cbca8f7462
				4f30334ddf948d5da12f0c7414d80cc0} \\ \hline
			\depositContract{} path 1   & \texttt{33c957bb2f75e797d240a38504ce49a3
				aeaaceb72f8577096b4f2ff23f5b3a1e} \\ \hline
			\depositContract{} path 2   & \texttt{cd090c90afaacc0e2648834fe96f6177
				ec2f967b7e50245537afdaf0d5a80263} \\ \hline
			\depositContract{} path 3   & \texttt{505c7f1f3862b7f5c6b78f72cce5e37a
				655b946fbdc7d03526055f7ea206781a} \\ \hline \hline
			Initiate \collateralContract{}  & \texttt{ea830dba56000b3486cf1c5122fedcf8
				8169ab596536fd406b4f989e7761c1b4} \\ \hline
			\collateralContract{} path 1    & \texttt{4c06ebff8de6bb56242c75849767a633
				9e40a0442f815a2487fd9d6237c51b9f} \\ \hline
			\collateralContract{} path 2    & \texttt{68270b94ca80281e31e193dac6779d3a
				22d2799fe2afff8cef66c0ec6b420c88} \\ \hline
			
		\end{tabular}
		\label{tab:transaction_ids_bitcoin}
		\negspace
		\negspace
	\end{center}
\end{table}

\begin{table}[t]
	\scriptsize
	\captionof{table}{Ethereum main-net experiment transaction IDs.}	
	\begin{center}
		\begin{tabular}{| >{\centering\arraybackslash}p{0.25\linewidth} | >{\centering\arraybackslash}m{0.64\linewidth} |} 
			\hline
			\textbf{Description} &  \textbf{Transaction ID}  \\
			\hline
			\hline
			Initiation & \texttt{f10be5e53b9ad8a6f10d7e9b9bfbd63a
				b8737c50274885182a67e7adc3fa59c2} \\ \hline
			$\mhRedeemPathOne$   & \texttt{36e349b4fdc5385ef57a88d077837223
				b3a26b0e6afc75f90bbaf2860d9295fd} \\ \hline
			$\mhRedeemPathTwo$   & \texttt{84aa626d659b63e0554f8de1a3d6e204
				41d8d778b7e1e79d0a36ded325afedb4 } \\ \hline
			$\mhRedeemPathThree$ (ours)   & \texttt{ebdb267e8b612d59910bc2348a95eec8
				388e62dbd6d64458c982f0cdacea67d9} \\ \hline 
			$\mhRedeemPathThree$ (other) & \texttt{74e87bba99ccd7a0bd794b793f108674
				5b462390df01594ce057a430c122635a} \\ \hline 
			
		\end{tabular}
		\label{tab:transaction_ids_ethereum}
		\negspace
		\negspace
	\end{center}
\end{table}

	\section{\htlc{} Bribe Attack Analysis Proof}
	\label{appendix:htlc_additional_proofs}

We recall Lemma~\ref{lemma:htlc_include_regular_in_nonfinal_subgame} and prove it. 

\begin{lemmaNo}[\ref{lemma:htlc_include_regular_in_nonfinal_subgame}]
	For any~$\blockCountSpecificVal \in \left[1,\timeout - 1\right]$, the unique subgame perfect equilibrium is that every miner includes an unrelated transaction in~$\gameDefHTLC{\blockCountSpecificVal}{\trueConst}$, and miner~$i$'s utility when doing so is~$\utilityOfEntity{i}{\strategyProfile}{\gameDefHTLC{\blockCountSpecificVal}{\trueConst}} = \prob{i} \left( \left(\timeout - \blockCountSpecificVal\right) \txFee + \bobFeeHTLC\right)$.
\end{lemmaNo}

\begin{proof}
	Note that in~$\gameDefHTLC{\blockCountSpecificVal}{\trueConst}$ there are two actions available, either include an unrelated transaction and receive~$\txFee$ reward, or include~$\aliceTransactionHTLC$ and receive~$\aliceFeeHTLC$ reward.
	
	Consider any miner~$i$.
	Denote by~$\minersOnBoardProbability{\blockCountSpecificVal}$ the accumulated block-creation rates of miners, excluding miner~$i$, that choose to include an unrelated transaction in~$\gameDefHTLC{\blockCountSpecificVal}{\trueConst}$.
	Therefore, the accumulated probabilities of miners that choose to include~$\aliceTransactionHTLC$, excluding miner~$i$, is~$1-\minersOnBoardProbability{\blockCountSpecificVal} -\prob{i}$.

	If miner~$i$ chooses to include an unrelated transaction then either of the following occurs.
	First, with probability~$\prob{i}$ miner~$i$ gets to create a block, includes an unrelated transaction and receives a reward of~$\txFee$. 
	The subsequent subgame is~$\gameDefHTLC{\blockCountSpecificVal + 1}{\trueConst}$.
	Alternatively, with probability~$\minersOnBoardProbability{\blockCountSpecificVal}$ another miner that includes an unrelated transaction gets to create a block, miner~$i$ gets no reward and the subsequent subgame is~$\gameDefHTLC{\blockCountSpecificVal + 1}{\trueConst}$.
	Finally, with probability~$1-\minersOnBoardProbability{\blockCountSpecificVal} -\prob{i}$ another miner that includes~$\aliceTransactionHTLC$ gets to create a block, miner~$i$ gets no reward and the subsequent subgame is~$\gameDefHTLC{\blockCountSpecificVal + 1}{\falseConst}$.
	
	Therefore, miner~$i$'s utility when including an unrelated transaction in these subgames is
	\begin{equation}
	\label{eq:htlc_utility_in_redeemable_subgame_creating_regular}
	\begin{aligned}
	& \utilityOfEntity{i}{\strategyProfile}{\gameDefHTLC{\blockCountSpecificVal}{\trueConst}}  =  \\ 
	& \prob{i} \cdot \left(\txFee + \utilityOfEntity{i}{\strategyProfile}{\gameDefHTLC{\blockCountSpecificVal + 1}{\trueConst}} \right) + \\
	& \minersOnBoardProbability{\blockCountSpecificVal} \cdot \utilityOfEntity{i}{\strategyProfile}{\gameDefHTLC{\blockCountSpecificVal + 1}{\trueConst}}  + \\
	& \left(1 -\prob{i} - \minersOnBoardProbability{\blockCountSpecificVal} \right) \cdot \utilityOfEntity{i}{\strategyProfile}{\gameDefHTLC{\blockCountSpecificVal + 1}{\falseConst}} \,.
	\end{aligned}
	\end{equation}
	Similarly, if miner~$i$ chooses to include~$\aliceTransactionHTLC$ than either of the following occurs.
	First, with probability~$\prob{i}$ miner~$i$ gets to create a block, includes~$\aliceTransactionHTLC$ and receives a reward of~$\aliceTransactionHTLC$. 
	The subsequent subgame is~$\gameDefHTLC{\blockCountSpecificVal + 1}{\falseConst}$.
	Alternatively, with probability~$\minersOnBoardProbability{\blockCountSpecificVal}$ another miner that includes an unrelated transaction gets to create a block, miner~$i$ gets no reward and the subsequent subgame is~$\gameDefHTLC{\blockCountSpecificVal + 1}{\trueConst}$.
	Finally, with probability~$1-\minersOnBoardProbability{\blockCountSpecificVal} -\prob{i}$ another miner that includes~$\aliceTransactionHTLC$ gets to create a block, miner~$i$ gets no reward and the subsequent subgame is~$\gameDefHTLC{\blockCountSpecificVal + 1}{\falseConst}$.

	Therefore, miner~$i$'s utility when including~$\aliceTransactionHTLC$ in these subgames is
	\negspace
	\begin{equation}
	\label{eq:htlc_utility_in_redeemable_subgame_creating_alice}
	\begin{aligned}
	& \utilityOfEntity{i}{\strategyProfile}{\gameDefHTLC{\blockCountSpecificVal}{\trueConst}}  =  \\ 
	& \prob{i} \cdot \left(\aliceFeeHTLC + \utilityOfEntity{i}{\strategyProfile}{\gameDefHTLC{\blockCountSpecificVal + 1}{\falseConst}} \right) + \\
	& \minersOnBoardProbability{\blockCountSpecificVal} \cdot \utilityOfEntity{i}{\strategyProfile}{\gameDefHTLC{\blockCountSpecificVal + 1}{\trueConst}}  + \\
	& \left(1 -\prob{i} - \minersOnBoardProbability{\blockCountSpecificVal} \right) \cdot \utilityOfEntity{i}{\strategyProfile}{\gameDefHTLC{\blockCountSpecificVal + 1}{\falseConst}}
	\,\,.
	\end{aligned}
	\end{equation}

	To prove the lemma we need to show that for any~$\blockCountSpecificVal \in \left[1,\timeout - 1\right]$ the utility from including an unrelated transaction~(Eq.~\ref{eq:htlc_utility_in_redeemable_subgame_creating_regular}) exceeds that of including~$\aliceTransactionHTLC$~(Eq.~\ref{eq:htlc_utility_in_redeemable_subgame_creating_alice}).
	This reduces to showing that 
	\begin{equation}
	\label{eq:htlc_subgame_perfect_equilibrium_condition}
	\begin{aligned}
	&\txFee + \utilityOfEntity{i}{\strategyProfile}{\gameDefHTLC{\blockCountSpecificVal + 1}{\trueConst}} > \\
	& \aliceFeeHTLC + \utilityOfEntity{i}{\strategyProfile}{\gameDefHTLC{\blockCountSpecificVal + 1}{\falseConst}}\,\,,
	\end{aligned}
	\end{equation}
	which we do inductively.

	\paragraph{Base} 
	First, consider~$\blockCountSpecificVal=\timeout - 1$. 
	Using Lemma~\ref{lemma:htlc_include_bob_in_final_subgame} and Lemma~\ref{lemma:htlc_include_unrelated_in_irredeemable} we get the condition presented in Eq.~\ref{eq:htlc_subgame_perfect_equilibrium_condition} is~$\txFee + \prob{i} \bobFeeHTLC > \aliceFeeHTLC + \prob{i} \txFee$, 
	or alternatively, 
	\begin{equation}
	\label{eq:htlc_subgame_perfect_equilibrium_condition_for_k_2}
	\bobFeeHTLC > \tfrac{\aliceFeeHTLC - \txFee}{\prob{i}} + \txFee\,\,.
	\end{equation}	
	
	\negspace
	Since~$\probMin \le \prob{i}$ and~$\bobFeeHTLC > \tfrac{\aliceFeeHTLC - \txFee}{\probMin} + \txFee$, the condition (Eq.~\ref{eq:htlc_subgame_perfect_equilibrium_condition_for_k_2}) holds, meaning that in any subgame perfect equilibrium miner~$i$ is strictly better by including an unrelated transaction in subgame~$\gameDefHTLC{\timeout - 1}{\trueConst}$.
	
	Therefore, all miners choose to include unrelated transactions in such subgames, meaning~$\minersOnBoardProbability{\blockCountAnotherSpecificVal} = 1 -\prob{i}$ and~$1 -\prob{i} - \minersOnBoardProbability{\blockCountAnotherSpecificVal} = 0$.
	Therefore, miner~$i$'s utility (Eq.~\ref{eq:htlc_utility_in_redeemable_subgame_creating_regular}) is~$\utilityOfEntity{i}{\strategyProfile}{\gameDefHTLC{\blockCountSpecificVal}{\trueConst}}  = \prob{i} \left( \txFee + \bobFeeHTLC\right)$.

	\paragraph{Assumption} 
	Consider any~$\blockCountSpecificVal \in \left[1, \timeout - 2\right]$ and assume that the claim holds for~$\blockCountSpecificVal +1$.
	That is, the unique subgame perfect equilibrium in subsequent games~$\gameDefHTLC{\blockCountSpecificVal + 1}{\trueConst}$ is for all miners to include an unrelated transaction, and the utility of miner~$i$ when doing so  is~$\utilityOfEntity{i}{\strategyProfile}{\gameDefHTLC{\blockCountSpecificVal + 1}{\trueConst}} = \prob{i} \left( \left(\timeout - \blockCountSpecificVal\right) + \bobFeeHTLC\right)$.
	
	\paragraph{Step} 
	Using the inductive assumption and Lemma~\ref{lemma:htlc_include_unrelated_in_irredeemable} the condition of Eq.~\ref{eq:htlc_subgame_perfect_equilibrium_condition} translates to~$\txFee + \prob{i} \left( \left(\blockCountSpecificVal + 1\right) \txFee + \bobFeeHTLC\right) > \aliceFeeHTLC + \prob{i} \left(\blockCountSpecificVal + 1\right) \txFee$,
	or alternatively, 
	\begin{equation}
	\label{eq:htlc_subgame_perfect_equilibrium_condition_for_any_k}
	\bobFeeHTLC > \tfrac{\aliceFeeHTLC - \txFee}{\prob{i}} + \txFee\,\,.
	\end{equation}	
	
	\negspace
	Again, since~$\probMin \le \prob{i}$ and~$\bobFeeHTLC > \tfrac{\aliceFeeHTLC - \txFee}{\probMin} + \txFee$, the condition (Eq.~\ref{eq:htlc_subgame_perfect_equilibrium_condition_for_any_k}) holds, meaning that in the subgame perfect equilibrium miner~$i$'s strict best response is to include an unrelated transaction in subgame~$\gameDefHTLC{\blockCountSpecificVal}{\trueConst}$.
	
	Since all miners include unrelated transactions, we get~$\minersOnBoardProbability{\blockCountAnotherSpecificVal} = 1 -\prob{i}$ and~$1 -\prob{i} - \minersOnBoardProbability{\blockCountAnotherSpecificVal} = 0$.
	Therefore, miner~$i$'s utility (Eq.~\ref{eq:htlc_utility_in_redeemable_subgame_creating_regular}) is~$\utilityOfEntity{i}{\strategyProfile}{\gameDefHTLC{\blockCountSpecificVal}{\trueConst}}  = \prob{i} \left( \left( \timeout - \blockCountSpecificVal\right) \txFee + \bobFeeHTLC\right)$.
\end{proof}

\end{document}